\documentclass[]{article}

\usepackage{amsmath}

\usepackage{amsfonts, graphicx, amssymb, setspace, caption, subcaption, color}
\usepackage{multirow}
\usepackage{url}
\usepackage{xr}
\usepackage{authblk}

\usepackage[margin=1in]{geometry}

\usepackage{times}
\usepackage{bm}
\usepackage{natbib}
\usepackage{threeparttable}
\usepackage[plain,noend]{algorithm2e}
\usepackage{pdfpages}
\usepackage{amsthm}

\makeatletter
\renewcommand{\algocf@captiontext}[2]{#1\algocf@typo. \AlCapFnt{}#2} 
\def\@algocf@capt@plain{top}
\renewcommand{\algocf@makecaption}[2]{%
  \addtolength{\hsize}{\algomargin}%
  \sbox\@tempboxa{\algocf@captiontext{#1}{#2}}%
  \ifdim\wd\@tempboxa >\hsize
    \hskip .5\algomargin%
    \parbox[t]{\hsize}{\algocf@captiontext{#1}{#2}}
  \else%
    \global\@minipagefalse%
    \hbox to\hsize{\box\@tempboxa}
  \fi%
  \addtolength{\hsize}{-\algomargin}%
}
\makeatother

\def\df{\textsc{df}}
\def\snr{\textsc{snr}}
\def\mse{\textsc{mse}}

\def\RHS{\textsc{rhs}}
\def\LHS{\textsc{lhs}}

\newcommand{\bs}[1]{\boldsymbol{#1}}

\newcommand{\wt}[1]{\widetilde{#1}}
\newcommand{\wh}[1]{\widehat{#1}}

\newcommand{\e}{\varepsilon}

\newcommand{\K}{\tilde{K}}

\newcommand{\Om}{\Omega}
\newcommand{\up}{\Omega}
\newcommand{\hd}{\wh{\Delta}}
\newcommand{\de}{\delta}

\newtheorem{lemma}{Lemma}
\newtheorem{theorem}{Theorem}
\newtheorem{proposition}{Proposition}
\newtheorem{definition}{Definition}
\newtheorem{corollary}{Corollary}
\begin{document}

%
%
\title{Nonparametric regression with adaptive truncation via a convex hierarchical penalty}

\author[1]{Asad Haris}
\author[2]{Ali Shojaie}
\author[3]{Noah Simon}
\affil[1,2,3]{Department of Biostatistics, University of Washington-Seattle}

\maketitle

\begin{abstract}
	We consider the problem of nonparametric regression with a potentially large number of covariates. We propose a convex, penalized estimation framework that is particularly well-suited for high-dimensional sparse additive models and combines appealing features of finite basis representation and smoothing penalties. In the case of additive models, a finite basis representation provides a parsimonious representation for fitted functions but is not adaptive when component functions posses different levels of complexity. In contrast, a smoothing spline type penalty on the component functions is adaptive but does not provide a parsimonious representation. Our proposal simultaneously achieves parsimony and adaptivity in a computationally efficient way. We demonstrate these properties through empirical studies and show that our estimator converges at the minimax rate for functions within a hierarchical class. We further establish minimax rates for a large class of sparse additive models. We also develop an efficient algorithm that scales similarly to the lasso with the number of covariates and sample size.
\end{abstract}



\section{Introduction and motivation}
\label{sec:introduction}
Consider univariate nonparametric function estimation from observations $\{(x_i,y_i) \in \mathbb{R}^2: i = 1,\ldots, n\}$. Assume that $y_i = f(x_i) + \e_i$ ($i=1,\ldots, n$), where $\e_i$ are independent, mean zero, sub-Gaussian random variables. There are many proposals for estimating $f$, including local polynomials~\citep{stone1977consistent}, kernel smoothing~\citep{nadaraya1964estimating,watson1964smooth}, and splines~\citep{wahba1990spline}. To begin, we focus on basis expansion estimators, also known as projection estimators~\citep{chentsov1962evaluation}, which are widely used, and arguably the simplest. 

Let ${y} = (y_1,\ldots,y_n)^{\top}\in \mathbb{R}^n$ and ${x} = (x_1,\ldots, x_n)^{\top} \in \mathbb{R}^n$ be the response and covariate vectors. 
For ${v}\in \mathbb{R}^n$, let $\|{v}\|_n^2 = n^{-1}\sum_{i=1}^{n}v_i^2$ be a modified $\ell_2$-norm, referred to as the {empirical norm}.
Projection estimators are solutions to linear regression problems based on a set of basis functions $(\psi_k)_{k=1}^{\infty}$, with a truncation level $K$. 
More specifically, let ${\Psi}_K\in \mathbb{R}^{n\times K}$ be the $n \times K$ matrix with entries ${\Psi}_{K(i,k)}= \psi_k(x_i)$ ($k = 1,\ldots, K; i = 1, \ldots, n$). 
The basis expansion estimate of $f$ is given by $\wh{f} = \sum_{k = 1}^K\wh{\beta}^{\text{proj}}_k \psi_k$, where 
\begin{equation}
\label{eqn::projest}
\wh{{\beta}}^{\text{proj}} = \underset{{\beta} \in \mathbb{R}^K}{\operatorname{argmin }}\   \frac{1}{2}\left\|{y} - {\Psi}_{K}{\beta}\right\|_n^2. 
\end{equation}

To asymptotically balance bias and variance, $K \equiv K_n$ is allowed to vary with $n$. Unfortunately, choosing the    {truncation} level $K$ can be difficult in practice; it depends on the variance of $\e_i$, properties of $f$ such as smoothness, and the choice of basis functions. Usually, $K$ is chosen via split-sample validation. The challenge of tuning $K$ becomes more evident in multivariate problems, which we describe next, and is one of our main motivations.  

Multivariate additive models \citep{hastie2009elements} easily follow from projection estimators, where each ${x}_i = \left(x_{i1},\ldots,x_{ip}\right)^{\top}$ is now a $p$-vector, and the true underlying model is believed to be of the form
$y_i = \sum_{j=1}^p f_j\left(x_{ij}\right) +\e_i\ (i=1,\ldots,n).$
The components $f_j$ of this model can be estimated by using a basis expansion for each $j$ and solving
\begin{equation}\label{eqn::add}
\wh{{\beta}}^{\text{A-proj}}_1,\ldots,\wh{{\beta}}^{\text{A-proj}}_p  = \underset{ {\beta}_j \in \mathbb{R}^{K_j} }{\operatorname{argmin}}\  \frac{1}{2}\Big\|y
- \sum_{j=1}^p{\Psi}^{j}_{K_j}\beta_j\Big\|_n^2, 
\end{equation}
where ${\Psi}^{j}_{K_j}$ are $K_j$ basis functions for feature $j$ and $f_j$ is estimated as $\wh{f}_j = \sum_{k=1}^{K_j} \wh{\beta}^{\text{A-proj} }_{jk} \psi_k$. 

%

In practice, the same truncation level is used for each feature, $K_j \equiv K$, to reduce the number of tuning parameters. When $f_j$ have widely different complexities, this leads to poor estimates. This limitation, which is illustrated by the simulation in Figure~\ref{fig::smallSim}, becomes more hindering in higher dimensions, as $p$ increases. 
In high-dimensional problems, when $p \gg n$, it is often assumed that $f_j \equiv 0$ for many components. A popular choice is then to add a sparsity-inducing penalty to the basis expansion framework \citep{ravikumar2009sparse} and solve
\begin{equation}\label{eqn::spam}
\wh{\beta}^{\text{SPAM}}_1,\ldots,\wh{\beta}^{\text{SPAM}}_p  = \underset{ \beta_j \in \mathbb{R}^{K} }{\operatorname{argmin}}\  \frac{1}{2}\Big\|y - \sum_{j=1}^p  {\Psi}^{j}_{K}\beta_j\Big\|_n^2 + \lambda \sum_{j=1}^p \Big\| {\Psi}^{j}_{K} \beta_j\Big\|_n.
\end{equation}

In this manuscript, we propose a penalized estimation framework that penalizes function complexity, simultaneously selects the truncation level, and can be used to fit both univariate and multivariate additive models with or without sparsity. We present an extension for fully nonparametric multivariate settings, as well as a relaxed version, similar to the relaxed lasso \citep{meinshausen2007relaxed}, which reduces the bias. While our univariate proposal is similar to \eqref{eqn::projest}, our additive proposal data-adaptively selects the truncation level for each feature, $f_j$. This improves the prediction accuracy and provides  parsimonious estimates of ${f}_j$. We illustrate these advantages in a small simulation study with data $y_i = f_1(x_{i1})+f_2(x_{i2})+\e_i\ (i=1,\ldots,n)$ using $f_1,f_2$ shown in Figure~\ref{fig::smallSim}. We fit \eqref{eqn::add} using $K_j\equiv K$ selected to optimize mean square error, and compare it to our relaxed proposal. Tuning parameters for our method also minimizes the mean square error. The results in Figure~\ref{fig::smallSim} clearly demonstrate the superior performance of our method, which has lower mean square error while maintaining parsimony. In particular, we estimate the linear term, $f_1$, by a linear function, whereas~\eqref{eqn::add} uses an order 9 polynomial.

\begin{figure}
	\centering
	\includegraphics[scale = 0.4]{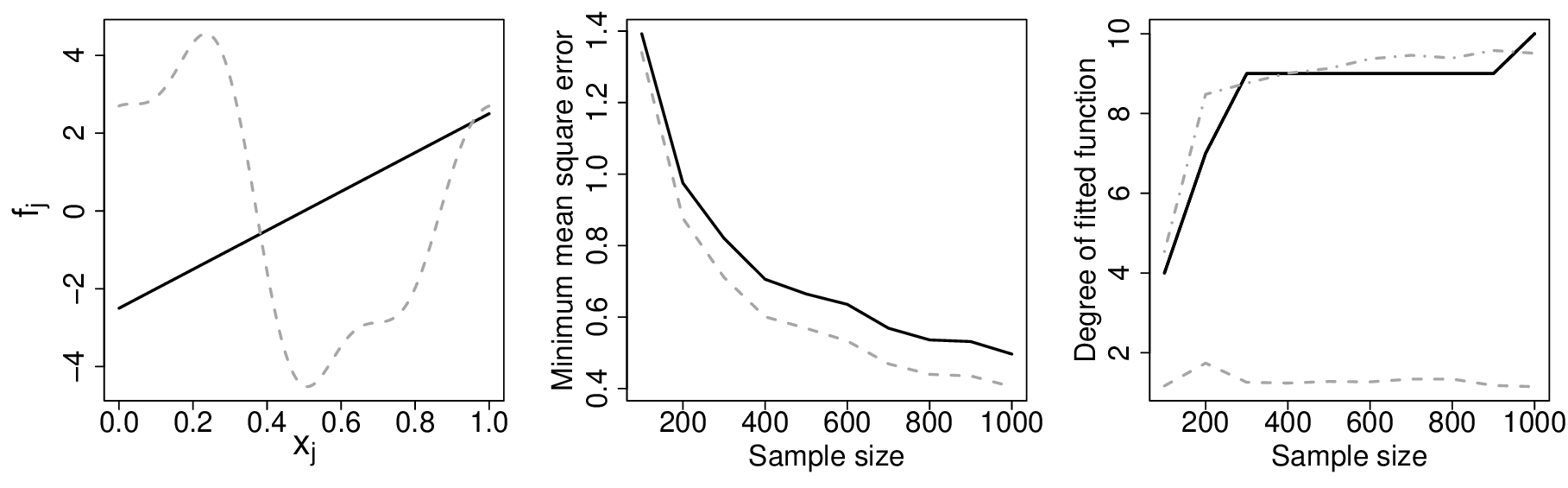}
	\caption{Left: Plots of component functions $f_1$~(\protect\includegraphics[height=0.5em]{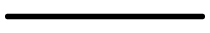}) and  $f_2$~(\protect\includegraphics[height=0.5em]{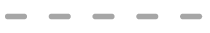}). Middle: Minimum mean square errors as functions of $n$ for our proposal~(\protect\includegraphics[height=0.5em]{NA_gray65_2.eps}) and \eqref{eqn::add}~(\protect\includegraphics[height=0.5em]{NA_gray0_1.eps}). Right: Degrees of fitted polynomial as functions of $n$: for our proposal, $\wh{f}_1$~(\protect\includegraphics[height=0.5em]{NA_gray65_2.eps}) and $\wh{f}_2$~(\protect\includegraphics[height=0.5em]{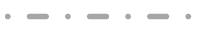}) are shown; for \eqref{eqn::add} both component functions have the same degree~(\protect\includegraphics[height=0.5em]{NA_gray0_1.eps}). All results are averaged over 100 replications.}
	\label{fig::smallSim}
\end{figure}    


In addition to adaptability and parsimony, our proposal is computationally efficient and can work with thousands of observations and features. Moreover, its estimates attain minimax optimal rates under standard smoothness assumptions, for univariate, multivariate, and sparse additive models.
The univariate estimator converges at rate $ n^{-{2m}/{(2m+1)}}$ where $m$ is the degree of smoothness; similarly, the multivariate estimator attains the rate $ n^{-{2m}/{(2m +p)}} $. For sparse additive models, under a suitable compatibility condition, our estimator converges at rate $\max\left\{ sn^{-{2m}/{(2m+1)}}, {s\log p}/{n} \right\} ,$ where $s$ is the number of non-zero $f_j$; even without the compatibility condition, consistency is achieved at the 
rate $ \max\left\{ sn^{-{m}/{(2m+1)}}, s({\log p}/{n})^{1/2} \right\}$.

\section{Methodology}
\label{sec:Methodology}
\subsection{Motivation for adaptive truncation}


Our proposal is motivated by the need to select the truncation levels in a data-driven manner.
Let us first reconsider the simple projection estimator. The bias-variance tradeoff and parsimony of estimates $\wh{f}_j$ in \eqref{eqn::add} are controlled by truncation levels $K_j$. While separately tuning $K_j$ over each component function may be feasible in low dimensions, it quickly becomes infeasible for additive models, as the optimal truncation level requires searching over a $p$-dimensional space. 

To bypass the tuning of multiple truncation levels, $K_j (j = 1, \ldots, p)$, one can instead use 
$n$ basis functions for each component, and consider a penalized version of the truncation estimator, 
\begin{equation}\label{eqn::project_pen}
\wh{{\beta}}^{\text{proj-pen}}_1,\ldots,\wh{{\beta}}^{\text{proj-pen}}_p  = \underset{ {\beta}_j \in \mathbb{R}^{n} }{\operatorname{argmin}}\  \frac{1}{2}\Big\|y
- \sum_{j=1}^p{\Psi}^{j}_{n}\beta_j\Big\|_n^2 + \lambda \sum_{j=1}^p \operatorname{max}(k: \,\beta_{jk} \neq 0), 
\end{equation}
where the truncation level for each feature is determined using the penalty $\operatorname{max}(k:\beta_{jk} \neq 0)$. The estimator in \eqref{eqn::project_pen} chooses the truncation level for each feature data-adaptively. However, the penalty in \eqref{eqn::project_pen} is non-convex, so solving \eqref{eqn::project_pen} becomes infeasible in moderate to high-dimensional problems. To mitigate this challenge, we formulate a convex problem using a novel penalty that can be seen as a convex relaxation of the penalty in \eqref{eqn::project_pen}.  

Our approach is particularly suitable for basis functions that possess a natural hierarchy, that is, when $\left(\psi_{k}\right)_{k=1}^{\infty}$ become increasingly complex for higher values of $k$, as opposed to, say, natural splines,  which rely on knot points. Examples of hierarchical basis functions include polynomial, trigonometric and wavelet basis functions; see 
the online Supplementary Material.


\subsection{The univariate proposal}
\label{sec:univariateHierBasis}
Consider again the projection estimator \eqref{eqn::projest}. As noted in Section~\ref{sec:introduction}, choosing the truncation level $K$ is key here: $K$ too small will result in large bias, while $K$ too large will over-inflate the variance. The bias and variance are balanced by taking $K = {O}\left\{ n^{{1}/{(2m+1)}} \right\}$, where $m$ relates to the smoothness of the underlying $f$, and is unknown in practice.
To circumvent this challenge, we use instead a complete basis with $K=n$ along with regularization to choose the truncation level. Our estimator is defined as 
\begin{equation}
\label{eqn::hierbasis}
\wh{\beta}^{\text{hier}} =  \underset{\beta \in \mathbb{R}^n }{\operatorname{argmin } }\   \frac{1}{2}\left\|y - \Psi_{n} \beta\right\|_n^2 + \lambda \Omega\left(\beta\right), \quad \Omega\left( {\beta}\right) = \sum_{k=1}^n w_k \left\|  {\Psi}_{k:n} {\beta}_{k:n} \right\|_n,
\end{equation}
%
with $w_k = k^m - (k-1)^m$. 
Here, ${\Psi}_{k:n}$ denotes the submatrix of $  {\Psi}_n$ containing columns $k$ to $n$, $  {\beta}_{k:n}$ is the subvector of $  {\beta}$ containing entries $k$ to $n$, and $m$ and $\lambda$ are tuning parameters. The choice of weights $w_k$ is theoretically motivated, and detailed in Section 5. Briefly, it defines a function class with desirable properties that allow us to establish convergence rates.

The hierarchical group lasso 
penalty $\Omega\left(\beta\right)$, will result in solutions $\wh{  {\beta}}^{\text{hier}}$ with {hierarchical sparsity}: that is, if $\wh{\beta}^{\text{hier}}_k = 0$ for some $k$, then $\wh{\beta}^{\text{hier}}_{k^{'}} = 0 $ for all $k^{'}>k $. For sufficiently large $\lambda$, many entries of $\wh{\beta}^{\text{hier}}$ will be $0$. For a given $\lambda$, we define the induced truncation level to be the minimal integer $K\leq n$ such that $\wh{\beta}^{\text{hier}}_{k} = 0$ for all integers $k>K$. Unlike the simple basis expansion estimator \eqref{eqn::projest}, this truncation level is data-adaptive, not prespecified. 

Equation~\eqref{eqn::hierbasis} involves two tuning parameters, $m$ and $\lambda$. The parameter $m$ is analogous to the smoothness parameter in smoothing splines \citep{wahba1990spline}, or the number of bounded derivatives used in  simple projection estimator \citep{chentsov1962evaluation}. 
In practice, using $m=2$ or $3$ gives good results; this is similar to the use of cubic smoothing splines. 
On the other hand, $\lambda$ determines the trade-off between goodness-of-fit and parsimony; 
a theoretically optimal $\lambda$-value is $\lambda \propto n^{-{m}/{(2m+1)}}$. Split-sample validation can be used to choose $\lambda$ in practice.

As with the lasso, the regularization in \eqref{eqn::hierbasis} results in bias, which can reduce the overall mean square error. To reduce this bias, we can consider the relaxed version of our estimator in \eqref{eqn::hierbasis} as the simple basis expansion estimator with $K = \|\wh{\beta}^{\text{hier}}\|_0$, the truncation level selected by \eqref{eqn::hierbasis}. 
This relaxed proposal is equivalent to using the penalty for selecting a truncation level. Therefore, in the univariate case, the relaxed estimator for a sequence of $\lambda$ values would match the simple basis expansion estimator~\eqref{eqn::projest} for a sequence of $K$ values. 
However, the advantages of our penalty become more clear in the case of multivariate additive models, discussed next. 


\subsection{The additive proposal}
\label{subsec:additiveHierBasis}
Ideally, the additive projection estimator \eqref{eqn::add} is obtained by considering a different truncation level $K_j$ for each feature. 
When $p$ is small, this can be achieved by using split-sample validation and searching over all combinations of $K_j (j = 1,\ldots,p)$; however, the number of candidate models grows exponentially in $p$ and becomes quickly unwieldy. Often, a single $K\equiv K_j$ is used in practice, which can lead to some $f_j$ estimates with too many degrees of freedom. As illustrated in Figure~\ref{fig::smallSim}, using a single truncation level can lead to poor estimates. 


Our proposal, which can be seen as a convex relaxation of \eqref{eqn::project_pen}, is designed to circumvent the above limitation of projection estimators in choosing the truncation level for models with multiple covariates. This differentiates our proposal from the projection estimator: in our framework, a single tuning parameter $\lambda$ leads to different $K_j$ for each fitted $f_j$.

Our additive framework is a direct extension of our univariate proposal~\eqref{eqn::hierbasis}. Specifically, we consider function estimates $\wh{f}_j = \sum_{k=1}^{n} \wh{\beta}^{\text{A-hier}}_{jk} \psi_k$, where 
\begin{equation}\label{eqn::addhier}
\wh{  {\beta}}^{\text{A-hier}}_1,\ldots,\wh{  {\beta}}^{\text{A-hier}}_p = \underset{  {\beta}_j \in \mathbb{R}^{n}}{\operatorname{argmin}} \ \frac{1}{2}\Big\|  {y} - \sum_{j=1}^p  {\Psi}^{j}_{n}  {\beta}_j\Big\|_n^2 + \lambda \sum_{j=1}^p \Omega_j\left(  {\beta}_j\right), 
\end{equation}
and $\Omega_j$ is the hierarchical group lasso penalty with weights $w_k = k^m - (k-1)^m$:
\begin{equation}\label{eqn::addPen}
\Omega_j\left(  {\beta}_j\right) = \sum_{k=1}^n w_k \left\|   {\Psi}^{j}_{k:n}  {\beta}_{j,k:n} \right\|_n. 
\end{equation}
%

The optimization problem~\eqref{eqn::addhier} results in $\wh{  {\beta}}_j$ estimates that are hierarchically sparse for each $j$. Specifically, for each $j$, there is some minimal $K_j$ such that $\wh{\beta}^{\text{A-hier}}_{jk} = 0$ for all integers $k > K_j$. Moreover, the major advantage of \eqref{eqn::addhier} is that the induced truncation level is feature-wise adaptive, with a different $K_j$ for each feature $j$. Additionally, as in the univariate setting, we can define a relaxed version of our estimator by fitting~\eqref{eqn::add}, where $K_j$ is now determined by \eqref{eqn::addhier}. As a result, our framework balances goodness-of-fit and parsimony for each feature individually, without requiring an exhaustive search. This is a major advantage over simple projection estimators.

The advantage of our method over simple projection estimators becomes more evident in high dimensions, when $p \gg n$.  
For instance, the popular estimator of \citeauthor{ravikumar2009sparse}, \eqref{eqn::spam}, is generally obtained by using a single truncation level, which, as noted above, can result in poor estimates. 
Similar to their proposal, our sparse additive framework encourages feature-wise sparsity using a group lasso penalty~\citep{yuan2006model}, and is defined as 
\begin{equation}\label{eqn::SparseAddhier}
\wh{  {\beta}}^{\text{S-hier}}_1,\ldots,\wh{  {\beta}}^{\text{S-hier}}_p = \underset{  {\beta}_j \in \mathbb{R}^{n}}{\operatorname{argmin}} \  \frac{1}{2}\Big\|  {y} - \sum_{j=1}^p  {\Psi}^{j}_{n}  {\beta}_j\Big\|_2 + \lambda^2 \sum_{j=1}^p \Omega_j\left(  {\beta}_j\right) + \lambda \sum_{j=1}^p \left\|  {\Psi}^{j}_{n}  {\beta}_j\right\|_n,
\end{equation}
with $\Omega_j\left(  {\beta}_j\right)$ given in~\eqref{eqn::addPen}. We can again define a relaxed version which fits~\eqref{eqn::add} with sparsity and $K_j$ selected by \eqref{eqn::SparseAddhier}. An important feature of the optimization problem~\eqref{eqn::SparseAddhier} is that the tuning parameters for the two penalty terms $\lambda$ and $\lambda^2$ are linked. This link is theoretically justified in Section~\ref{sec:TheoreticalResults}. Briefly, for an oracle $\lambda$, the choice of tuning parameters in \eqref{eqn::SparseAddhier} gives rate-optimal estimates. In practice, while this formulation gives good predictive performance in many cases, in other cases tuning sparsity and smoothness separately leads to strong predictive performance. Our numerical experiments in Section~\ref{sec:SimulationStudy} and \ref{sec:DataAnalysis} corroborate this finding.

As with $\wh{{\beta}}^{\text{SPAM}}_j$ in \eqref{eqn::spam}, for sufficiently large $\lambda$, our proposal gives a sparse solution with most $\wh{  {\beta}}^{\text{S-hier}}_j \equiv 0$. The two estimators differ, however, in their nonzero estimates: non-zero $\wh{  {\beta}}^{\text{S-hier}}_j$ are hierarchically sparse, with a data-driven feature-specific induced truncation level, whereas nonzero $\wh{  {\beta}}^{\text{SPAM}}_j$ in \eqref{eqn::spam} all have the same complexity. This additional flexibility of our methodology proves critical in high dimensions, and is achieved without paying a price in computational or sample complexity. Moreover, with the tuning parameters in \eqref{eqn::SparseAddhier}, this additional flexibility is in theory achieved with the same number of tuning parameters as \citeauthor{ravikumar2009sparse}'s method. 

\subsection{Relationship to existing methods}
The univariate framework of Section~\ref{sec:univariateHierBasis} builds upon existing penalized estimation methods. A popular choice is the smoothing spline estimator~\citep{wahba1990spline}, which sets $(\psi_k)_{k=1}^n$ to $n$ natural splines with knots at the observed covariates $x_1,\ldots,x_n$; this estimator is found by minimizing
$
\|  {y} -   {\Psi}  {\beta} \|_n^2 + 2\lambda\|C^{1/2}  {\beta}\|_n^2
$
over $\beta \in \mathbb{R}^n$, 
using $C\in \mathbb{R}^{n\times n}$ and $C_{j,k} = \int \psi_j^{\left( m/2+1/2  \right)}(t)\psi_k^{\left( m/2+1/2 \right)}(t)\, dt$ with $\psi^{(k) }$ denoting the $k$th order derivative of $\psi$. 
The smoothing spline eliminates the dependence on the truncation level and has an efficient-to-compute closed-form solution; but its estimated functions are piecewise polynomial splines of degree $m$ with $n$ knots. As a result, smoothing spline estimates are not parsimonious. 
To achieve more parsimonious estimates, \cite{mammen1997locally} use a data-driven approach to select the knots in spline functions. Their locally adaptive regression splines use the same natural spline basis and is found by minimizing 
$
\|  {y} -   {\Psi}  {\beta} \|_n^2 + 2\lambda(m!)^{-1}\|D  {\beta}\|_1,
$
over $\beta \in \mathbb{R}^n$,
using $D\in \mathbb{R}^{(n-m-1)\times n}$ with $D_{i,j} = \psi_j^{(m)}(t_i) - \psi_j^{(m)}(t_{i-1})$. This is closely related to the more computationally tractable trend filtering proposal~\citep{kim2009ell1,tibshirani2014adaptive}. 

\begin{table}
	\caption{Comparison of existing methods for sparse additive models.}
		\begin{tabular}{p{13mm}p{45mm}p{45mm}p{45mm}}
			& \hspace{20mm}  PE & \hspace{17mm} SS/RKHS & 
			 \hspace{20mm} TF\\
			$\quad $\newline 
			Scalability &
			\begin{center}
				$\checkmark$
			\end{center}
			\footnotesize
			Exact solution for univariate sub-problem; scales as $O(npK)$.
			& \begin{center}
				$\times$
			\end{center}
			\footnotesize
			No exact solution for univariate sub-problem; general convex solver scaling as $O\{(np)^3\}$.
			&\begin{center}
				$\times$
			\end{center} 
			\footnotesize
			Inefficient beyond first order TF, particularly for high-dimensional and unequally spaced covariates. \\
			$\quad $\newline Adaptability & \begin{center}
				$\times$
			\end{center}
			\footnotesize
			Smoothness controlled by basis expansion order, $K$, fixed for all component functions. 
			& \begin{center}
				\checkmark
			\end{center}
			\footnotesize
			Smoothness controlled by smoothness norm, varied for component functions. 
			& \begin{center}
				\checkmark
			\end{center}
			\footnotesize
			Smoothness controlled by smoothness norm and number of knots, $K$, varied for component functions.\\
			$\quad $ \newline Parsimony &  \begin{center}
				$\checkmark$
			\end{center}
			\footnotesize
			Component functions of order $K\ll n$ expansions. 
			& \begin{center}
				$\times$
			\end{center}
			\footnotesize
			Component functions of order $n$ expansions. 
			& \begin{center}
				$\checkmark$
			\end{center}
			\footnotesize
			Component functions have sparsity in number of knots. 
	\end{tabular}
	\begin{tablenotes}
		\small
		\item PE, projection estimators~\citep{ravikumar2009sparse,lou2016sparse}; RKHS, reproducing kernel Hilbert spaces~\citep{raskutti2012minimax,koltchinskii2010sparsity,yuan2015minimax}; SS, smoothing splines~\citep{meier2009high}; TF, trend filtering~\citep{petersen2016fused,sadhanala2017additive}.
	\end{tablenotes}
	\label{tab:TableMethods}
\end{table}

Despite their appealing properties in the univariate setting, locally adaptive regression splines and trend filtering are computationally difficult to extend to high-dimensional sparse additive models; even for a single feature, neither estimator has a closed-form solution. 
\citeauthor{ravikumar2009sparse}'s estimator \eqref{eqn::spam} overcomes this difficulty by using a fixed truncation level for all $p$ components. As mentioned earlier, its main drawback is that all nonzero components of the additive model have the same complexity. The sparse partially linear additive model of \cite{lou2016sparse} partly mitigates this by setting some of the nonzero components to linear functions using a hierarchical penalty of the form 
$
\sum_{j=1}^p\lambda_1 \|  {\beta}_j\|_2 + \lambda_2 \|  {\beta}_{j,-1}\|_2;
$
here $\beta_{j,1}$ is the coefficient of the linear term in the basis expansion, $  {\beta}_{j,-1} = (\beta_{j,2},\ldots, \beta_{j,K})^{\top} \in \mathbb{R}^{K-1}$, and $\lambda_1$ and $\lambda_2$ are tuning parameters. The first term in the penalty sets all of the coefficients for the $j$th feature to zero, whereas the second term only sets the $K-1$ coefficients corresponding to higher-order terms to zero. 

Our additive and sparse additive proposals of Section~\ref{subsec:additiveHierBasis} generalize those of \cite{ravikumar2009sparse} and \cite{lou2016sparse}.
The first becomes a special case of \eqref{eqn::SparseAddhier} if the weights in \eqref{eqn::addPen} are set to  
$
w_1 = 1 \text{ and } w_k = 0\text{ for } k >1. 
$
Similarly, with an orthogonal design matrix, $(  {\Psi}_K^j)^{\top}  {\Psi}_K^j/n = I_{K} \ (j = 1,\ldots,p)$, \citeauthor{lou2016sparse}'s method is a special case of  \eqref{eqn::SparseAddhier}
with weights in \eqref{eqn::addPen} set to 
$
w_{1} = w_{2} = 1 \text{ and } w_{k} = 0 \text{ for } k>2. 
$
Our theoretical analysis in Section~\ref{sec:additiveHierBasisConvergence} indicates that, in addition to the improved flexibility, our choice of weights~\eqref{eqn::addPen} results in optimal rates of convergence. 

{There are other proposals for estimating sparse additive models, including extensions of trend-filtering and smoothing splines for additive models. Extensions of trend filtering were either not shown to be rate optimal~\citep{petersen2016fused} or only shown to be so for low-dimensional additive models~\citep{sadhanala2017additive}. These extensions are also computationally challenging beyond first order trend filtering. Similarly, the extension of smoothing splines by \cite{meier2009high} is computationally inefficient, and not rate-optimal.} 

Using properties of reproducing kernel Hilbert spaces, some proposals offer minimax-optimal convergence rates for the prediction error over smooth additive classes~\citep{koltchinskii2010sparsity,raskutti2012minimax,yuan2015minimax} similar to our results in Section~\ref{sec:additiveHierBasisConvergence}. 
However, their estimators are given as minimizers of $(np)$-dimensional second order cone programs, for which they do not discuss efficient algorithms, at most mentioning generic convex solvers. The computation for general-purpose second order convex cone program solvers scales roughly as $(np)^3$; thus, even for moderate $p$ and $n$, these proposals become quickly intractable. We compare and contrast the strengths and weaknesses of existing proposals in Table~\ref{tab:TableMethods}. 

\section{Computational considerations and extensions}
\label{sec:extraStuff}
\subsection{Conservative basis truncation}
\label{subsec:fewBasis}
Our proposal~\eqref{eqn::hierbasis} uses a basis expansion with $n$ basis functions. In practice, for any reasonable choice of $\lambda$, $\wh{  {\beta}}$ will never have $n$ nonzero entries, and will generally have very few non-zero entries, $K_0 \ll n$. If we instead solve
\begin{equation}
\label{eqn::hierTrunc}
\wh{  {\beta}}^{\text{hier}({\K})} = \underset{  {\beta} \in \mathbb{R}^{\K} }{\operatorname{argmin}}\   \frac{1}{2}\Big\|  {y} -   {\Psi}_{\K}  {\beta}\Big\|_2 + \lambda \sum_{k=1}^{\K} w_k \Big\|   {\Psi}_{k:\K}  {\beta}_{k:\K} \Big\|_n,
\end{equation}
for $\K<n$, then provided $\K \geq K_0$, the solution will be identical to that of the original proposal~\eqref{eqn::hierbasis}. Even when not identical, so long as $\K$ is sufficiently large, $\K \gtrsim  n^{{2m}/\{(2m-1)(2m+1)\}}$, where $a_n\gtrsim b_n$ means $a_n \ge Cb_n$ for some constant $C$, the theoretical properties of \eqref{eqn::hierbasis} will be maintained. This bound relies on the smoothness of the underlying $f$; choosing $\K\gtrsim {n}^{2/3}$ gives a conservative upper bound which is independent of the underlying $f$. Our theoretical results do not establish tight bounds on function approximation, but we conjecture that they can be improved to obtain the usual $\K\gtrsim n^{1/(2m+1)}$ truncation level. Additionally, as discussed in Section~\ref{sec:Algorithm}, by using $\K$, rather than $n$, basis functions, the computational complexity decreases from 
{$O(n^2)$ to $O(n\K)$}. A similar result holds for the  sparse additive framework with
\begin{equation}
\small
\label{eqn::AhierTrunc}
\wh{\beta}^{\text{S-hier}(\K)}_1,\ldots,\wh{\beta}^{\text{S-hier}(\K)}_p = \underset{  {\beta}_j \in \mathbb{R}^{\K}}{\operatorname{argmin}} \ \frac{1}{2}\Big\|  {y} - \sum_{j=1}^p  {\Psi}^{j}_{\K}  {\beta}_j\Big\|_2 + \lambda^2 \sum_{j=1}^p \Omega_j\left(  {\beta}_j\right) + \lambda \sum_{j=1}^p \Big\|  {\Psi}^{j}_{\K}  {\beta}_j\Big\|_n,
\end{equation}
where, now, $\Omega_j\left(  {\beta}_j\right) = \sum_{k=1}^{\K} w_k \big\|   {\Psi}^{j}_{k:\K}  {\beta}_{j,k:\K} \big\|_n$. Choosing the pre-truncation level, $\K$, is easier than the truncation level for the simple basis expansion estimator \citep{chentsov1962evaluation}. The latter requires an {exact} truncation level that is neither too large, nor too small, whereas 
the former only requires a level that is not too small.

\subsection{Algorithm for the univariate and  sparse additive framework}
\label{sec:Algorithm}
An appealing feature of our framework is its computational efficiency. Problem~\eqref{eqn::hierTrunc} can be solved via a one-step coordinate descent algorithm.
Using a QR decomposition ${\Psi}_{\K} = UV$ with $U\in \mathbb{R}^{n\times \K}$ and $U^{\top} U/n = I_{\K}$, we can re-write \eqref{eqn::hierTrunc} as
\begin{equation}
\underset{\wt{  {\beta}}\in \mathbb{R}^{\K} }{\text{ minimize } } \frac{1}{2n}\, \left\|  {y} - U\wt{  {\beta}} \right\|_2^2 + \lambda\sum_{k=1}^{\K} w_{k} \left\| \wt{  {\beta}}_{k:\K} \right\|_2 ,
\label{eqn:HierBasisSimple1}
\end{equation}
where $\wt{  {\beta}} = V  {\beta}$.
Applying the results of \cite{jenatton2010proximal}, gives us Algorithm~\ref{alg:univariate}.
\begin{algorithm}
	\caption{One-step coordinate descent for the univariate problem} \label{alg:univariate}
	\vspace*{-12pt}
	\begin{tabbing}
		\enspace Initialize $  {\beta}^{0} =\cdots =   {\beta}^{\K} \gets U^{\top}   {y}/n$ \\
		\enspace For{ $k = \K, \ldots, 1$} \\
		\qquad Update $  {\beta}^{k-1}_{k:\K} \gets \left( 1 - {w_{k}\lambda  }/{\|{  {\beta}}^k_{k:\K}\|_2} \right)_{+}{  {\beta}}^k_{k:\K}$, where $(x)_{+} = \max(x,\, 0)$\\
		\enspace Return $  {\beta}^0$ 
	\end{tabbing}
\end{algorithm}

The reformulation in \eqref{eqn:HierBasisSimple1} can also be used to  efficiently solve the sparse additive extension~\eqref{eqn::AhierTrunc} via a block coordinate descent algorithm. Specifically, given a set of estimates $({\beta}_j)_{j=1}^p$, we fix all but one of the vectors $ {\beta}_j$ and optimize over the non-fixed vector using Algorithm~\ref{alg:univariate}. Iterating until convergence yields the solution to problem~\eqref{eqn::AhierTrunc}; see 
the Supplementary Material.

Solving problem \eqref{eqn::hierTrunc} requires a QR decomposition of the matrix ${\Psi}_{\K}$ followed by the multiplication $U^{\top}{y}$; these steps require ${O}(n\K^2)$ and ${O}(n\K)$ operations, respectively. However, these steps are only needed once for a sequence of $\lambda$ values. For the additive proposal \eqref{eqn::AhierTrunc}, $p$ such QR decompositions are needed once for the entire $\lambda$ sequence. 

By Proposition~2 of \cite{jenatton2010proximal}, for a given $\lambda$, problem \eqref{eqn:HierBasisSimple1} can be solved in ${O}(\K)$ operations. Each block update requires a matrix multiplication $U_j^{\top}  {r}_{-j}$ followed by solving the proximal problem \eqref{eqn:HierBasisSimple1}, see the Supplementary Material. This requires ${O}(n\K)$ operations. Thus, our sparse additive proposal requires ${O}(np\K)$ operations, which is the computational complexity of the lasso~\citep{friedman2010regularization} when $\K=1$.

The above computational complexity calculations indicate that our univariate and  sparse additive estimates can be obtained very efficiently. 
In fact, using our \texttt{R} implementation, the median time for solving the univariate problem for an example with $\K = n = 300$ is 0$\cdot$17 seconds on an Intel\textregistered \  CORE\texttrademark \ i5-3337U, 1$\cdot$80 GHz processor. 
The median time for solving the sparse additive framework for the simulation setting of Section~\ref{sec:SimulationAdditive} on a grid of 50 $\lambda$ values is 5$\cdot$96 seconds. 

\subsection{Degrees of freedom}
\label{sec:DegreesOfFreedom}
For regression with fixed design and $\e_i\sim \mathcal{N}(0,\sigma^2)$, we consider the definition of degrees of freedom given by~\cite{stein1981estimation}, 
$
\df = \sum_{i=1}^{n}\mathrm{cov}(y_i,\, \wh{y}_i)/\sigma^2\ ,
$
where $\wh{y}_i$ are the fitted response values. 
We apply Claim 3$\cdot$2 of \cite{haris2016convex} to derive an unbiased estimate of $\df$ for the estimator~\eqref{eqn:HierBasisSimple1}, using the decomposition ${\Psi}_{\K} = UV$ from Section~\ref{sec:Algorithm}. 
Let $K_0 = \max (k: \wh{\beta}_k\not=0)$, and let $U_{K_0} \in \mathbb{R}^{n\times K_0}$ denote the first $K_0$ columns of $U$. For a vector $  {\nu}\in \mathbb{R}^{\K}$, define $  {\nu}_{k:K_0}\in \mathbb{R}^{K_0}$ as $  {\nu}_{k:K_0} = (0, \ldots, 0, \nu_k,\nu_{k+1}, \ldots,\nu_{K_0})^{\top}$. We arrive at the following lemma. 
%
\begin{lemma}
	An unbiased estimator for the degrees of freedom of $\wh{\beta}$ in \eqref{eqn::hierTrunc} is
	\begin{equation*}
	\scriptsize
	\wh{\df} = 1 + \mathrm{tr}\left\{ U_{K_0}\left[ I_{K_0} + \sum_{k=1}^{K_0}\lambda w_{k}\left\{ \frac{\mathrm{diag}(  {1}_{k:K_0})}{\| \wh{  {\beta}}_{k:K_0} \|_2}  - \frac{\wh{  {\beta}}_{k:K_0}\wh{  {\beta}}_{k:K_0}^{\top} }{\| \wh{  {\beta}}_{k:K_0} \|_2^3}\right\} \right]^{-1}\frac{U_{K_0}^{\top}}{n}\left( I_n -   {1}_n  {1}_n^{\top}/n \right)  \right\},
	\end{equation*}
	where $\mathrm{diag}(  {\nu})\in \mathbb{R}^{K_0\times K_0}$ is a diagonal matrix with $  {\nu}\in \mathbb{R}^{K_0}$ on the main diagonal.
\end{lemma}
\vspace{-3mm}
\subsection{Non-additive multivariate regression}
\label{sec:MultivariateHB}
For vectors $x\in \mathbb{R}^p$ and $\nu_k \in \mathbb{Z}_+^p$, define 
$ {x}^{\nu_k} = x_1^{\nu_{11}}\times \cdots\times x_p^{\nu_{kp}}$. Now for functions $f^0:\mathbb{R}^p\to \mathbb{R}$, consider the basis representation
$
f^0(x) = \sum_{k=1}^{\K} \psi_k({x}^{  {\nu_k}})\beta^0_k,
$
for univariate functions $(\psi_j)_{j=1}^{\infty}$, and $\nu_1,\ldots,\nu_{\K} \in \mathbb{Z}_+^p$, where 
\begin{equation*}
\|{\nu_k}\|_1 = 1\ \left(k=1,\ldots,p\right),\quad  \|{\nu_k}\|_1 = 2\ \Big(k=p+1,\ldots, \binom{p+2}{p}-1\Big),
\end{equation*}
and so on. As in the univariate case, let ${\Psi}_{\K}\in \mathbb{R}^{n\times \K}$ be the matrix with entries ${\Psi}_{\K (i,k)} = \psi_k(  {x_i}^{  {\nu_k}})$. Then, our multivariate regression estimator is simply \eqref{eqn::hierTrunc} with  weights given by 
\begin{equation}
w_{q_k} = k^m - (k-1)^m, \quad  q_k = \binom{k+p-1}{p},
\label{eqn:weightsMultivariateHB}
\end{equation}
and $w_k = 0$ for all other $k$. Figure~\ref{fig:penaltyMultivariate} demonstrates the multivariate penalty for $p = 2$ and $\psi_k$ the identity function; that is, for $z\in \mathbb{R},\, \psi_k(z) \equiv z$. It is clear from the figure how the multivariate penalty is a natural extension of the univariate one: when $\psi_k(z) = z$, the fitted model can be a multivariate polynomial of any degree. With this choice of basis functions, our multivariate proposal acts as a procedure for selecting the complexity level of interaction models. This problem can be solved using Algorithm~\ref{alg:univariate} with a single pass over the basis elements. 

\begin{figure}
	\centering
	\includegraphics[scale = 0.34]{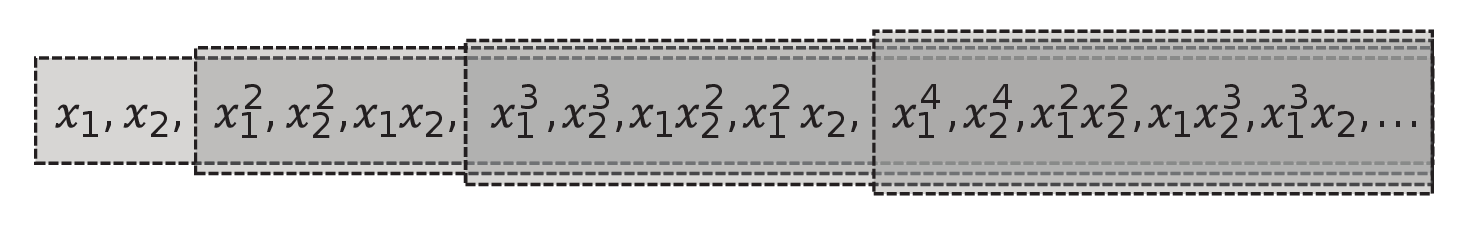}
	\caption{Visual representation of the multivariate penalty with $p = 2$ and $\psi_j(x) \equiv x$.}
	\label{fig:penaltyMultivariate}
\end{figure}


\subsection{Extension to classification}
\label{sec:classification}
We can extend our methodology to the setting of binary classification via a logistic loss function. Let $ y_i\in \left\{ -1,\, 1 \right\}\ (i=1,\ldots,n)$ be the observed response. We then fit
\begin{equation}
(\beta_0, \wh{  {\beta}}) = \underset{\beta_0\in \mathbb{R},\,   {\beta}\in \mathbb{R}^{n} }{\arg\min }\  \frac{1}{2n} \sum_{i=1}^n \log\left( 1 + \exp\left[ -y_i\{\beta_0 + \left(  {\Psi}_n  {\beta}\right)_i\} \right] \right) + \lambda \Om(  {\beta}).
\label{eqn:LogisticHierBasis}
\end{equation}
%
As with the least squares loss, \eqref{eqn:LogisticHierBasis} can be naturally extended to sparse additive models, by using both penalties in \eqref{eqn::SparseAddhier}. 
The problem can be efficiently solved via a proximal gradient descent algorithm \citep{combettes2011proximal}; see 
the Supplementary Material for details.

\section{Theoretical results}
\label{sec:TheoreticalResults}
\subsection{Summary of theoretical contributions}

To investigate finite sample properties of our estimators, we combine previously developed ideas from empirical process theory and metric entropy with a number of novel results about convergence rates of sparse additive models, and the metric entropy of our hierarchical class. 

Our new results in Section~\ref{sec:additiveHierBasisConvergence} allow one to establish convergence rates for a broad class of penalized sparse additive model estimators. Under a {compatibility condition} on the features, these rates match the minimax lower bound for estimation of sparse additive models under independent component functions \citep{raskutti2009lower}. Thus, our sparse additive estimators are rate-optimal. With no such assumptions, in Theorem~\ref{thm:thm1} we obtain rates that are the additive analog to assumption-free convergence rates for the lasso~\citep{chatterjee2013assumptionless}. Such assumption-free convergence rates have not been previously derived for sparse additive models. 

Finally, key for our theoretical analyses is the entropy of our hierarchical class; we calculate these with matching upper and lower bounds in Lemmas~\ref{lemma:upperBoundEntropy} and \ref{lemma:lowerBoundEntropy}. These new results allows us to show that our univariate and sparse additive estimators, \eqref{eqn::hierbasis} and \eqref{eqn::SparseAddhier}, are minimax rate-optimal within the hierarchical univariate and hierarchical sparse additive classes, respectively. 

\subsection{Entropy-based rates}
We begin by stating two well-known results. We then present our contributions in Sections~\ref{sec:hierbasisConvergence} and \ref{sec:additiveHierBasisConvergence}.
Firstly, Theorem 1 of \cite{yang1999information} establishes a lower bound for the minimax rate subject to certain conditions. Secondly, a framework for establishing an upper bound on convergence rates is given by Theorem 10$\cdot$2 of \cite{vandegeer2000empirical}. Here, we require a slight generalization of this result, which we state below and prove 
in the Supplementary Material. 

We first introduce some terminology and notation. For a set $\mathcal{F}$ equipped with some metric $d(\cdot,\, \cdot)$, the subset $\{f_1,\ldots,f_N\} \subset \mathcal{F}$ is a $\delta$-cover if for any $f\in \mathcal{F}$
$
\min_{1\le i \le N} d(f,\, f_i) \le \de.
$
The log-cardinality of the smallest $\de$-cover is the $\de$-entropy of $\mathcal{F}$ with respect to metric $d(\cdot,\, \cdot)$. We denote by $H(\de,\, \mathcal{F},\, Q)$, the $\de$-entropy of a function class $\mathcal{F}$ with respect to the $\|\cdot\|_Q$ metric for a measure $Q$, where $\|f\|_Q^2 =\int \{f(x)\}^2\,dQ(x)$.  For a fixed sample $x_1,\ldots,x_n$, we denote by $Q_n$ the empirical measure $Q_n = n^{-1}\sum_{i=1}^{n}\delta_{x_i}$ and use the short-hand notation $\|\cdot\|_n = \|\cdot\|_{Q_n}$.

\begin{theorem}[Theorem 1, \cite{yang1999information}]
	\label{thm:lowerMinimax}
	Consider the model
	$y_i = f^0(x_i) + \e_i\ (i=1,\ldots,n),$
	with independent and identically distributed $\e_i{\sim}\mathcal{N}(0,\sigma^2)$, $x_i\sim Q$.
	Assume the entropy condition 
	$H(\delta,\ \mathcal{F},\, Q ) = A_0\delta^{-\alpha}$ 
	holds for some function class $\mathcal{F}$ for $\alpha\in (0,2)$, and $A_0>0$.
	Then, for a constant $A_1$ depending on $A_0,\ \alpha$ and $\sigma^2$, 
	\begin{equation*}
	\min_{\wh{f}}\max_{f^0\in \mathcal{F} }\ E \Big( \big\|\wh{f} - f^0\big\|_{Q}^2\Big) \ge A_1n^{-{2}/{(2+\alpha)}}  \ ,
	\end{equation*}
	where the minimum is over the space of all measurable functions. 
\end{theorem}
\begin{theorem}[Theorem 10$\cdot$2, \cite{vandegeer2000empirical}]
	\label{thm:upperRate}
	Consider the model $y_i = f^0(x_i) + \e_i\ (i=1,\ldots,n)$, with independent sub-Gaussian noise $\e_i$. Let
	$\wh{f} = \arg\min_{f\in \mathcal{F}_{n} } \  \left\|   {y} - f\right\|_n^2 + 2\lambda_n^2\Om(f\mid Q_n),$
	for some function class $\mathcal{F}_n$ and semi-norm $\Om(\cdot\mid Q_n)$ on $\mathcal{F}_n$ which satisfy the entropy condition 
	$H[\delta,\{f\in \mathcal{F}_n: \Om(f\mid Q_n)\le 1\},Q_n]\le A_0\delta^{-\alpha},$ 
	for $\alpha\in (0,2)$. Then for  
	$\lambda_n^{-1} = n^{{1}/(2+\alpha)}\left\{\Om(f_{n}^*\mid Q_n) \right\}^{(2-\alpha)/\{2(2+\alpha)\} }$,
	and for any function $f_n^*\in \mathcal{F}_n$, there is a constant $c$ such that for all $T\ge c$, with probability at least $1-c\exp\left\{ -(T/c)^2 \right\}$, 
	\begin{equation*}
	\big\|\wh{f} - f^{0}\big\|_n^2 \le 5 \max \left\{2\big\|f^0 - f_n^*\big\|_n^2,\ C_0\lambda_n^2\Om(f_n^*|Q_n) \right\},
	\end{equation*}
	where $C_0$ is a constant that depends on $\alpha$ and $T$.
\end{theorem}


Before specializing Theorems~\ref{thm:lowerMinimax} and \ref{thm:upperRate} to our proposal, we briefly discuss their assumptions. The main assumption for Theorem~\ref{thm:lowerMinimax} is an entropy condition for the function class $\mathcal{F}$, which contains $f^0$. This is a common condition needed to quantify the size of an infinite-dimensional space $\mathcal{F}$. The entropy condition $H(\delta, \mathcal{F}, Q) = A_0\delta^{-\alpha}$, is satisfied by commonly-used function classes. Examples include, bounded Lipschitz functions with $\alpha=1$, bounded monotone functions with $\alpha=1$, and $m$th order Sobolev functions with $\alpha = 1/m$. Theorem~\ref{thm:upperRate} requires a similar entropy condition on a sequence of function spaces, but it does relax conditions on $f^0$ allowing it to be arbitrary, not necessarily in a specific class. Requiring an entropy condition for a sequence $(\mathcal{F}_n)_{n=1}^{\infty}$ may seem restrictive; but often, as with our hierarchical function class defined below, for all $n$, $\mathcal{F}_n \subseteq \mathcal{F}$ for some class  $\mathcal{F}$. Thus, it suffices to prove an entropy bound for $\mathcal{F}$. Finally, we also require the noise $\e_i$ to be independent and sub-Gaussian to use standard results from the empirical processes literature. While the original theorem of \cite{yang1999information} requires identically distributed Gaussian noise to prove a lower bound, the fact that Gaussian random variables are sub-Gaussian allows us to generalize the original result. In the following section, we define and establish entropy bounds for our hierarchical function class.


\subsection{Entropy results for the proposed penalty}
\label{sec:entropyResultsHB}
To set up the notation, we define the univariate function class
\begin{equation}
\mathcal{F}_n = \Big\{ f_{\beta}(x) = \sum_{k=1}^{\K_n} \psi_k(x)\beta_k\, : \int \psi_k\psi_l\, dQ = 0 \text{ for } k\not= l, \int \psi_k^2\, dQ = 1 \Big\}\ ,\  x\in \mathbb{R},
\label{eqn:functionalClass}
\end{equation}
and the multivariate function class
\begin{equation}
\footnotesize
\mathcal{F}_{p,n} = \Big\{ f_{\beta}(  {x}) = \sum_{k=1}^{\K_n} \psi_k(  {x}^{  {\nu_k}})\beta_k\, : \int \psi_k(  {x}^{  {\nu_k}})\psi_l(  {x}^{  {\nu_l}})\, dQ = 0 \text{ for } k\not= l, \int \{\psi_k(  {x}^{  {\nu_k}})\}^2\, dQ = 1 \Big\}\ ,\  {x}\in \mathbb{R}^p,
\label{eqn:functionalClassMulti}
\end{equation}
where ${\nu_k}\in \mathbb{Z}_+^p$, ${x^{\nu_k}}$ was defined in Section~\ref{sec:MultivariateHB}, and $Q$ is the probability measure associated with ${x}$. In \eqref{eqn:functionalClass} and \eqref{eqn:functionalClassMulti}, we allow for the limiting case of $n = \infty$ with $\K_\infty = \infty$. With some abuse of notation, for ${\beta} \in \ell^2(\mathbb{R})$, we define
$
\|{\beta}_{k:\infty}\|_2^2 = \sum_{l = k}^{\infty} \beta_l^2.
$


To specialize Theorems~\ref{thm:lowerMinimax} and \ref{thm:upperRate}, we need to characterize $H(\delta, \mathcal{F}_{\infty}^M, Q)$, for $\mathcal{F}_{\infty}^M$ defined below in \eqref{eqn:populationHBpen}, and establish an upper bound for $H[\delta, \{f_{  {\beta}}\in \mathcal{F}_n: \Om(  {\beta}) \le 1 \}, Q_n]$.  
In the next lemma, Lemma~\ref{lemma:reductionEuclidean}, we show that the calculation of $H(\delta, \mathcal{F}_{\infty}^M, Q)$ and $H[\delta, \{f_{  {\beta}}\in \mathcal{F}_n: \Om(  {\beta}) \le 1 \}, Q_n]$ is equivalent to an entropy calculation for subsets of $\ell^2(\mathbb{R})$ and $\mathbb{R}^{\K_n}$, respectively, with respect to the usual $\|\cdot\|_2$ norm. This reduction allows us to use simple volume arguments and existing results for establishing the entropy conditions. The lemma considers our proposed penalty in full generality, that is the penalty~(\ref{eqn::hierbasis}) with any set of non-negative weights $w_k$. This lemma gives a similar reduction of entropy calculations for the multivariate case with little extra work.

\begin{lemma}[Reduction to $\ell^2(\mathbb{R})$ and $\mathbb{R}^{\K_n}$]
	\label{lemma:reductionEuclidean}
	Let $\mathcal{F}_{n}^{M}$ and $\mathcal{F}_{p,n}^M$ be the univariate and multivariate hierarchical basis expansion class with bounded penalty, respectively. Specifically,
	\begin{equation}
	\small
	\mathcal{F}_n^M = \{ f_{  {\beta}} \in \mathcal{F}_n: \sum_{k=1}^{\K_n}w_k\|  {\beta}_{k:\K_n}\|_2  \le M\}, \quad  \mathcal{F}_{p,n}^M = \{ f_{  {\beta}} \in \mathcal{F}_{p,n}: \sum_{k=1}^{\K_n}w_k\|  {\beta}_{k:\K_n}\|_2  \le M\},
	\label{eqn:populationHBpen}
	\end{equation}
	where we allow the limiting case of $n = \infty$. Then, $H(\delta,\, \mathcal{F}_n^M,\, Q)$ or $H(\delta,\, \mathcal{F}_{p,n}^M,\, Q)$ is equal to $H( \de,\, \mathcal{H}_{\K_n}^{w/M})$, the entropy of $\mathcal{H}_{\K_n}^{w/M}$ with respect to the $\|\cdot\|_2$ norm, where
	\begin{equation*}
	\mathcal{H}^{w/M}_{\K_n} = \Big\{ \beta\in \mathbb{R}^{\K_n}: \sum_{k=1}^{\K_n} {w_k}/{M}\|\beta_{k:{\K_n}}\|_2\le 1 \Big\}\ .
	\end{equation*}
	Secondly, assume that the Gram matrix $\Psi_{\K_n}^{\top}\Psi_{\K_n}/n$ has a finite maximum eigenvalue denoted by $\Lambda_{\max}$. Then, denoting $\mathcal{H}^w_{\K_n} = \mathcal{H}^{w/1}_{\K_n}$, we have
	\begin{equation*}
	H\Big[ \delta,\, \Big\{f_{\beta} \in \mathcal{F}_n :\sum_{k=1}^{\K_n}w_k\|\Psi_{k:\K_n}  {\beta}_{k:\K_n}\|_n \le 1\Big\},\, Q_n \Big] \le H\Big(\delta\Lambda_{\max}^{-1/2},\, \mathcal{H}_{\K_n}^w\Big). 
	\end{equation*}
	The above inequality also holds with $\mathcal{F}_n$ replaced by $\mathcal{F}_{p,n}$.
	
\end{lemma}

Lemma~\ref{lemma:reductionEuclidean} establishes the connections between entropy of the function classes of interest and the set $\mathcal{H}_{\K_n}^w$. It is easy to see that $H(\de, \mathcal{H}_{\K_n}^{w/M})$ and $H(\de \Lambda_{\max}^{-1/2}, \mathcal{H}_{\K_n}^w)$ are proportional to $H(\de, \mathcal{H}_{\K_n}^w)$ where the proportionality constants depend on $M$ and $\Lambda_{\max}$, respectively. The next lemma establishes an upper bound for $H(\de, \mathcal{H}_{\K_n}^w)$ for the proposed choice of univariate and multivariate weights. This upper bound is all we need to specialize Theorem~\ref{thm:upperRate}. 

\begin{lemma}[An upper bound]
	\label{lemma:upperBoundEntropy}
	Suppose $\de\ge 0$. For the region $\mathcal{H}_{\K_n}^w$ with univariate weights $w_k = k^m - (k-1)^m$, 
	$H( \delta,\, \mathcal{H}_{\K_n}^w)\le U_{E,1}\delta^{-{1}/{m}}$, 
	for constant $U_{E,1} > 0$. Moreover, for the multivariate weights~(\ref{eqn:weightsMultivariateHB}), we have
	$H( \delta,\, \mathcal{H}_{\K_n}^w)\le U_{E,2}\delta^{-{p}/{m}}$,
	for constant $U_{E,2} > 0$.
\end{lemma}


While Lemma~\ref{lemma:upperBoundEntropy} is sufficient for applying Theorem~\ref{thm:upperRate}, to invoke Theorem~\ref{thm:lowerMinimax} we need an exact value for the entropy up to a proportionality constant. A natural way to achieve this is to find a lower bound for the entropy which matches the upper bound; we do this in the following lemma. 
\begin{lemma}[A lower bound]
	\label{lemma:lowerBoundEntropy}
	For $\de \in ((w_1+\cdots+w_{\K_n+1})^{-m}, 1/2)$, for the region $\mathcal{H}_{\K_n}^w$ with univariate weights, $w_k = k^m - (k-1)^m$, we have
	$H(\delta,\, \mathcal{H}_{\K_n}^w) \ge L_{E,1} \delta^{-{1}/{m}}$,
	and for the multivariate weights~(\ref{eqn:weightsMultivariateHB}) we have
	$H(\delta,\, \mathcal{H}_{\K_n}^w) \ge L_{E,2} \delta^{-{p}/{m}}$,
	for constants $L_{E,1}, L_{E,2}>0$ and where we assume, for simplicity, that $\K_n = q_{\K'}-1$ for some $\K'$ and $q_{\K'}$ as defined in \eqref{eqn:weightsMultivariateHB}.
\end{lemma}

Lemmas~\ref{lemma:reductionEuclidean}--\ref{lemma:lowerBoundEntropy} demonstrate the motivation for our weights $w_k$; they define a function class with the same entropy as an $m$th order Sobolev class. In fact, our class $\mathcal{F}_{\infty}^M$ is a subset of the $m$th order Sobolev class $\mathcal{G}_2^M$, where $\mathcal{G}_{q}^M = \{ f_{  {\beta}} \in \mathcal{F}_\infty: \sum_{k=1}^{\infty}(k^m|\beta_k|)^q  \le M^q  \}$ is the weighted $L_q$ space~\citep{rauhut2016interpolation}. Furthermore, we prove that $\mathcal{G}_1^M\subseteq \mathcal{F}_{\infty}^M \subseteq \mathcal{G}_2^M$; see the Supplementary Material. While the Sobolev class $\mathcal{G}_2^M$ is common in the literature~\citep{ravikumar2009sparse,geer2010lasso}, the class $\mathcal{G}_1^M$ has recently gained attention in the function interpolation literature; see, for example, \cite{rauhut2016interpolation}, \cite{candes2008enhancing} and the references therein.

\subsection{Specializing Theorems~\ref{thm:lowerMinimax} and \ref{thm:upperRate}} 
\label{sec:hierbasisConvergence}
The following proposition establishes a lower bound for the minimax rate of estimating $f^0$, the true function which belongs to some function class $\mathcal{F}$. We consider three different choices for $\mathcal{F}$: the univariate class~\eqref{eqn:functionalClass}; the multivariate class~\eqref{eqn:functionalClassMulti}; and the Sobolev class $\mathcal{G}_2^M$. To prove the result, we use the fact that if an upper bound for the convergence rate can be found that matches the lower bound, then we can conclude that our estimator is minimax.  
\begin{proposition}
	\label{lemma:lowerBoundHB}
	For the $m$th order function class 
	$
	\mathcal{\mathcal{F}}_{\infty}^M =\{f\in \mathcal{F}_\infty:  \sum_{k=1}^{\infty} w_k\|\beta_{k:\infty}\|_2  \le M\},
	$
	where $w_k = k^m - (k-1)^m$,
	\begin{equation*}
	\min_{\wh{f}}\max_{f^0\in \mathcal{F}_{\infty}^{M}} E \left( \big\|\wh{f} - f^0 \big\|_Q^2\right) \ge A_1n^{-{2m}/{(2m+1)}}.
	\end{equation*}
	For the $m$th order multivariate class
	$
	\mathcal{\mathcal{F}}_{p,\infty}^M =\{f\in \mathcal{F}_{p,\infty}: \sum_{k=1}^{\infty} w_k\|\beta_{k:\infty}\|_2  \le M\},
	$ where $w_k$ are the weights defined in \eqref{eqn:weightsMultivariateHB},
	\begin{equation*}
	\min_{\wh{f}}\max_{f^0\in \mathcal{F}_{p,\infty}^{M}}  E \left( \big\|\wh{f} - f^0 \big\|_Q^2 \right) \ge A_2n^{-{2m}/{(2m+p)}}.
	\end{equation*}
	Finally, for the $m$th order Sobolev class $\mathcal{G}_{2}^M = \{ f\in \mathcal{F}_{\infty}: \sum_{k=1}^{\infty} (k^m\beta_k)^2 \le M^2 \}$,
	\begin{equation*}
	\min_{\wh{f}}\max_{f^0\in \mathcal{G}_{2}^{M}} E \left( \big\|\wh{f} - f^0 \big\|_Q^2\right) \ge A_3n^{-{2m}/{(2m+1)}}.
	\end{equation*}
\end{proposition}
We now specialize Theorem~\ref{thm:upperRate} to establish an upper bound for the convergence rate of the proposed univariate and multivariate estimators. The following proposition reveals some interesting insights. Firstly, with respect to the empirical norm, $\|\cdot\|_n$, our estimators achieve the minimax rate for the classes $\mathcal{F}_{\infty}^M$ and $\mathcal{F}_{p,\infty}^M$, as defined in \eqref{eqn:populationHBpen}. For the Sobolev class, $\mathcal{G}_{2}^M$, if $\sum_{k=1}^{\infty}w_k\|  {\beta}_{k:\infty}\|_2 \le C(M)$ for all $f_{  {\beta}} \in \mathcal{G}_2^M$, then our univariate estimator is minimax over the Sobolev class as well. This result also gives insight into the role of $\K_n$.
\begin{proposition}
	\label{lemma:upperBoundHB}
	Consider the model $y_i = f^0(  {x_i}) + \e_i\ (i=1,\ldots,n)$ for mean zero, sub-Gaussian noise $\e_i$. 
	Define the univariate and multivariate estimators as
	\begin{equation*}
	\wh{f}^{\text{uni}} = \underset{f_{\beta}\in \mathcal{F}_n }{\arg\min }\  \frac{1}{2}\left\|   {y} - f_{\beta}\right\|_n^2 + \lambda_n^2\Om^{\text{uni}}(  {\beta}); \quad \wh{f}^{\text{multi}} = \underset{f_{\beta}\in \mathcal{F}_{p, n} }{\arg\min }\  \frac{1}{2}\left\|   {y} - f_{\beta}\right\|_n^2 + \lambda_n^2\Om^{\text{multi}}(  {\beta}) ,
	\end{equation*}
	for $p = 1$ and $p>1$, respectively, where $\Om^{\text{uni}}$ is the penalty in \eqref{eqn::hierTrunc} and $\Om^{\text{multi}}$ is the penalty in \eqref{eqn:weightsMultivariateHB}.
	Assume that 
	$
	\max_{k}\|\psi_k\|_{\infty} = \psi_{\max} <\infty 
	$
	and that the Gram matrix $  {\Psi}_{\K_n}^{\top}  {\Psi}_{\K_n}/n$ has a bounded maximum eigenvalue denoted by $\Lambda_{\max}$. Then:
	
	for $p=1$ and $f^0\in \mathcal{F}_{\infty}^M$ there is a constant $c>0$ such that for all $T\ge c$, with probability at least $1-c\exp\left\{-(T/c)^2\right\}$,
	\begin{equation*}
	\|\wh{f}^{\text{uni}} - f^{0}\|_n^2 \le 5 \max \left\{C_1\K_n^{-(2m-1)},\ C_2n^{{-{2m}/{(2m+1)}}}
	\right\},
	\end{equation*}
	where $C_1,\, C_2>0$ are constants that depend on $M$, $\psi_{\max}$, $\Lambda_{\max}$, $m$ and $T$;
	
	for $p=1$, $f^0\in \mathcal{G}_2^M$ there is a constant $c>0$ such that for all $T\ge c$, 
	\begin{equation*}
	\|\wh{f}^{\text{uni}} - f^{0}\|_n^2 \le 5 \max \left\{C_1\K_n^{-(2m-1)},\ C_2C_3n^{{-{2m}/{(2m+1)}}}
	\right\},
	\end{equation*}
	with probability at least $1-c\exp\left\{-(T/c)^2\right\}$, where $C_1,\, C_2>0$ are constants that depend on $M$, $\psi_{\max}$, $\Lambda_{\max}$, $m$, $T$ and, for $f^0 = \sum_{k=1}^{\infty}\psi_k\beta^0_k$, we have $C_3^{(2m+1)/2} =  \sum_{k=1}^{\infty} w_k\|  {\beta}^0_{k:\infty}\|_2$;
	
	for $1 < p  < 2m$ and $f^0\in \mathcal{F}_{p,\infty}^M$, let $\K'$ be such that $\K_n = q_{\K'}-1$ for $q_{\K'}$ as in \eqref{eqn:weightsMultivariateHB}. Then there is a constant $c>0$ such that for all $T\ge c$, with probability at least $1-c\exp\left\{-(T/c)^2\right\}$,
	\begin{equation*}
	\|\wh{f}^{\text{multi}} - f^{0}\|_n^2 \le 5 \max \left\{C_1\K'^{-(2m-1)},\ C_2n^{{-{2m}/{(2m+p)}}}
	\right\},
	\end{equation*}
	where $C_1,\, C_2>0$ are constants that depend on $M$, $\psi_{\max}$, $\Lambda_{\max}$, $m$, and $T$.
\end{proposition}
The pre-truncation level $\K_n$ in Proposition~\ref{lemma:upperBoundHB} is not the truncation order selected by our proposal; rather, it is the pre-specified maximum order of our proposal with conservative truncation in \eqref{eqn::hierTrunc}. The above result demonstrates that we achieve usual non-parametric rates as long as the truncation level $\K_n$ satisfies $\K_n \gtrsim n^{2m/\{(2m+1)(2m-1)\} }$ justifying $n^{2/3}$ as a conservative choice. Furthermore, if a function belongs to the $m$th order hierarchical class, then it also belongs to the $m'$th order class for all $m'\le m$ and Proposition~\ref{lemma:upperBoundHB} holds with $m$ replaced by $m'$. This means we can misspecify the smoothness order in our estimator and still get nonparametric convergence rates.

\subsection{Theoretical results for sparse additive models}
\label{sec:additiveHierBasisConvergence}
We next establish convergence rates of high-dimensional sparse additive models in terms of a general entropy condition. 
Our first contribution is an oracle inequality for an upper bound on the prediction error of additive models, which establishes the consistency of estimators with slow convergence rates; these rates are ${O}(s\nu_n)$ where $s\nu_n^2$ is the minimax lower bound of \cite{raskutti2009lower} for sparse additive model and, $s$ is the cardinality of the set $S$ defined below. For completeness, in Theorem~\ref{thm:lowerAdditiveMinimax} we state the result of \cite{raskutti2009lower}, which assumes independent covariates. 
We then proceed to state a {compatibility condition} which leads to two propositions: firstly, it establishes convergence rates of order ${O}(s\nu_n^2)$ and, secondly, it automatically establishes minimax rates for univariate regression as a special case of an additive model with $p=1$. These contributions extend to a broad class of estimators; consequently, we can establish new results on convergence rates for some existing methods, including methods that extend smoothing splines and trend filtering \citep{meier2009high,sadhanala2017additive}.

Let $f^0$ be the true function such that
$y_i = f^0(  {x_i}) +\e_i \ (i =1,\ldots, n)$,
for independent, mean-zero noise $\e_i$, ${x_i} = (x_{i1},\ldots,x_{ip})^{\top} \in \mathbb{R}^p$.
Let $f^*$ be a sparse additive approximation to $f^0$,   
\begin{equation*}
f^*(  {x_i}) = c^0 + \sum_{j = 1}^p f^*_j(x_{ij}) = c^0 + \sum_{j \in S} f^*_j(x_{ij}),
\end{equation*}
where $S = \{j:f_j^*\not= 0\}$, which we call the active set, is a subset of $\{1,\ldots, p\}$ of size $s=|S|$ and, $c^0 = E(\bar{y})$ where $\bar{y}$ is the sample mean. To ensure identifiability, we assume 
$\sum_{i=1}^{n} f^*_j(x_{ij}) = 0$  $(j =1, \ldots, p).$
Consider the estimator $\wh{f} = \sum_{j=1}^{p} \wh{f}_j$, where 
\begin{equation}
\wh{f}_1,\ldots, \wh{f}_p = \underset{(f_j)_{j=1}^p\in \mathcal{F}}{ \arg\min }\  \frac{1}{2n}\sum_{i=1}^{n} \Big\{ y_i - \bar{y} - \sum_{j=1}^{p}f_{j}(x_{ij}) \Big\}^2 + \lambda_n\sum_{j=1}^{p}I(f_j)\ ,
\label{eqn:sparseAdditiveOptim}
\end{equation}
where $I(\cdot)$ is a penalty of the form 
$I(f_j) = \|f_j\|_n + \lambda_n \up(f_j),$
for a semi-norm $\up(\cdot)$. We can think of $\up(f_j)$ as a smoothness penalty for function $f_j$.

\begin{theorem}[Theorem 1, \cite{raskutti2009lower}]
	\label{thm:lowerAdditiveMinimax}
	Consider $n$ independent identically distributed samples from the sparse additive model 
	$y_i  = \sum_{j\in {S}} f^0_{j}(x_{ij}) + \e_i$ $(i=1,\ldots,n)$,
	where $|S| = s\le p/4$, $  {x_i}\sim Q$, $\e_i \sim \mathcal{N}(0,\sigma^2)$ and, $f^0_j \in \mathcal{F}$ where $\mathcal{F}$ is a class satisfying the entropy condition
	$H(\delta,\, \mathcal{F},\, Q) = A_0\delta^{-1/m},$
	with $m>1/2$. Further assume the covariates are independent, so, $Q = \bigotimes_{j=1}^p Q_j$. Then for a constant $C>0$, 
	\begin{equation*}
	\min_{(\wh{f})_{j=1}^p}\max_{ (f^0_{j} )_{j=1}^p\in \mathcal{F}}\ E\Big( \Big\|\sum_{j=1}^{p}\wh{f}_j - f^0_{j} \Big\|_Q^2  \Big) \ge \max\Big\{ \frac{\sigma^2 s\log(p/s)}{32n} ,\ C s\Big(\frac{\sigma^2}{n}\Big)^{{2m}/{(2m+1)}} \Big\},
	\end{equation*}
	where the minimum is over the set of all measurable functions.
\end{theorem}

We next state the first key result of this section, which establishes an oracle inequality for additive models, as well as slow rates of convergence.  
\begin{theorem}
	\label{thm:thm1}
	\sloppy Assume the model $y_i  =f^0(x_{i}) + \e_i$ $(i=1,\ldots,n)$, with mean-zero $\e_i$ satisfying $\max_{i\in \{1,\ldots,n\}} L^2 \left( E \left[ \exp\left\{(\e_i/L)^2 \right\} \right] -1 \right)\le \sigma_0^2$ for constants $L$ and $\sigma_0$. Assume the entropy condition 
	$H[\delta, \{f\in \mathcal{F} : \up(f)\le 1\}, Q_n] \le A_0\delta^{-1/m},$
	holds for $m>1/2$, for some function class $\mathcal{F}$ and, some constant $A_0$. Let
	$
	\rho_n = \kappa\max \{ n^{-{m}/{(2m+1)}}, ( {\log p}/{n})^{1/2} \},
	$
	where $\kappa = \kappa(A_0, m, L, \sigma_0)$ is a sufficiently large positive constant. 
	Then, for $\lambda_n\ge 4\rho_n$, with probability at least $1-2\exp(-c_1n\rho_n^2 ) -c_2\exp(-c_3n\rho_n^2 )$ the estimator (\ref{eqn:sparseAdditiveOptim}) satisfies
	\begin{equation*}
	\small
	\label{eqn:inequalitySlowRates}
	\begin{split}
	\|\wh{f} - f^0\|_n^2 + \lambda_n\sum_{j\in S^c}& \|\wh{f}_j - f^*_j\|_n + \frac{3\lambda_n^2}{2}\sum_{j\in S}\up(\wh{f}_j - f^*_j)
	\le 3 \lambda_n\sum_{j\in S} \|\wh{f}_j - f^*_j\|_n + {4\lambda_n^2} \sum_{j\in S}\up(f^*_j) + \|f^* - f^0\|_n^2,
	\end{split}
	\end{equation*}
	where $c_1 = c_1(A_0,\sigma_0),\, c_2 = c_2(A_0, m,L,\sigma_0)$ and $c_3 \ge 1/c_2^2$ are positive constants.
	Furthermore, if the function class $\mathcal{F}$ satisfies $\sup_{f\in \mathcal{F}}\|f\|_n\le R$, we have
	\begin{equation*}
	\label{eqn:inequalitySlowRatesMAIN}
	\|\wh{f} - f^0\|_n^2 \le 2 C_{s}\max\Big\{ sn^{-{m}/{(2m+1)}} ,\, s\Big(\frac{\log p}{n}\Big)^{1/2}\Big\}  + \|f^*-f^0\|_n^2,
	\end{equation*}
	where $C_s\ge 0$ depends on $\kappa,\, R$ and $\sum_{j\in S}\up(f_j^*)/s$.
\end{theorem}
Theorem~\ref{thm:thm1} needs the same assumptions as Theorem~\ref{thm:upperRate}, namely sub-Gaussian noise $\e_i$, and an entropy condition on the univariate function class $\mathcal{F}$. Its second part assumes a bound on the univariate class $\mathcal{F}$ to control the term $\|\wh{f}_j - f^*\|_n$. In Proposition~\ref{thm:additiveFastRates} we drop this bounded class assumption and establish fast rates of convergence using the compatibility condition stated next. 

\begin{definition}[Compatibility Condition]
	We say that the compatibility condition is met for the set $\wt{S}$, if for some constant $\phi(\wt{S})>0$, and for all additive $ f = \sum_{j=1}^{p}f_j$, satisfying
	$\sum_{j\in \wt{S}^c}\|f_j\|_n 
	\le 4\sum_{j\in \wt{S}} \|f_j\|_n,$
	it holds that
	$
	\sum_{j\in \wt{S}}\|f_j\|_n \le |\wt{S}|^{1/2}\|f\|_n/\phi(\wt{S}).
	$
\end{definition}

The above condition is a functional analogue of the compatibility condition used to prove oracle inequalities for the lasso~\citep{geer2009conditions}. Such conditions are common in the high-dimensional literature, for example \cite{meier2009high} and \cite{geer2010lasso} use a similar condition; recently \cite{raskutti2012minimax} and \cite{yuan2015minimax} used a functional version of the restricted eigenvalue condition for proving fast rates in additive models. 

\begin{proposition}
	\label{thm:additiveFastRates}
	Assume the conditions of Theorem~\ref{thm:thm1} and the compatibility condition for $S = \{j: f_j^*\not\equiv0\}$ hold. Then, with probability at least $1-2\exp(-c_1n\rho_n^2)-c_2\exp(-c_3n\rho_n^2)$, 
	\begin{equation*}
	\frac{1}{2}\|\wh{f} - f^0\|_n^2 
	\le C_{f}\max\left\{ sn^{-{2m}/{(2m+1)}} ,\, {\frac{s\log p}{n}}\right\} 
	+ 2\|f^* - f^0\|_n^2,
	\end{equation*}
	where $C_f\ge 0$ is a constant that depends on $\phi(S)$ and $\sum_{j\in S} \up(f^*_j)/s$.
\end{proposition}

Misspecifying the order $m$ is especially important here since $f_j$ can have different orders of smoothness. Using the same argument as the univariate case, let $m_j\ (j=1,\ldots,p)$ denote the smoothness order of the $j$th component; then our results are valid for any $m\le \min_j m_j$.
We end this section by specializing Theorem~\ref{thm:thm1} to univariate regression. 

\begin{proposition}
	Assuming the conditions of Theorem~\ref{thm:thm1} with $p=1$, 
	the compatibility condition holds trivially  with $\phi(S)=1$. Moreover, for a constant $C_f\ge 0$ that depends on $\up(f^*)$,  
	\begin{equation*}
	\frac{1}{2}\|\wh{f} - f^0\|_n^2 \le C_fn^{-{2m}/{(2m+1)}} + 2\|f^* - f^0\|_n^2,
	\end{equation*}
	with probability at least $1-2\exp\left\{-c_1\kappa n^{{1}/{(2m+1)}}\right\}-c_2\exp\left\{ - c_3 \kappa n^{{1}/{(2m+1)}}\right\}$.
\end{proposition}

\section{Simulation studies}
\label{sec:SimulationStudy}
\subsection{Simulation for univariate regression}
We begin with a simulation to compare the performance of our univariate framework to smoothing splines~\citep{wahba1990spline} and trend filtering~\citep{kim2009ell1,tibshirani2014adaptive}. Smoothing splines and trend filtering are implemented in \texttt{R} packages \texttt{splines} and \texttt{genlasso}~\citep{arnold2014genlasso}.

\begin{figure}
	\centering
	\begin{tabular}{cc}
		\includegraphics[width = 0.4\textwidth]{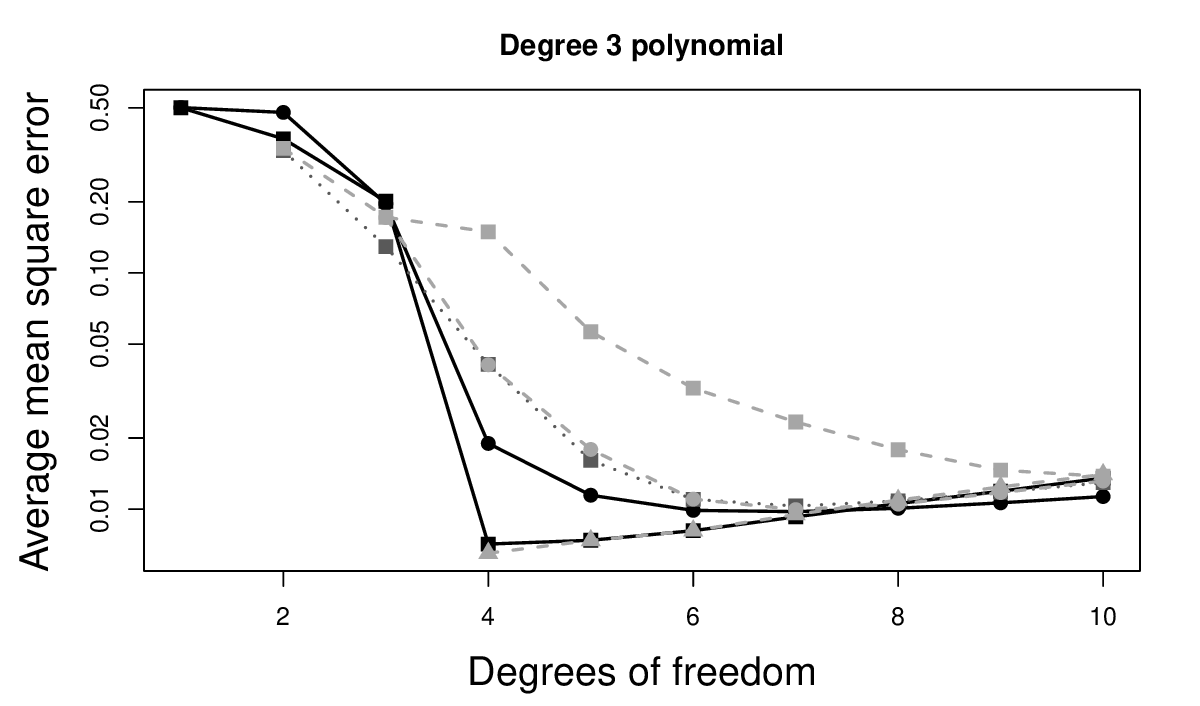} &
		\includegraphics[width = 0.4\textwidth]{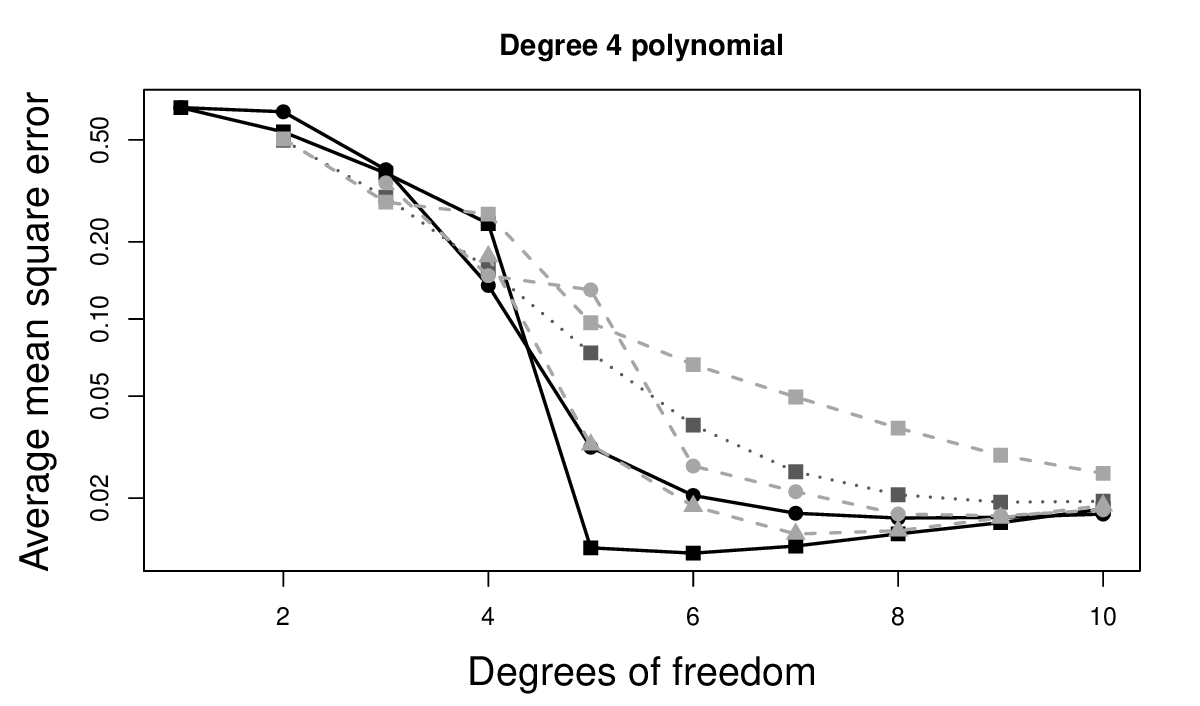} \\
		\includegraphics[width = 0.4\textwidth]{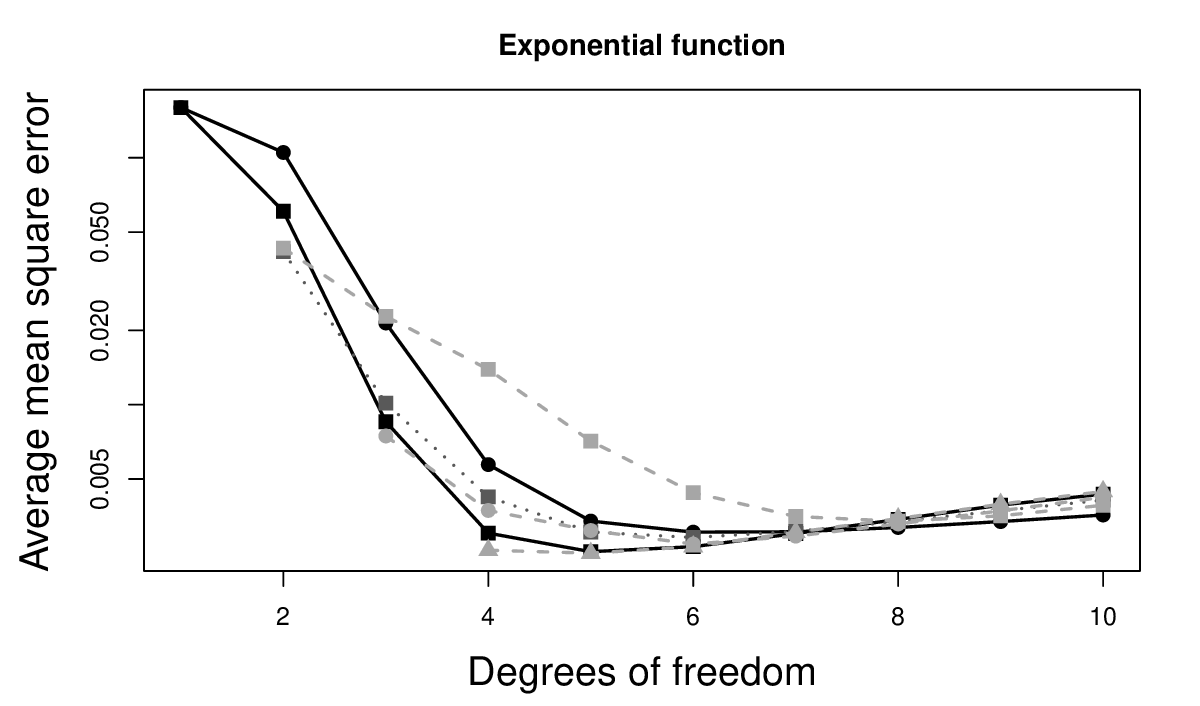} &
		\includegraphics[width = 0.4\textwidth]{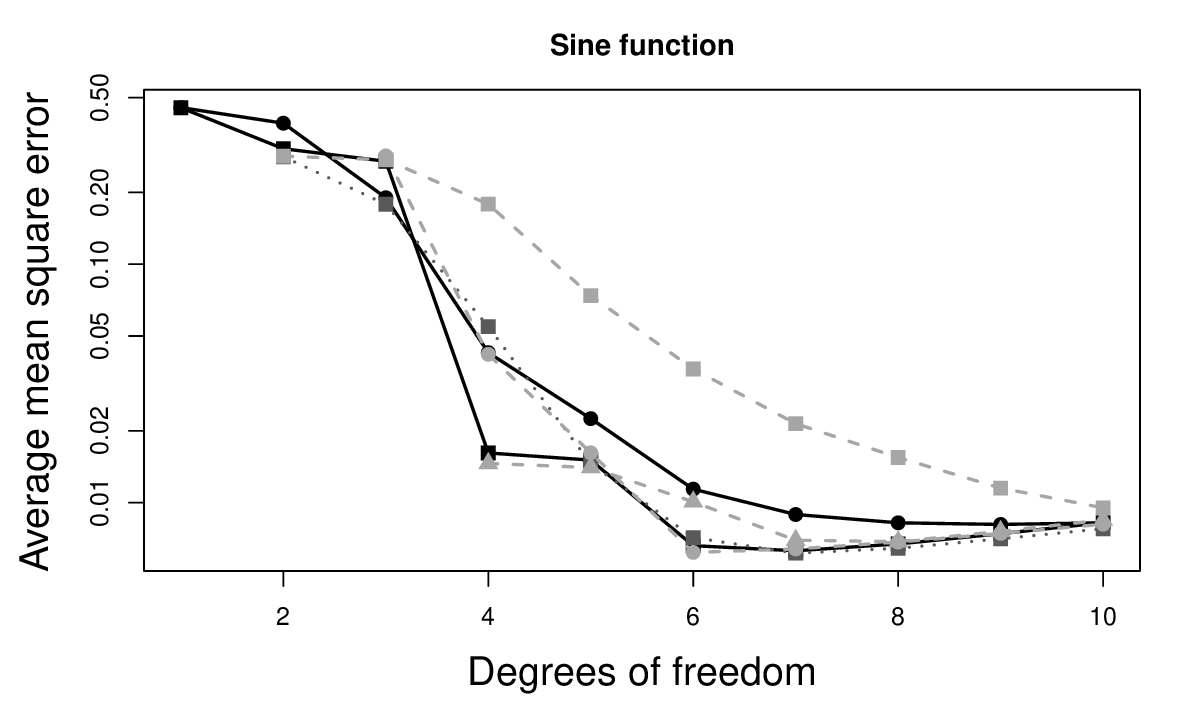}
	\end{tabular}
	\caption{Average mean square error, over 100 simulated datasets, as a function of degrees of freedom for true models given by $g_1,\,\ldots,\, g_4$ in \eqref{eq:simgs}. The colored lines indicate results for our framework with $m = $ 3~(\protect\includegraphics[scale=0.4]{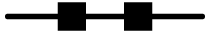}) and 1~(\protect\includegraphics[scale=0.4]{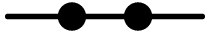}), Trend Filtering of order 1~(\protect\includegraphics[scale=0.4]{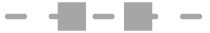}), 2~(\protect\includegraphics[scale=0.4]{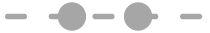}) and 3~(\protect\includegraphics[scale=0.4]{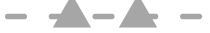}), and Smoothing Splines~(\protect\includegraphics[scale=0.4]{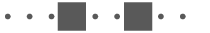}).} 
	\label{fig:MSEvsDOF}
\end{figure}


We generate the data as $y_i=f^0(x_i) + \e_i\ (i=1,\ldots,n)$ for different choices of the function $f^0$, and errors $\e_i \sim \mathcal{N}(  {0},\sigma^2)$ with $\sigma^2$ chosen to attain a fixed signal-to-noise ratio, 
$
\snr = {(n-1)}^{-1}\sum_{i=1}^{n} \{f^0(x_i)\}^2   / \sigma^2.
$
For this simulation we consider a fixed design with $x_i = i/n\ (i =1,\ldots, n)$. This facilitates comparison to trend filtering, which can become substantially slow for random $x_i$, particularly when the covariates are not uniformly distributed over a closed interval. We consider $n = 150$, $\snr \in \{2,3\}$, and four different functions, $f^0$ 
\begin{equation}
\label{eq:simgs}
\begin{split}
g_1(x) &= - \text{0$\cdot$43}+ \text{4$\cdot$83}x - \text{14$\cdot$65}x^2 + \text{11$\cdot$76}x^3,\\  
g_2(x) &= \text{0$\cdot$23} - \text{8$\cdot$44}x + \text{45$\cdot$20}x^2 - \text{81$\cdot$41}x^3 + \text{46$\cdot$59}x^4,\\ 
g_3(x) &= \exp( -5x + \text{0$\cdot$5} ) - \text{0$\cdot$4} \sinh( \text{2$\cdot$5} ), \quad g_4(x) = -\sin(7x - \text{0$\cdot$4}).
\end{split}
\end{equation}
We apply our proposal using a sequence of 100 $\lambda$ values linear on the log scale from $\lambda_{\max}$, for which $\wh{\beta} = 0$, down to $10^{-4}\lambda_{\max}$. We apply smoothing splines on a grid of 100 values of degrees of freedom from 10 to 1. Trend filtering is applied on a sequence of $\lambda$ values automatically selected by its \texttt{R} implementation with an average length of 279, 448 and 612 for trend filtering of order 1, 2 and 3, respectively. 
%
%
Figure~\ref{fig:MSEvsDOF} displays the mean square prediction error, $\mse = \|f^0 - \wh{f}\|_n^2$, of our method with $m = 1$ and $3$, smoothing splines and trend filtering of orders 1, 2 and 3, as a function of degrees of freedom. 
Our proposal appears to outperform the competitors in terms of mean square prediction error especially for polynomials. We observe comparable performance for the exponential and sine functions. This also provides empirical evidence for the theoretical results, where we proved our method to converge with rates comparable to smoothing splines. Since the functions considered in this simulation are smooth, as expected, our method with $m=1$ does not converge as fast as competing methods. 


We also compare the methods using an optimal tuning parameter that minimizes their prediction error on an independent test set of size $75$. The ratio of $\mse$ for each method over the $\mse$ of our proposal, and the corresponding ratios of degrees of freedom are shown in Table~\ref{tab:tab1}. 

\begin{table}
	\centering
	\caption{Average mean square errors of existing approaches relative to our univariate proposal with $m =3$; a value greater than 1 indicates a lower corresponding value for our method. The results presented are averages over 100 datasets along with $100\times$standard errors in parentheses}
		\begin{tabular}{lcccccccc}
			& \multicolumn{2}{c}{Degree 3 polynomial} & \multicolumn{2}{c}{Degree 4 polynomial}& \multicolumn{2}{c}{Exponential function} & \multicolumn{2}{c}{Sine function}\\
			& $\df$ &  $\mse$  & $\df$ &  $\mse$ & $\df$ &  $\mse$ & $\df$ &  $\mse$\\
			SS & 1$\cdot$32 (03) &  1$\cdot$48 (07) & 1$\cdot$25 (30) &  1$\cdot$58 (09) & 1$\cdot$13 (03) &  1$\cdot$21 (04) &  1$\cdot$10 (03) &  1$\cdot$00 (04)\\
			TF-1 & 1$\cdot$88 (13) &      2$\cdot$40 (30) & 1$\cdot$92 (12) &  2$\cdot$46 (30) & 1$\cdot$67 (16) &  1$\cdot$76 (31) & 1$\cdot$85 (13)  & 1$\cdot$88 (24) \\
			TF-2 & 1$\cdot$30 (09) &   1$\cdot$61 (24) &  1$\cdot$42 (09) &  1$\cdot$81 (26) & 1$\cdot$34 (13) &  1$\cdot$48 (29) & 1$\cdot$20 (10) & 1$\cdot$30 (21) \\
			TF-3 & 1$\cdot$06 (12) &  1$\cdot$37 (32) & 1$\cdot$27 (11) &  1$\cdot$66 (31) & 1$\cdot$47 (16) &  1$\cdot$51 (35) & 1$\cdot$34 (12)  & 1$\cdot$48 (25)
	\end{tabular}
	\begin{tablenotes}
		\small
		
		\item $\mse$, mean square prediction error; $\df$, degrees of freedom; SS, smoothing splines; TF-1, first order trend filter; TF-2, second order trend filter; TF-3, third order trend filter.
	\end{tablenotes}
	\label{tab:tab1}
\end{table}

%

\subsection{Simulation for multivariate additive regression}
\label{sec:SimulationAdditive}
Next, we compare the performance of our sparse additive estimator with   \citeauthor{ravikumar2009sparse}'s method. Their method is implemented in the \texttt{R} package \texttt{SAM}~\citep{zhao2014SAM} which uses natural spline basis functions. For a fairer comparison, we also implement their method using a polynomial basis expansion. Due to a lack of \texttt{R} packages for proposals of \cite{meier2009high} and \cite{lou2016sparse}, we defer the comparison to these methods to future work.

\begin{figure}
	\centering
	\includegraphics[width = 0.45\textwidth, clip=TRUE, trim=3mm 3mm 5mm 10mm
	]{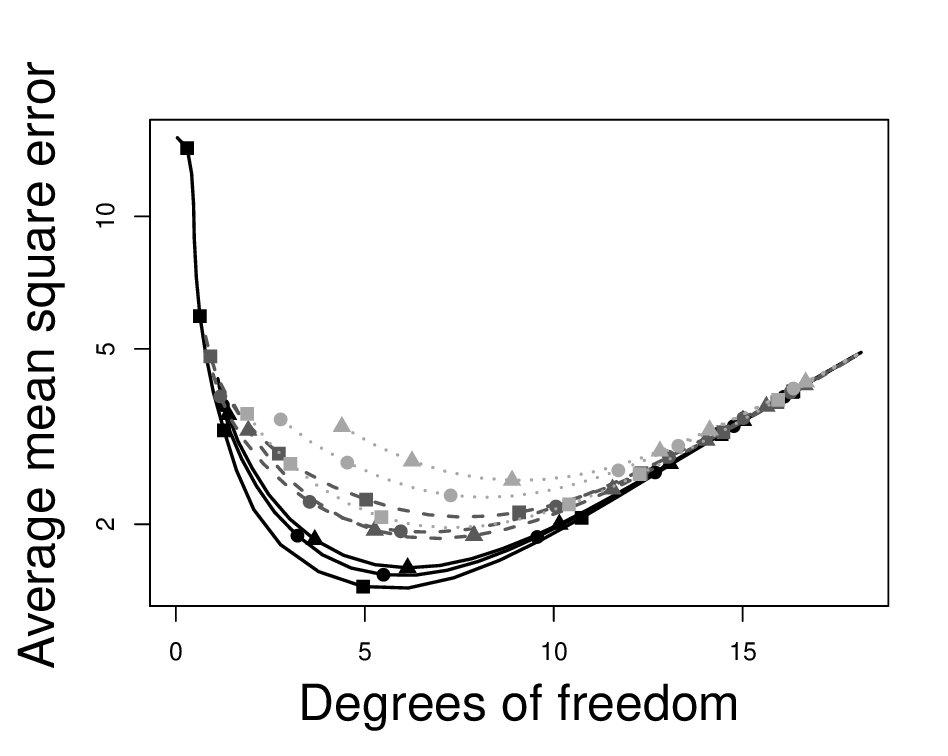}
	\includegraphics[width = 0.45\textwidth, clip=TRUE, trim=3mm 3mm 5mm 10mm
	]{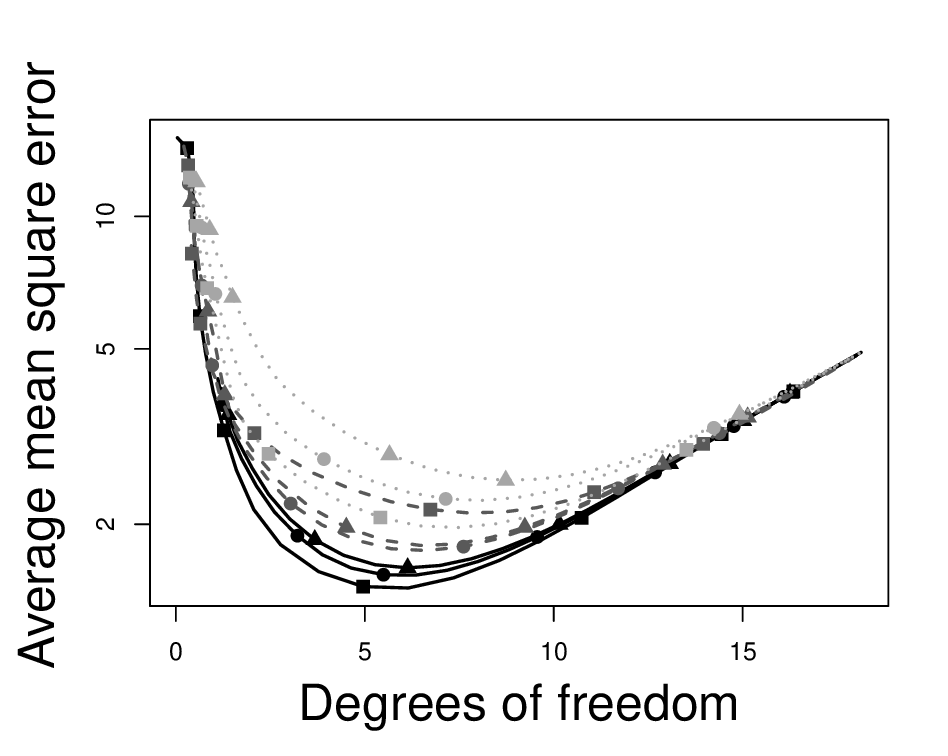}
	\caption{Average mean square error, over 100 simulated datasets, as a function of degrees of freedom for our proposal with $m=$ 1~(\protect\includegraphics[scale=0.3]{15_gray0_1.eps}),  
		2~(\protect\includegraphics[scale=0.3]{16_gray0_1.eps}) and 3~(\protect\includegraphics[scale=0.3]{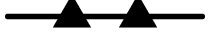}), compared to  \cite{ravikumar2009sparse} with 3~(\protect\includegraphics[scale=0.3]{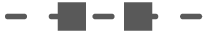}), 
		5~(\protect\includegraphics[scale=0.3]{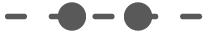}), 
		8~(\protect\includegraphics[scale=0.3]{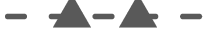}), 
		10~(\protect\includegraphics[scale=0.3]{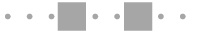}),
		15~(\protect\includegraphics[scale=0.3]{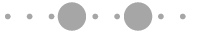}) and 
		20~(\protect\includegraphics[scale=0.3]{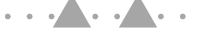}) basis functions. Left: Polynomial basis functions. Right: Natural spline basis functions.} 
	\label{fig:MSEvsLam}
\end{figure}


We consider the simulation setting of \cite{meier2009high} with some modifications to have high-dimensional data and smaller signal-to-noise ratio. We generate $n = 200$ samples for $p = 500$ features. The data is generated as
$y_i = 5f_1(x_{i1})+3f_2(x_{i2})+4f_3(x_{i3})+6f_4(x_{i 4}) + \e_i$ ($i=1,\ldots,n$), 
where $\e_i\sim \mathcal{N}(0, \sigma^2)$ with $\sigma^2$ such that $\snr = 3$ and 
\begin{gather*}
f_1(x) = x,\quad  f_2(x) = (2x-1)^2,\quad f_3 = 2\sin(2\pi x)/ \left\{2 - \sin(2\pi x )\right\},\\
f_4(x) = \text{0$\cdot$1}\sin(2\pi x) + \text{0$\cdot$2}\cos(2\pi x) + \text{0$\cdot$3}\sin^2(2\pi x) + \text{0$\cdot$4}\cos^3(2\pi x) + \text{0$\cdot$5}\sin^3(2\pi x)\ ;
\end{gather*}
the covariates are independently drawn from Uniform$(0,1)$. For $m=1$, we use the parametrization \eqref{eqn::AhierTrunc}. For $m = 2$ and $3$, we use 
\begin{equation}
\underset{  {\beta}_j \in \mathbb{R}^{\K} }{\text{minimize}} \frac{1}{2}\, \Big\|   {y} - \sum_{j=1}^{p} \Psi_{\K}^j  {\beta}_j \Big\|_n^2 + \gamma\lambda \sum_{j=1}^{p} \Om_j(  {\beta}_j) + (1-\gamma)\lambda \sum_{j=1}^{p} \|\Psi_{\K}^j  {\beta}_j\|_n \, ,
\end{equation}
with $\gamma =$ 0$\cdot$01 and 0$\cdot$001, respectively. All methods were fit over a sequence of 50 $\lambda$ values, decreasing linearly on the log-scale. We set the maximum number of basis functions $\K = 20$ for our estimator and use 3, 5, 8, 10, 15 and, 20 basis functions for \citeauthor{ravikumar2009sparse}'s proposal.

In terms of mean square prediction error, it is not surprising to observe superior performance of our methodology over that of \citeauthor{ravikumar2009sparse}'s proposal with polynomial basis in Figure~\ref{fig:MSEvsLam}. However, in the same figure, our method also seems to outperform the original proposal of \cite{ravikumar2009sparse} using natural splines. Overall, our proposal achieves the smallest mean square error for a substantial interval of degrees of freedom values in both panels of Fig.~\ref{fig:MSEvsLam}. 


\section{Analysis of Parkinson's telemonitoring data}
\label{sec:DataAnalysis}

We apply our method to the Parkinson's telemonitoring dataset~\citep{tsanas2010accurate}, obtained from the University of California Irvine Machine Learning Repository. The data consists of $n=5875$ observations and $p=18$ covariates including 16 biomedical voice measurements, age, and time of reading. Our goal is to predict the motor Unified Parkinson's Disease Rating Scale score. Apart from the analysis of the original dataset, we add 100 noise variables, uniformly generated from the unit interval, to study the sparsity properties of our proposal. We use 2/3 of the data as training and remaining as test, and average the results over 100 training-test splits. 

We compare our additive framework to \eqref{eqn::add} with fixed truncation, $K_j=K$, in the low-dimensional setting and compare our sparse additive framework to that of \cite{ravikumar2009sparse},  lasso~\citep{tibshirani1996regression}, and elastic net~\citep{zou2005regularization}  in the high-dimensional setting. For computational convenience, we fit our proposal with conservative  basis truncation~\eqref{eqn::hierTrunc} with $\K=\lceil (2n/3)^{1/2}\rceil$, the square-root of the number of {training} observations. In both settings, we fit our proposal with order $m=3$. We fit \citeauthor{ravikumar2009sparse}'s proposal with $K=2,4,8$; for a fairer comparison, we use a polynomial basis expansion. Other values of $K$, had comparable or worse performance and are not presented here. For each proposal, we also implement a relaxed version; re-fitting the selected non-zero coefficients via least squares. For a sequence of $\lambda$ values, or sequence of $K$ for \eqref{eqn::add}, we calculate the mean square test error and the model parsimony, the number of total non-zero basis functions used. Lower values of model parsimony correspond to more sparsity in either the number of components or truncation levels $K_j$. Figure~\ref{fig:DataAnalysis}, shows that our proposal outperforms competitors in the low and high-dimensional setting in terms of mean square error for a sequence of fitted models. The relaxed version of \citeauthor{ravikumar2009sparse}'s method achieves a lower test error but at the cost of reduced parsimony. The linear models, implemented by lasso and elastic net, had a substantially higher test error compared to additive models.

\begin{figure}
	\centering
	\includegraphics[width = 0.28\textwidth]{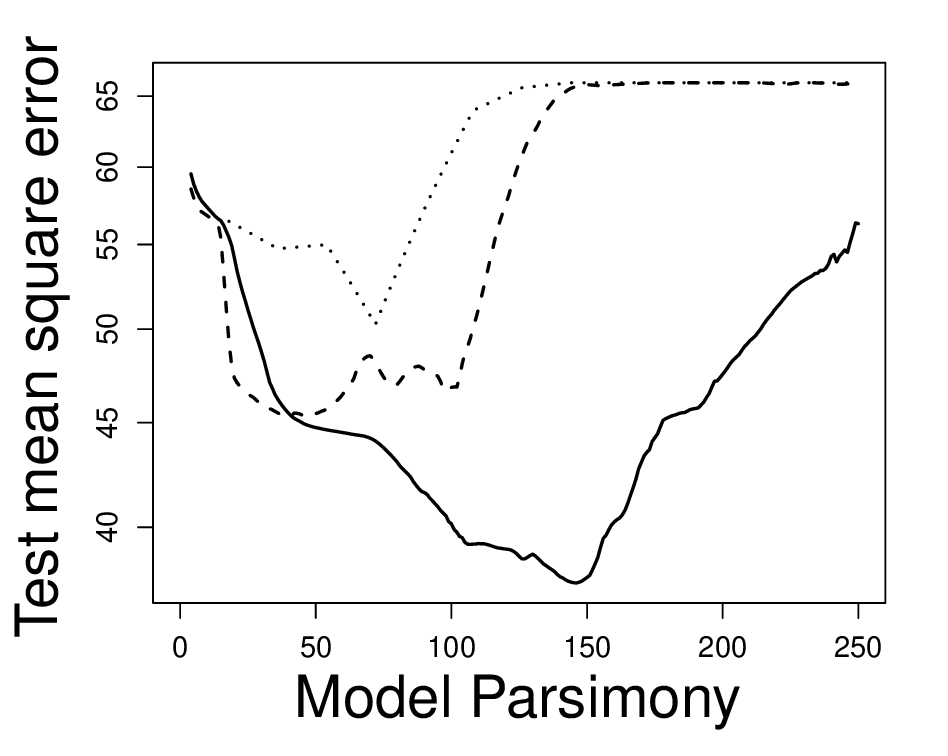}
	\includegraphics[width = 0.28\textwidth]{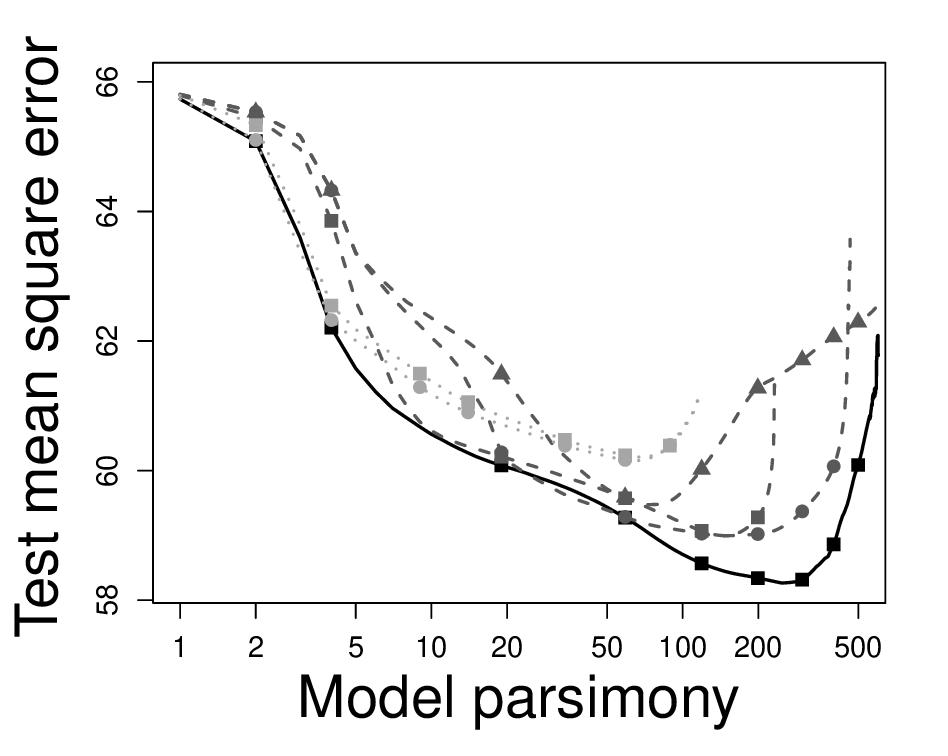}
	\includegraphics[width = 0.28\textwidth]{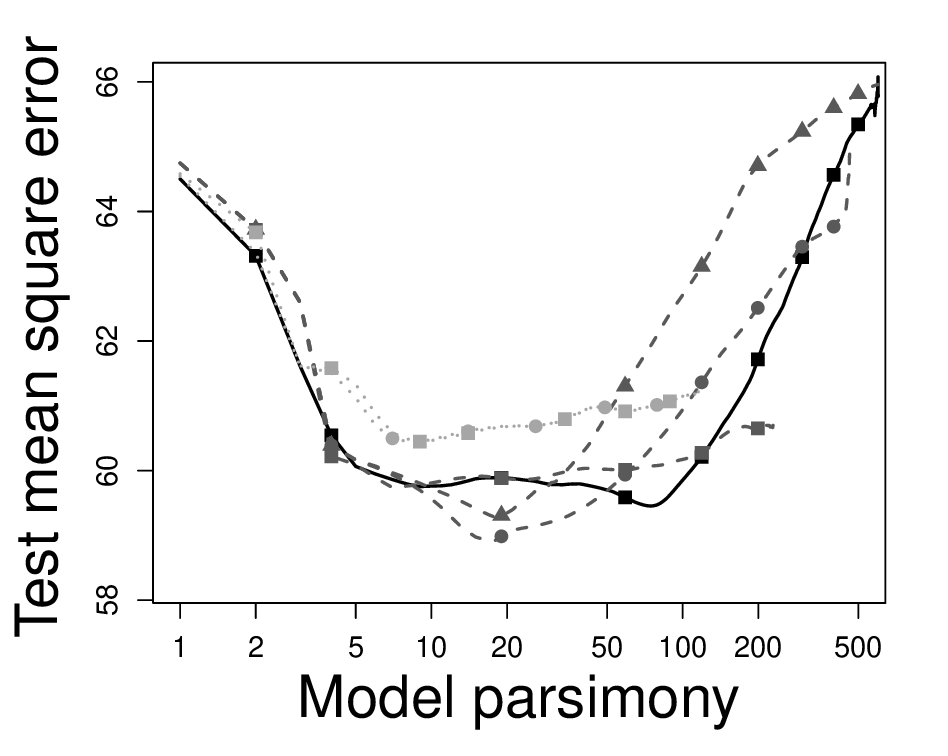}
	\caption{Analysis of Parkinson's telemonitoring dataset. {Left}: Low-dimensional results for our proposal~(\protect\includegraphics[height=0.5em]{NA_gray0_1.eps}), the simple truncation estimator~\eqref{eqn::add}~(\protect\includegraphics[height=0.5em]{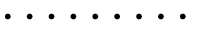}) and our relaxed proposal~(\protect\includegraphics[height=0.5em]{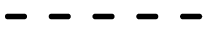}). Middle: High-dimensional results for our proposal~(\protect\includegraphics[height=0.5em]{15_gray0_1.eps}), \citeauthor{ravikumar2009sparse}'s proposal with 2~(\protect\includegraphics[height=0.5em]{15_gray35_2.eps}), 
		4~(\protect\includegraphics[height=0.5em]{16_gray35_2.eps}), and 8~(\protect\includegraphics[height=0.5em]{17_gray35_2.eps}) basis functions, lasso~(\protect\includegraphics[height=0.5em]{15_gray65_3.eps}) and elastic net~(\protect\includegraphics[height=0.5em]{16_gray65_3.eps}). Right: Results for the relaxed versions of all methods in the middle panel.}
	\label{fig:DataAnalysis}
\end{figure}


\section*{Acknowledgement}

We thank the associate editor and two referees for insightful comments that improved the manuscript. This work was partially supported by National Institutes of Health grants to A.S. and N.S., and National Science Foundation grants to A.S.

\section*{Supplementary material}
\label{SM}
\sloppy The online supplementary material~(appended below) includes additional figures, algorithm details, extension to binary response, and proofs for results in Section~\ref{sec:TheoreticalResults}. The \texttt{R} package \texttt{HierBasis}, available on \url{https://github.com/asadharis/HierBasis}, implements the proposed methods.

\bibliographystyle{biometrika}

\begin{thebibliography}{38}
	\expandafter\ifx\csname natexlab\endcsname\relax\def\natexlab#1{#1}\fi
	
	\bibitem[{Arnold \& Tibshirani(2014)}]{arnold2014genlasso}
	\textsc{Arnold, T.~B.} \& \textsc{Tibshirani, R.~J.} (2014).
	\newblock \textit{genlasso: Path algorithm for generalized lasso problems}.
	\newblock R package version 1.3.
	
	\bibitem[{Candes et~al.(2008)Candes, Wakin \& Boyd}]{candes2008enhancing}
	\textsc{Candes, E.~J.}, \textsc{Wakin, M.~B.} \& \textsc{Boyd, S.~P.} (2008).
	\newblock Enhancing sparsity by reweighted $\ell_1$ minimization.
	\newblock \textit{Journal of Fourier analysis and applications} \textbf{14},
	877--905.
	
	\bibitem[{\v{C}encov(1962)}]{chentsov1962evaluation}
	\textsc{\v{C}encov, N.} (1962).
	\newblock Evaluation of an unknown distribution density from observations.
	\newblock \textit{Doklady} \textbf{3}, 1559--1562.
	
	
	\bibitem[{Chatterjee(2013)}]{chatterjee2013assumptionless}
	\textsc{Chatterjee, S.} (2013).
	\newblock Assumptionless consistency of the lasso.
	\newblock \textit{arXiv preprint arXiv:1303.5817} .
	
	\bibitem[{Combettes \& Pesquet(2011)}]{combettes2011proximal}
	\textsc{Combettes, P.~L.} \& \textsc{Pesquet, J.} (2011).
	\newblock Proximal splitting methods in signal processing.
	\newblock In \textit{Fixed-Point Algorithms for Inverse Problems in Science and
		Engineering}, H.~H. Bauschke, R.~S. Burachik, P.~L. Combettes, V.~Elser,
	D.~R. Luke \& H.~Wolkowicz, eds. Springer New York, pp. 185--212.
	
		\bibitem[{Dumer(2006)}]{dumer2006covering}
	\textsc{Dumer, I.} (2006).
	\newblock Covering an ellipsoid with equal balls.
	\newblock \textit{Journal of Combinatorial Theory, Series A} \textbf{113},
	1667--1676.
	
	\bibitem[{Friedman et~al.(2010)Friedman, Hastie \&
		Tibshirani}]{friedman2010regularization}
	\textsc{Friedman, J.~H.}, \textsc{Hastie, T.~J.} \& \textsc{Tibshirani, R.~J.}
	(2010).
	\newblock Regularization paths for generalized linear models via coordinate
	descent.
	\newblock \textit{Journal of Statistical Software} \textbf{33}, 1--22.
	
	\bibitem[{Haris et~al.(2016)Haris, Witten \& Simon}]{haris2016convex}
	\textsc{Haris, A.}, \textsc{Witten, D.} \& \textsc{Simon, N.} (2016).
	\newblock Convex modeling of interactions with strong heredity.
	\newblock \textit{Journal of Computational and Graphical Statistics}
	\textbf{25}, 981--1004.
	
	\bibitem[{Hastie et~al.(2009)Hastie, Tibshirani \&
		Friedman}]{hastie2009elements}
	\textsc{Hastie, T.~J.}, \textsc{Tibshirani, R.~J.} \& \textsc{Friedman, J.~H.}
	(2009).
	\newblock \textit{The Elements of Statistical Learning}.
	\newblock Springer New York, 2nd ed.
	
	\bibitem[{Jenatton et~al.(2010)Jenatton, Mairal, Bach \&
		Obozinski}]{jenatton2010proximal}
	\textsc{Jenatton, R.}, \textsc{Mairal, J.}, \textsc{Bach, F.~R.} \&
	\textsc{Obozinski, G.~R.} (2010).
	\newblock Proximal methods for sparse hierarchical dictionary learning.
	\newblock In \textit{Proceedings of the 27th International Conference on
		Machine Learning (ICML-10)}.
	
	\bibitem[{Kim et~al.(2009)Kim, Koh, Boyd \& Gorinevsky}]{kim2009ell1}
	\textsc{Kim, S.}, \textsc{Koh, K.}, \textsc{Boyd, S.} \& \textsc{Gorinevsky,
		D.} (2009).
	\newblock $\ell_1$ trend filtering.
	\newblock \textit{SIAM Review} \textbf{51}, 339--360.
	
	\bibitem[{Koltchinskii \& Yuan(2010)}]{koltchinskii2010sparsity}
	\textsc{Koltchinskii, V.} \& \textsc{Yuan, M.} (2010).
	\newblock Sparsity in multiple kernel learning.
	\newblock \textit{Ann. Statist.} \textbf{38}, 3660--3695.
	
	\bibitem[{Lou et~al.(2016)Lou, Bien, Caruana \& Gehrke}]{lou2016sparse}
	\textsc{Lou, Y.}, \textsc{Bien, J.}, \textsc{Caruana, R.} \& \textsc{Gehrke,
		J.} (2016).
	\newblock Sparse partially linear additive models.
	\newblock \textit{Journal of Computational and Graphical Statistics}
	\textbf{25}, 1126--1140.
	
	\bibitem[{Mammen \& van~de Geer(1997)}]{mammen1997locally}
	\textsc{Mammen, E.} \& \textsc{van~de Geer, S.} (1997).
	\newblock Locally adaptive regression splines.
	\newblock \textit{The Annals of Statistics} \textbf{25}, 387--413.
	
	\bibitem[{Meier et~al.(2009)Meier, van~de Geer \& B{\"u}hlmann}]{meier2009high}
	\textsc{Meier, L.}, \textsc{van~de Geer, S.} \& \textsc{B{\"u}hlmann, P.}
	(2009).
	\newblock High-dimensional additive modeling.
	\newblock \textit{The Annals of Statistics} \textbf{37}, 3779--3821.
	
	\bibitem[{Meinshausen(2007)}]{meinshausen2007relaxed}
	\textsc{Meinshausen, N.} (2007).
	\newblock Relaxed lasso.
	\newblock \textit{Computational Statistics \& Data Analysis} \textbf{52},
	374--393.
	
	\bibitem[{Nadaraya(1964)}]{nadaraya1964estimating}
	\textsc{Nadaraya, E.~A.} (1964).
	\newblock On estimating regression.
	\newblock \textit{Theory of Probability and Its Applications} \textbf{9},
	141--142.
	
	\bibitem[{Petersen et~al.(2016)Petersen, Witten \& Simon}]{petersen2016fused}
	\textsc{Petersen, A.}, \textsc{Witten, D.} \& \textsc{Simon, N.} (2016).
	\newblock Fused lasso additive model.
	\newblock \textit{Journal of Computational and Graphical Statistics}
	\textbf{25}, 1005--1025.
	
	\bibitem[{Raskutti et~al.(2012)Raskutti, Wainwright \&
		Yu}]{raskutti2012minimax}
	\textsc{Raskutti, G.}, \textsc{Wainwright, M.~J.} \& \textsc{Yu, B.} (2012).
	\newblock Minimax-optimal rates for sparse additive models over kernel classes
	via convex programming.
	\newblock \textit{The Journal of Machine Learning Research} \textbf{13},
	389--427.
	
	\bibitem[{Raskutti et~al.(2009)Raskutti, Yu \& Wainwright}]{raskutti2009lower}
	\textsc{Raskutti, G.}, \textsc{Yu, B.} \& \textsc{Wainwright, M.~J.} (2009).
	\newblock Lower bounds on minimax rates for nonparametric regression with
	additive sparsity and smoothness.
	\newblock In \textit{Advances in Neural Information Processing Systems}.
	
	\bibitem[{Rauhut \& Ward(2016)}]{rauhut2016interpolation}
	\textsc{Rauhut, H.} \& \textsc{Ward, R.} (2016).
	\newblock Interpolation via weighted $\ell_1$ minimization.
	\newblock \textit{Applied and Computational Harmonic Analysis} \textbf{40}, 321
	-- 351.
	
	\bibitem[{Ravikumar et~al.(2009)Ravikumar, Lafferty, Liu \&
		Wasserman}]{ravikumar2009sparse}
	\textsc{Ravikumar, P.}, \textsc{Lafferty, J.}, \textsc{Liu, H.} \&
	\textsc{Wasserman, L.} (2009).
	\newblock Sparse additive models.
	\newblock \textit{Journal of the Royal Statistical Society: Series B
		(Statistical Methodology)} \textbf{71}, 1009--1030.
	
	\bibitem[{Sadhanala \& Tibshirani(2018)}]{sadhanala2017additive}
	\textsc{Sadhanala, V.} \& \textsc{Tibshirani, R.~J.} (2018).
	\newblock Additive models with trend filtering.
	\newblock \textit{arXiv preprint arXiv:1702.05037} .
	
	\bibitem[{Stein(1981)}]{stein1981estimation}
	\textsc{Stein, C.~M.} (1981).
	\newblock Estimation of the mean of a multivariate normal distribution.
	\newblock \textit{The Annals of Statistics} \textbf{9}, 1135--1151.
	
	\bibitem[{Stone(1977)}]{stone1977consistent}
	\textsc{Stone, C.~J.} (1977).
	\newblock Consistent nonparametric regression.
	\newblock \textit{The Annals of Statistics} \textbf{5}, 595--620.
	
	\bibitem[{Tibshirani(1996)}]{tibshirani1996regression}
	\textsc{Tibshirani, R.~J.} (1996).
	\newblock Regression shrinkage and selection via the lasso.
	\newblock \textit{Journal of the Royal Statistical Society. Series B
		(Methodological)} , 267--288.
	
	\bibitem[{Tibshirani(2014)}]{tibshirani2014adaptive}
	\textsc{Tibshirani, R.~J.} (2014).
	\newblock Adaptive piecewise polynomial estimation via trend filtering.
	\newblock \textit{The Annals of Statistics} \textbf{42}, 285--323.
	
	\bibitem[{Tsanas et~al.(2010)Tsanas, Little, McSharry \&
		Ramig}]{tsanas2010accurate}
	\textsc{Tsanas, A.}, \textsc{Little, M.~A.}, \textsc{McSharry, P.~E.} \&
	\textsc{Ramig, L.~O.} (2010).
	\newblock Accurate telemonitoring of parkinson's disease progression by
	noninvasive speech tests.
	\newblock \textit{IEEE transactions on Biomedical Engineering} \textbf{57},
	884--893.
	
	\bibitem[{van~de Geer(2000)}]{vandegeer2000empirical}
	\textsc{van~de Geer, S.} (2000).
	\newblock \textit{Empirical Processes in M-Estimation}.
	\newblock Cambridge University Press.
	
	\bibitem[{van~de Geer(2010)}]{geer2010lasso}
	\textsc{van~de Geer, S.} (2010).
	\newblock \textit{The Lasso with within group structure}, vol.~7 of \textit{IMS
		Collections}.
	\newblock Beachwood, Ohio, USA: Institute of Mathematical Statistics, pp.
	235--244.
	
	\bibitem[{van~de Geer \& B{\"u}hlmann(2009)}]{geer2009conditions}
	\textsc{van~de Geer, S.} \& \textsc{B{\"u}hlmann, P.} (2009).
	\newblock On the conditions used to prove oracle results for the lasso.
	\newblock \textit{Electronic Journal of Statistics} \textbf{3}, 1360--1392.
	
	
	\bibitem[{Wahba(1990)}]{wahba1990spline}
	\textsc{Wahba, G.} (1990).
	\newblock \textit{Spline Models for Observational Data}.
	\newblock SIAM.
	
	\bibitem[{Watson(1964)}]{watson1964smooth}
	\textsc{Watson, G.~S.} (1964).
	\newblock Smooth regression analysis.
	\newblock \textit{Sankhy{\=a}: The Indian Journal of Statistics, Series A}
	\textbf{26}, 359--372.
	
	\bibitem[{Yang \& Barron(1999)}]{yang1999information}
	\textsc{Yang, Y.} \& \textsc{Barron, A.} (1999).
	\newblock Information-theoretic determination of minimax rates of convergence.
	\newblock \textit{The Annals of Statistics} \textbf{27}, 1564--1599.
	
	\bibitem[{Yuan \& Lin(2006)}]{yuan2006model}
	\textsc{Yuan, M.} \& \textsc{Lin, Y.} (2006).
	\newblock Model selection and estimation in regression with grouped variables.
	\newblock \textit{Journal of the Royal Statistical Society: Series B
		(Statistical Methodology)} \textbf{68}, 49--67.
	
	\bibitem[{Yuan \& Zhou(2015)}]{yuan2015minimax}
	\textsc{Yuan, M.} \& \textsc{Zhou, D.-X.} (2015).
	\newblock Minimax optimal rates of estimation in high dimensional additive
	models: Universal phase transition.
	\newblock \textit{arXiv preprint arXiv:1503.02817} .
	
	\bibitem[{Zhao et~al.(2014)Zhao, Li, Liu \& Roeder}]{zhao2014SAM}
	\textsc{Zhao, T.}, \textsc{Li, X.}, \textsc{Liu, H.} \& \textsc{Roeder, K.}
	(2014).
	\newblock \textit{SAM: Sparse Additive Modelling}.
	\newblock R package version 1.0.5.
	
	\bibitem[{Zou \& Hastie(2005)}]{zou2005regularization}
	\textsc{Zou, H.} \& \textsc{Hastie, T.~J.} (2005).
	\newblock Regularization and variable selection via the elastic net.
	\newblock \textit{Journal of the Royal Statistical Society: Series B
		(Statistical Methodology)} \textbf{67}, 301--320.
	
	
	
\end{thebibliography}

\appendix
\numberwithin{equation}{section}
\numberwithin{algocf}{section}
\numberwithin{figure}{section}

\section{Additional figures}
\label{sec:AddFigures}
In this appendix we present some additional figures referenced in Sections~\ref{sec:Methodology},  \ref{sec:SimulationStudy} and \ref{sec:DataAnalysis} of the main manuscript. 

\vspace{5mm}

\noindent
Figure~\ref{fig:ExamplePlots}, shows examples of some fitted models for a fixed value of degrees of freedom. Our method seems to perform very well and is mostly robust to changes in the value of $m$. The smoothing splines estimates are unable to do as well for the same effective degrees of freedom. The plots in the bottom panel of Figure~\ref{fig:ExamplePlots} also suggest that first order trend filter can perform poorly in presence of model misspecification.

\vspace{5mm}

\noindent
In Figure~\ref{fig:AdditiveExamplePlots}, we show some of the fitted functions for both \citeauthor{ravikumar2009sparse}'s method and our proposal using the $\lambda$ value which minimizes the test set error for \citeauthor{ravikumar2009sparse}'s proposal with $3$ and $10$ basis functions.

\vspace{5mm}

\noindent
In Figure~\ref{fig:BasisPlots}, we show examples of basis functions which possess a natural hierarchy. Our proposal is specifically suited for such systems of hierarchical basis functions.

\begin{figure}
	\centering
	\includegraphics[scale = 0.52]{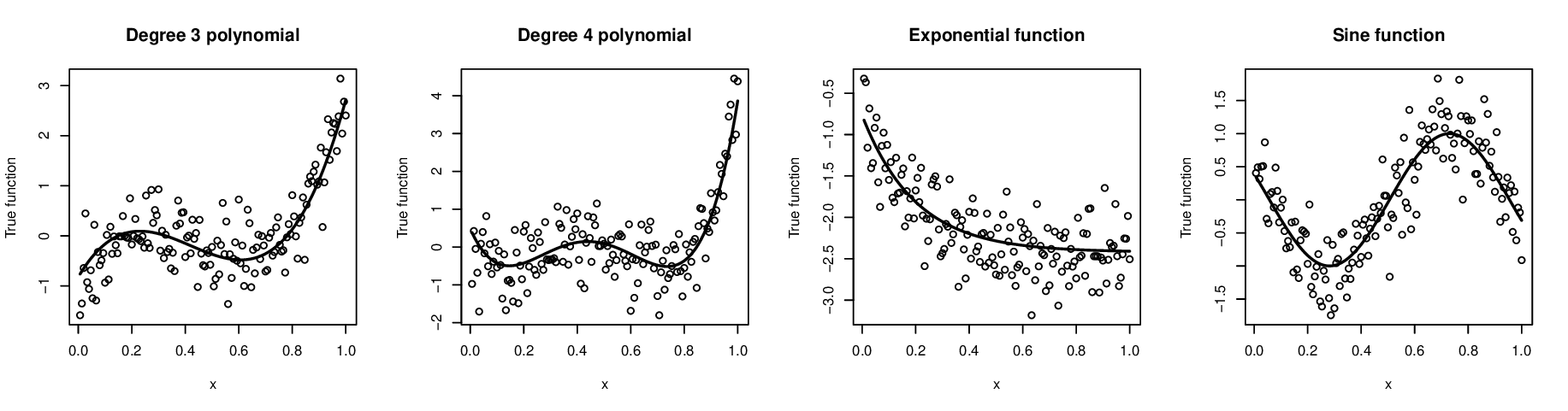}
	\includegraphics[scale = 0.52]{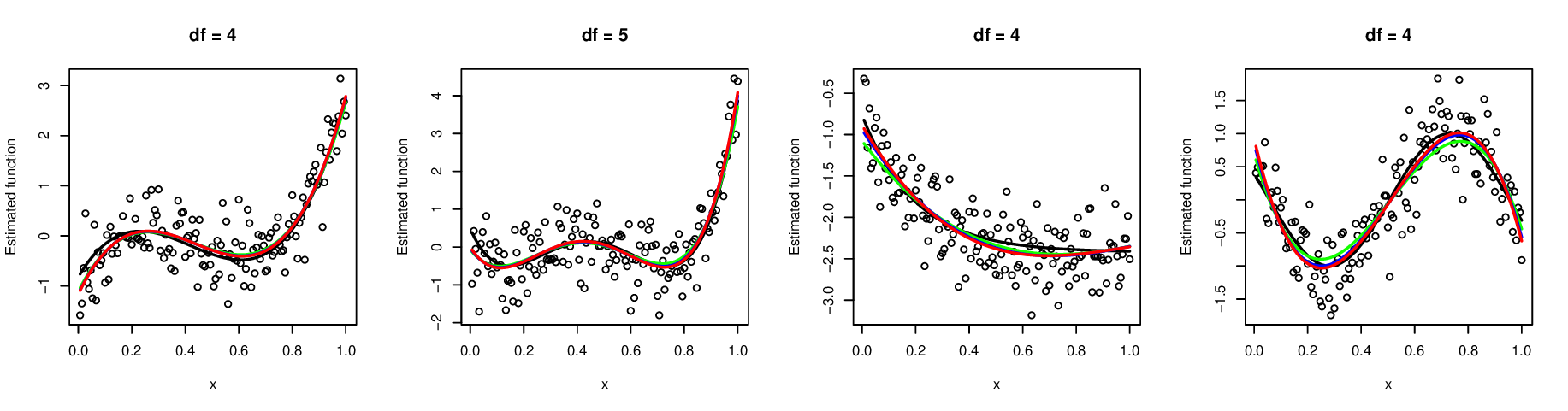}
	\includegraphics[scale = 0.52]{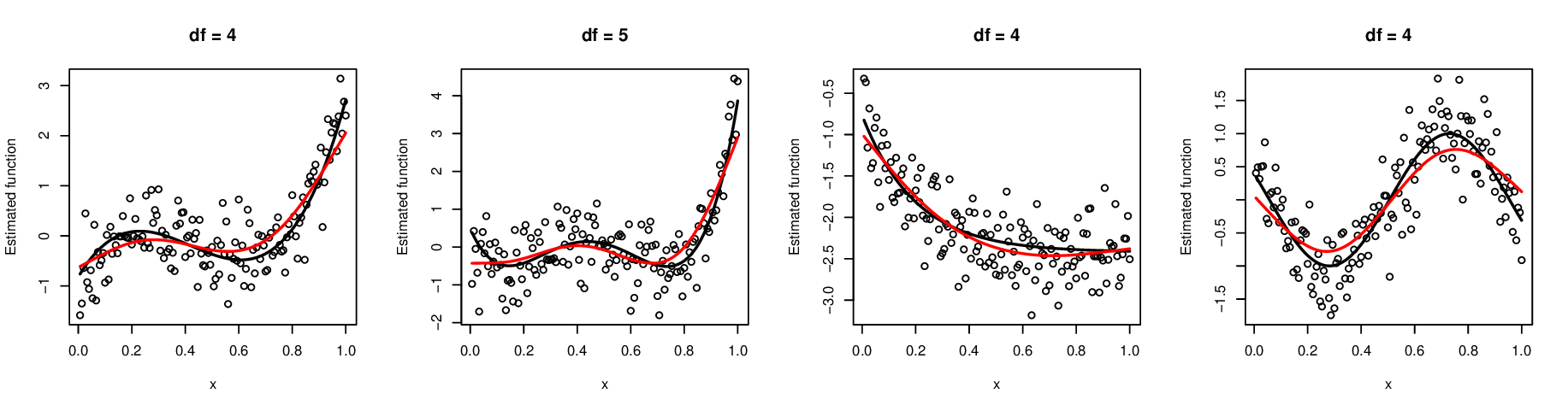}
	\includegraphics[scale = 0.52]{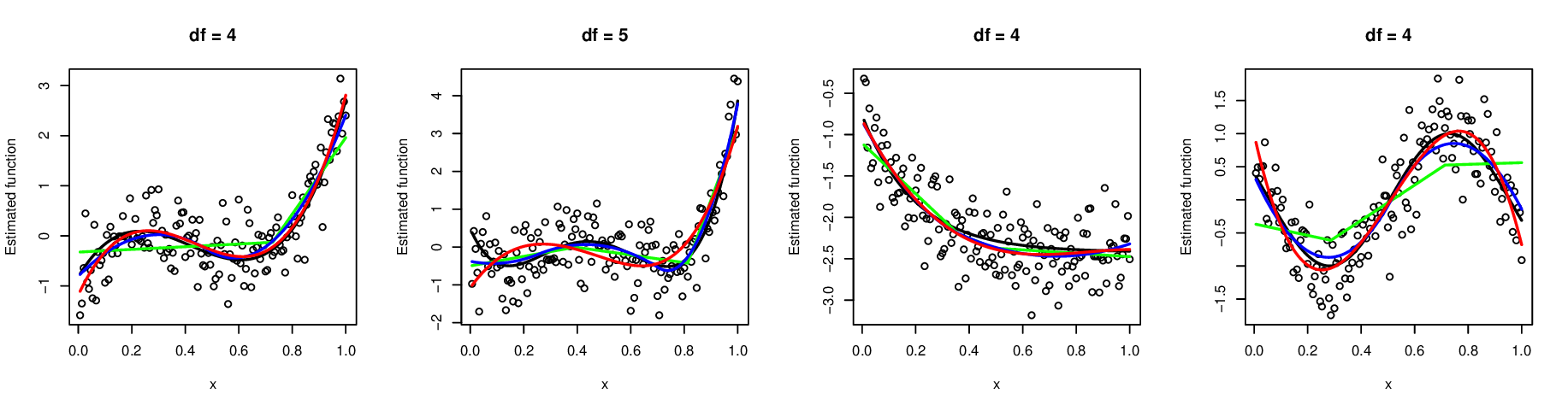}
	\caption{Top row: Scatter plots of one simulated data along with the true functions~(\protect\includegraphics[height=0.5em]{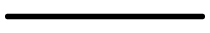}). Second row: The first row figures with estimated functions for our proposal with $m=1$~(\protect\includegraphics[height=0.5em]{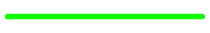}), $2$~(\protect\includegraphics[height=0.5em]{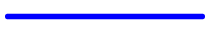}), and 3~(\protect\includegraphics[height=0.5em]{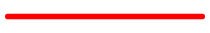}). Third row: The first row figures with estimated functions via smoothing splines~(\protect\includegraphics[height=0.5em]{NA_red_1.eps}). Fourth row: The first row figures with estimated functions for trend filtering of order 1~(\protect\includegraphics[height=0.5em]{NA_green_1.eps}), 2~(\protect\includegraphics[height=0.5em]{NA_blue_1.eps}), and 3~(\protect\includegraphics[height=0.5em]{NA_red_1.eps}).}
	\label{fig:ExamplePlots}
\end{figure}
\begin{figure}
	\centering
	\includegraphics[width= \textwidth]{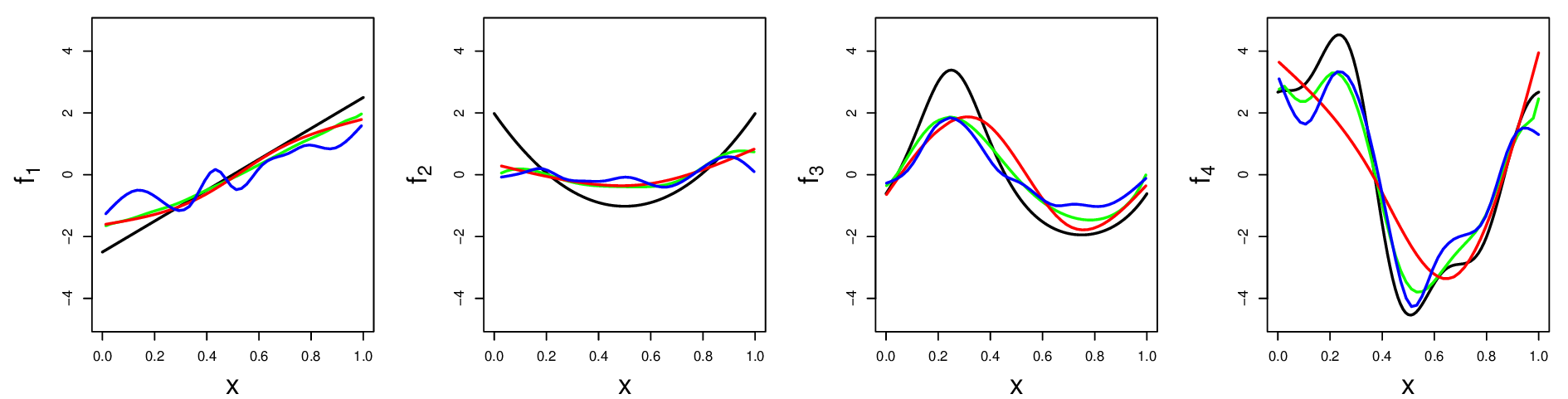}
	\caption{The first 4 component functions of the simulation study from Section~\ref{sec:SimulationAdditive}. We show the estimates of our proposal~(\protect\includegraphics[height=0.5em]{NA_green_1.eps}), and that of \cite{ravikumar2009sparse} fitted with 10~(\protect\includegraphics[height=0.5em]{NA_blue_1.eps}) and  3~(\protect\includegraphics[height=0.5em]{NA_red_1.eps}) basis functions. In each case, the tuning parameter leading to the smallest mean square error was used.}
	\label{fig:AdditiveExamplePlots}
\end{figure}
\begin{figure}
	\centering
	\includegraphics[width=\textwidth]{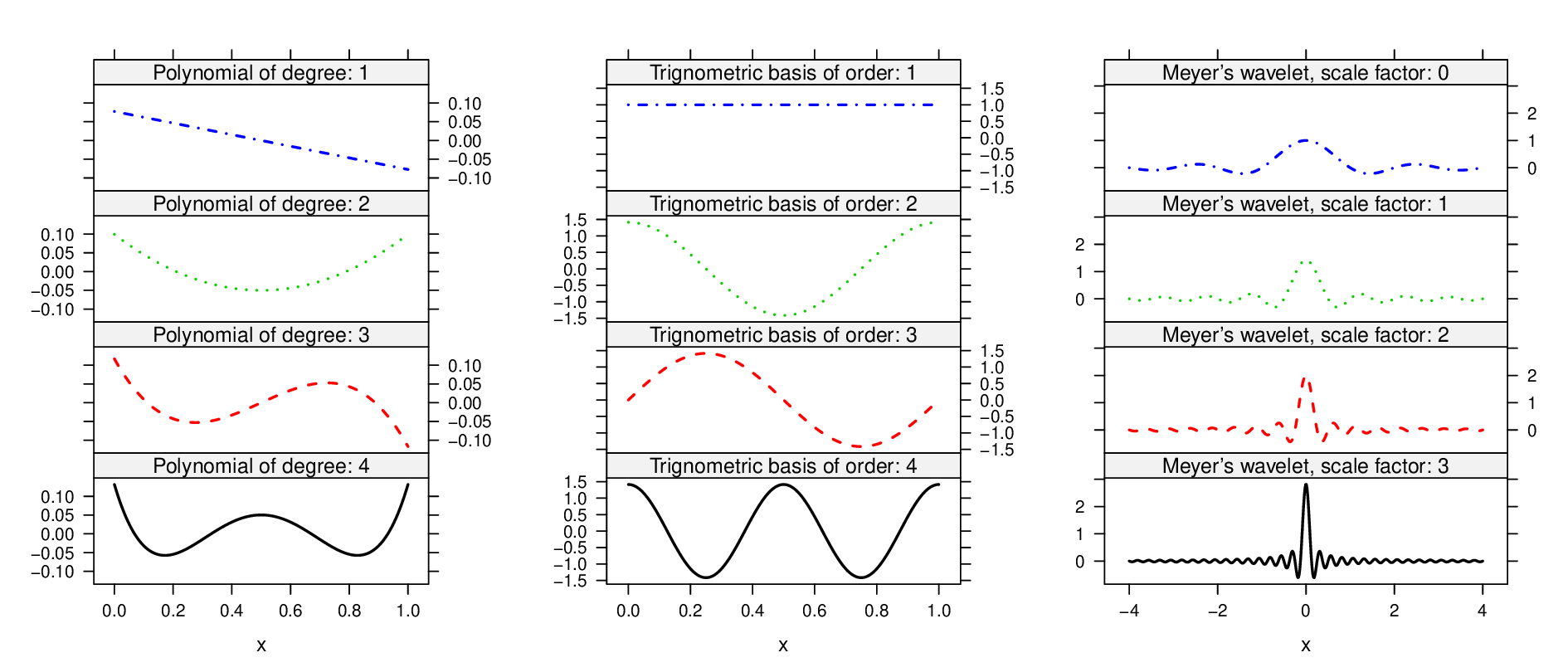}
	\caption{Examples of basis functions with natural hierarchical complexity; we plot $\psi_1(x)$~(\protect\includegraphics[height=0.5em]{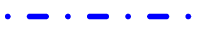}),
		$\psi_2(x)$~(\protect\includegraphics[height=0.5em]{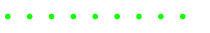}),
		$\psi_3(x)$~(\protect\includegraphics[height=0.5em]{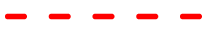}),
		$\psi_4(x)$~(\protect\includegraphics[height=0.5em]{NA_black_1.eps}). Polynomial, trigonometric and wavelet basis functions are shown in the left, middle and right panels, respectively.}
	\label{fig:BasisPlots}
\end{figure}

\section{Algorithms for additive framework and  extension to classification}
\label{app:additiveAlgorithm}
Here we give an algorithm for our additive and sparse additive framework as well as an algorithm for the extension of our proposal to classification. We use a block coordinate descent algorithm for solving the additive and sparse additive proposal. This algorithm cyclically iterates through features, and for each feature applies the univariate solution detailed in Algorithm~\ref{alg:univariate} of the main manuscript. The exact details are given in Algorithm~\ref{alg:additive} below.

\begin{algorithm}
	\caption{Block coordinate descent for the additive and sparse additive framework}
	\label{alg:additive}
	\begin{tabbing}
		\enspace Initialize ${\beta}_j \gets 0$ for $j=1,\ldots,p$\\
		\enspace While $l\le max\_iter$ \textbf{ and } not converged\\
		
		\qquad For $ j = 1,\, \ldots,\, p$\\
		\qquad \qquad Set 
		$ {r}_{-j} \gets {y} - \sum_{j^{'} \neq j}{\Psi}_{\K}^{j^{'}}{{\beta}}_{j^{'}} $\\
		\qquad \qquad Update ${\beta}_j \gets \underset{{\beta} \in \mathbb{R}^{\K}}{\arg\min} \frac{1}{2} \left\| {r}_{-j} - {\Psi}_{\K}^j {\beta} \right\|_n^2 + \lambda^2 \sum_{k=1}^{\K} \wt{w}_{k} \left\|{\Psi}^j_{k:\K} {\beta}_{j,k:\K} \right\|_n,$\\ 
		\qquad \qquad \qquad where $\wt{w}_1 = w_1+\lambda^{-1}$ and $\wt{w}_k = w_k$ for $k=2,\,\ldots,\, \K$.\\
		
		\enspace Return ${\beta}_1,\, \ldots,\, {\beta}_p$ 
	\end{tabbing}
\end{algorithm}
We also give an algorithm for the extension of our method to classification based on proximal gradient descent. To begin let $L(\beta_0,{\beta}) = {1}/({2n}) \sum_{i=1}^n \log\left( 1 + \exp\left[ -y_i\left\{\beta_0 + \left(\Psi {\beta}\right)_i\right\}  \right] \right)$. We denote by $\nabla L(\beta_0,{\beta})$, the derivative of $L$ at the point $(\beta_0,\, {\beta})\in \mathbb{R}^{\K+1}$. Algorithm~\ref{alg:univariateLogistic} presents the steps for solving \eqref{eqn:LogisticHierBasis}. The algorithm for extension of additive models to classification can be similarly derived and is omitted in the interest of brevity. 
\begin{algorithm}
	\caption{Proximal gradient descent for extension to classification}
	\label{alg:univariateLogistic}
	\enspace Initialize $(\beta_0^0, {\beta}^0)$ \\
	\enspace For $ l = 1,\, 2,\,\ldots$ until convergence\\
	\qquad Select a step size $t_l$ via line search \\
	\qquad Update $$(\beta_0^{l},\, {\beta}^{l}) \gets \underset{(\beta_0,{\beta}) \in \mathbb{R}^{\K+1}}{\arg\min} \frac{1}{2} \left\| (\beta_0,{\beta})  - \left\{ (\beta_0^{l-1},{\beta}^{l-1}) - t_{l}\nabla L(\beta_0^{l-1}, {\beta}^{l-1}) \right\}\right\|^2_2 + \lambda\Om({\beta}).$$\\
	
	\enspace Return $(\beta_0^{l}, {\beta}^{l})$ 
\end{algorithm}

\section{Proofs for Section~\ref{sec:entropyResultsHB}}
\label{app:proofsUnivariate}

\begin{proof}[of Lemma~\ref{lemma:reductionEuclidean}]
	Firstly, we have for $f_{\beta^{[1]}}(x) = \sum_{k=1}^{\K_n}\psi(x)\beta_k^{[1]},\, f_{\beta^{[2]}}(x) = \sum_{k=1}^{\K_n}\psi(x)\beta_k^{[2]} \in \mathcal{F}_n$
	\begin{align*}
	\|f_{\beta^{[1]}}-f_{\beta^{[1]}}\|_Q^2 &= \int (f_{\beta^{[1]}}-f_{\beta^{[1]}})^2dQ = \int \left\{ \sum_{k=1}^{\K_n}\psi_k(x)(\beta^{[1]}_k - \beta^{[2]}_k)\right\}^2dQ \\
	&= \int \left\{\sum_{k=1}^{\K_n} \psi_k^2(x)\left(\beta^{[1]}_k - \beta^{[2]}_k\right)^2 + \sum_{k\not=l}\psi_k(x)\psi_l(x)\left(\beta^{[1]}_k - \beta^{[2]}_k\right)\left(\beta^{[1]}_l - \beta^{[2]}_l\right) \right\}dQ\\
	&= \left\|\beta^{[1]} - \beta^{[2]} \right\|_2^2,
	\end{align*}
	where the final equality follows due to the orthonormality of $\psi_k$. Similarly for $f_{\beta^{[1]}},f_{\beta^{[1]}}\in \mathcal{F}_{p,n}$ we can show that $\|f_{\beta^{[1]}}-f_{\beta^{[1]}}\|_Q^2 = \|{\beta}^{[1]}-{\beta}^{[2]}\|_2^2$. Thus if $\{{\beta}^1,\ldots, {\beta}^N \}$ is the smallest $\de$-cover of $H^{w/M}_{\K_n}$ then the functions $\{ f_{\beta^1},\ldots, f_{\beta^N} \}$ form the smallest $\de$-cover with respect to the $L_Q$ norm. This can be extended to the case $n = \infty$. This proves the first part.

	Secondly, note that for $f_{\beta^{[1]}},f_{\beta^{[2]}}\in \mathcal{F}_n$ 
	\begin{equation*}
	\|f_{\beta^{[1]}}-f_{\beta^{[2]}}\|_n^2 = ({\beta}^{[1]}-{\beta}^{[2]})^{\top}\frac{ {\Psi}_{\K_n}^{\top} {\Psi}_{\K_{n}} }{n}({\beta}^{[1]}-{\beta}^{[2]}) \le \Lambda_{\max}\|{\beta}^{[1]}-{\beta}^{[2]}\|_2^2,
	\end{equation*}
	thus if  $\{{\beta}^1,\ldots, {\beta}^N \}$ is the smallest $\delta$-cover for $\mathcal{H}_{\K_n}^w$, then $\{f_{\beta^{1}},\ldots, f_{\beta^{N}}\}$ is a $\Lambda_{\max}^{1/2}\de$ cover of $\{f_{\beta} \in \mathcal{F}_n :\sum_{k=1}^{\K_n}w_k\|{\Psi}_{k:\K_n}{\beta}_{k:\K_n}\| \le 1\}$ with respect to the $Q_n$ metric. Since this is a cover and not the smallest cover, we have  
	$$H\Big[ \Lambda_{\max}^{1/2} \de, \{f_{\beta} \in \mathcal{F}_n :\sum_{k=1}^{\K_n}w_k\|{\Psi}_{k:\K_n}{\beta}_{k:\K_n}\| \le 1\}, Q_n\Big] \le H(\de, \mathcal{H}_{\K_n}^w),$$
	and since the inequality holds for all $\de>0$, we can select $\de = \de'\Lambda_{\max}^{-1/2}$ giving us the result. 
	
	For the multivariate case we can repeat the same argument as above replacing $\mathcal{F}_n$ by $\mathcal{F}_{p,n}$.
\end{proof}

\begin{proof}[ of Lemma~\ref{lemma:upperBoundEntropy}]
	For the Ellipsoid $E_{\K_n}^w$ where
	\begin{equation}
	E^w_{\K_n} = \left\{ {\beta}\in \mathbb{R}^{\K_n}: \sum_{k=1}^{\K_n} \beta_k^2\left(w_1+\cdots+w_k\right)^2\le 1 \right\}\ ,
	\label{eqn:ellipsoidRegion}
	\end{equation}
	we show that $\mathcal{H}_{\K_n}^w\subset E_{\K_n}^w$ in Lemma~\ref{app:BoundRegion}\ref{lemma:upBound} of Appendix~\ref{app:BoundRegion}. \cite{dumer2006covering} proved an upper bound for ellipsoids which we state in Appendix~\ref{sec:upperBoundEntropy}. For the special case of $w_k = k^m-(k-1)^m$, this theorem yields the desired upper bound as shown in Corollary~\ref{sec:entropyResults}\ref{lemma:UpperBoundSobolev}. Therefore we have $H(\de,\, \mathcal{H}_{\K_n}^w) \le H(\de,\, E_{\K_n}^w) \le U_{E,1}\de^{-{1}/{m}}$.
	
	Similarly, we can consider the special case of our  multivariate framework weights in Corollary~\ref{sec:entropyResults}\ref{lemma:UpperBoundMultivariateHB}, which gives us the result $H(\de,\, \mathcal{H}_{\K_n}^w) \le H(\de,\, E_{\K_n}^w) \le U_{E,2}\de^{-{p}/{m}}$.
\end{proof}

\begin{proof}[ of Lemma~\ref{lemma:lowerBoundEntropy}]
	Let $d$ be the integer such that $(w_1+\cdots+w_{d+1})^{-1} \le \delta\le (w_1+\cdots+w_{d})^{-1}$ for $\delta\in ((w_1+ \cdots +  w_{\K_n+1})^{-1}, 1)$. Note that since $\delta\ge (w_1+\cdots+w_{\K_n+1})^{-1}$, $d\le \K_n$. We define the truncated region as
	\begin{equation*}
	\wt{H}_d^w = 
	\left\{ {\beta}\in \mathcal{H}_{\K_n}^w: \beta_j = 0\ \text{ for all } \ j\ge d+1 \right\} .
	\label{eqn:trunHB2}
	\end{equation*}
	Then we have that $\mathcal{H}_d^w\subset \wt{\mathcal{H}}_{d}^w\subseteq \mathcal{H}_{\K_n}^w$ where $\mathcal{H}_d^w$ is simply viewing $\wt{\mathcal{H}}_{d}^w$ as a subset of $\mathbb{R}^{d}$. Let $\mathbb{B}_n(r)$ be the $n$-ball of radius $r$. By Lemma~\ref{app:BoundRegion}\ref{lemma:lowBound}, we have $\mathbb{B}_{d}\left\{ (w_1+\cdots+w_d)^{-1} \right\}  \subset \mathcal{H}_d^w$. The lower bound of the entropy of a ball can be obtained by a simple volume argument. Since $(w_1+\cdots+w_d)^{-1} \ge \delta$ then $\mathbb{B}_d(\de)\subseteq \mathbb{B}_{d}\{(w_1+\cdots+w_d)^{-1}\} $ and hence
	\begin{align*}
	H(\delta/2, \mathcal{H}_d^w  )\ge  H\{\delta/2, \mathbb{B}_{d}\left( \de \right)  \} \ge \log \frac{Vol\{ \mathbb{B}_{d}\left( \de\right)\} }{Vol\{\mathbb{B}_{d}(\delta/2)\} } =  d\log(2).
	\end{align*}
	Since the above inequality holds for $\delta\le 1 $, for $\delta\in ((w_1+ \cdots +  w_{\K_n + 1})^{-1},\, 1/2)$ we have 
	$H(\de, \mathcal{H}_d^w)\ge d\log 2$.
	\newline

	\noindent Now for the univariate case we have $(w_1+\cdots + w_{d+1})^{-1} = (d+1)^{-m}\le \delta$ or $(d+1)\ge \de^{-{1}/{m}}$ and hence we have 
	\begin{align*}
	H(\de, \mathcal{H}_d^w)\ge d\log 2 \ge (\de^{-{1}/{m}}-1)\log 2 = \de^{-{1}/{m}}\left( 1- \de^{1/m} \right)\log 2\ge \de^{-{1}/{m}}\left( 1- 2^{-1/m} \right)\log 2.
	\end{align*}
	
	Now for the multivariate case, the argument is slightly different due to presence of zero weights. As before, there is some $d'$ such that $(w_1+\cdots+w_{q_{d'}-1})^{-1} \le \delta \le  (w_1+\cdots+w_{q_{d'} })^{-1}$ and hence $d = q_{d'}-1$. Note that by assumption we have $\K_n= q_{\K'}-1$ and hence $\delta \ge (w_1+\cdots+w_{q_{\K'}})^{-1}$ which implies that $d'\le \K'$ and hence $d\le \K_n$. Finally we have that since $w_1+\cdots+w_{q_{d'}-1} = w_1+\cdots+w_{q_{d'-1}} = (d'-1)^m$, therefore $d'-1\ge \de^{-{1}/{m}}$. Now we have that 
	\begin{align*}
	H(\de, \mathcal{H}_d^w)&\ge d\log(2) = (q_{d'}-1)\log(2)= \left\{\binom{d'+p-1}{p}-1\right\}\log(2) \\
	&\ge \left\{\frac{(d'+p-1)^p}{p^p}-1\right\}\log(2)\ge \left\{\frac{(\de^{-{1}/{m}}+p)^p}{p^p}-1\right\}\log(2) \\
	&= \de^{-{p}/{m}} \underbrace{ \left\{ \frac{(1+p\de^{1/m})^p}{p^p} - \de^{{p}/{m}}\right\} }_{g(\de)}\log(2) \ge \de^{-{p}/{m}}A\log(2),
	\end{align*}
	where the last inequality follows from the fact that $g^{\prime}(\de)>0$ for all $\de\in (0,1)$.
\end{proof}

\section{Details for Proposition~\ref{lemma:upperBoundHB}}
\label{app:detailsUnivariate}
\subsection{Univariate case}

Firstly, if $f^0(x) = \sum_{k=1}^{\infty}\psi_k(x)\beta^0_{k}$ then we select $f_n^*(x) = \sum_{k=1}^{\K_n}\psi_k(x)\beta^0_{k} \in \mathcal{F}_n$. Secondly, we note that for the univariate estimator we have $\Om(f_n^*\mid Q_n) = \Om^{\text{uni}}({\beta}^0_{1:\K_n})$. For brevity we will drop the dependence on ${\beta}^0$ and denote $\Om^{\text{uni} }({\beta}^0_{1:\K_n})$ by  $\Om$. Thus we have
\begin{align*}
\lambda_n^2\Om(f_n^*\mid Q_n) &=  n^{-{2}/{(2+\alpha)}}\Om^{-{(2-\alpha)}/{(2+\alpha)}} \Om = n^{-{2}/{(2+\alpha)}} \Om^{{2\alpha}/{(2+\alpha)}} = n^{-{2m}/{(2m+1)}}\Om^{{2}/{(2m+1)}},
\end{align*}
where we use the fact that for our class $\alpha = 1/m$.
For the term $\Om({\beta}^0_{1:\K_n})$ we have 
\begin{align*}
\Om({\beta}^0_{1:\K_n}) &= \sum_{k=1}^{\K_n}w_k \Big\{\left({\beta}^0_{1:\K_n}\right)^{\top} \frac{{\Psi}_{ k:\K_n}^{\top} {\Psi}_{k:\K_n}}{n} {\beta}^0_{1:\K_n} \Big\}^{1/2} \le \Lambda^{1/2}_{\max} \sum_{k=1}^{\K_n}w_k\|{\beta}^0_{1:\K_n}\|_2 \le  \Lambda^{1/2}_{\max} M,
\end{align*}
for $f^0\in \mathcal{F}_{\infty}^M$. For $\mathcal{G}_{2}^M$, we do not have the above bound and hence we keep the $\Om$ term in the inequality.

For the truncation error we note that 
\begin{align*}
\|f^0 - f_n^*\|_n^2 &= \frac{1}{n}\sum_{i=1}^{n}\left\{ \sum_{k=\K_n+1}^{\infty} \psi_k(x_i)\beta^0_{k} \right\}^2 \\
&\le \psi_{\max}^2\frac{1}{n}\sum_{i=1}^{n}\left( \sum_{k=\K_n+1}^{\infty} \beta^0_{k} \right)^2 = \psi_{\max}^2 \left(\sum_{k=\K_n+1}^{\infty}\beta_k^0\right)^2 \\ 
&= \psi_{\max}^2\left( \sum_{k=\K_n+1}^{\infty} \frac{k^m}{k^m}|\beta^0_{k}| \right)^2\\
&\le \psi_{\max}^2\left\{\sum_{k=\K_n+1}^{\infty} k^{2m}(\beta^0_{k})^2\right\} \left( \sum_{k=\K_n+1}^{\infty}\frac{1}{k^{2m}} \right),\\
&\le 
= \psi_{\max}^2M^2\sum_{k=\K_n+1}^{\infty} {k^{-2m}},
\end{align*}
where the last inequality follows from the proof of Lemma~\ref{app:BoundRegion}\ref{lemma:upBound}. The result now follows since 
\begin{align*}
\sum_{k=\K_n+1}^{\infty}k^{-2m} \le \left\{(2m-1)(\K_n+1)^{2m-1}\right\}^{-1} \le \frac{1}{2m-1}\frac{1}{\K_n^{2m-1}}.
\end{align*}

\subsection{Multivariate case}
Now we assume that $f^0({x}) = \sum_{k=1}^{\infty} \psi_k({x}^{{\nu_k}})\beta^0_{k}$ for ${x}\in \mathbb{R}^p$ and $\nu_k \in \mathbb{Z}_+^p$. Then we take $f_n^*({x}) = \sum_{k=1}^{\K_n} \psi_k({x}^{{\nu_k}}) \beta^0_{k}$. Now by the same calculations as in the univariate case, we have 
\begin{align*}
\lambda_n^2\Om(f_n^*\mid Q_n) &=  n^{-{2m}/{(2m+p)}} \Om^{{2p}/{(2m+p)}} \le  n^{-{2m}/{(2m+p)}} ({\Lambda^{1/2}_{\max}} M)^{{2p}/{(2m+p)}}.
\end{align*}

For the truncation error we note that $\K_n = q_{\K'}-1$ and hence
\begin{align*}
\|f^0 - f_n^*\|_n^2 &= \frac{1}{n}\sum_{i=1}^{n}\left\{ \sum_{k=\K_n+1}^{\infty} \psi_k({x_i}^{{\nu_k}})\beta^0_{k} \right\}^2 \le \psi_{\max}^2\frac{1}{n}\sum_{i=1}^{n}\left( \sum_{k=q_{\K'}}^{\infty} \beta^0_{k} \right)^2\\
&= \psi_{\max}^2\frac{1}{n}\sum_{i=1}^{n}\left\{ \sum_{k:\|{\nu_k}\|_1= \K'} \frac{\K'^m}{\K'^m}|\beta^0_{k}| + \sum_{k:\|{\nu_k}\|_1= \K'+1} \frac{(\K'+1)^m}{(\K'+1)^m}|\beta^0_{k}|+\cdots\right\}^2\\
&= \psi_{\max}^2\frac{1}{n}\sum_{i=1}^{n}\left( \sum_{R=\K'}^{\infty} \frac{R^m}{R^m} \sum_{k:\|{\nu_k}\|_1= R} |\beta^0_{k}| \right)^2\\
&= \psi_{\max}^2\frac{1}{n}\sum_{i=1}^{n}\left\{\left( \sum_{R=\K'}^{\infty}\frac{1}{R^{2m}} \right)^{1/2} \left(\sum_{R=K'}^{\infty} R^m \sum_{k:\| {\nu_k} \|_1= R } |\beta^0_{k}| \right)^{1/2} \right\}^2\\
&\le \psi_{\max}^2M^2\sum_{R = \K'}^{\infty}\frac{1}{R^{2m}} \le \frac{M^2\psi_{\max}^2}{2m-1}\frac{1}{(\K')^{2m-1}}.
\end{align*}

\section{Proof of Theorem~\ref{thm:thm1}}
\label{app:proofsAdditive}
\subsection{Initial results}

Recall that $(\wh{f}_j)_{j=1}^p \in \mathcal{F}$ where $\mathcal{F}$ is some arbitrary univariate function class 
. We denote the functions $\wh{f}({x}) = \sum_{j=1}^{p}\wh{f}_j({x}_j)$ and $f^0({x}) = \sum_{j=1}^{p}f^0_j({x}_j)$ for ${x} = (x_1,\ldots,x_p)^{\top}\in \mathbb{R}^p$. For the proof of Theorem~\ref{thm:thm1}, $\lambda_n$ and $\rho_n$ are functions of $n$ but for convenience we will simply write $\lambda,\, \rho$. Throughout this proof, instead of the smoothness level $m$, we will use $\alpha = 1/m$. Thus the entropy condition is $H[\delta, \{f\in \mathcal{F} : \up(f)\le 1\}, Q_n] \le A_0\delta^{-\alpha},$ for $\alpha\in (0,2)$, and so forth.

We begin the proof of Theorem~\ref{thm:thm1} with a basic inequality.

\begin{lemma}[Basic inequality]
	\label{lemma:basic}
	For any function $f^* = \sum_{j=1}^{p}f_j^*$, where $f_j^*\in \mathcal{F}$ and, the solution $\wh{f}$ of (\ref{eqn:sparseAdditiveOptim}), we have the following basic inequality
	\begin{equation*}
	\frac{1}{2}\|\wh{f} - f^0\|_n^2 + \lambda I_p(\wh{f}) \le  |\langle \e,\, \wh{f} - f^* \rangle_n| + \lambda I_p(f^*) + |\bar{\e}|\sum_{j=1}^{p}\|\wh{f}_j - f^*_j\|_n  +  \frac{1}{2}\|f^* - f^0\|_n^2,
	\end{equation*}
	where $\langle \e, f\rangle_n = {n}^{-1}\sum_{i=1}^{n}\e_if({x_i})$, $\bar{\e} = {n}^{-1}\sum_{i=1}^{n}\e_i$ and $I_p(f) = \sum_{j=1}^{p}I(f_j) = \sum_{j=1}^{p}\|f_j\|_n + \lambda\up(f_j)$ for an additive function $f$.
\end{lemma}

\begin{proof}
	We have 
	\begin{align*}
	\frac{1}{2n}\sum_{i=1}^{n} \left\{ y_i - \bar{y} - \wh{f}({x_{i}}) \right\}^2 + \lambda I_p(\wh{f}) \le \frac{1}{2n}\sum_{i=1}^{n} \left\{ y_i - \bar{y} - f^*({x_{i}}) \right\}^2 + \lambda I_p(f^*_j),
	\end{align*}
	which is equivalent to
	\begin{align*}
	& \frac{1}{2n}\sum_{i=1}^{n} \left\{ \e_i + c^0 - \bar{y} -(\wh{f} - f^0)({x_i})  \right\}^2 + \lambda I_p(\wh{f}) \le \frac{1}{2n}\sum_{i=1}^{n} \left\{ \e_i + c^0 - \bar{y}-(f^*-f^0)({x_i})  \right\}^2 + \lambda I_p(f^*_j).
	\end{align*}
	This gives us
	\begin{align*}
	& \frac{1}{2n}\sum_{i=1}^{n} \left( \e_i + c^0 - \bar{y}\right)^2 +(\wh{f} - f^0)^2({x_i}) -2(\e_i+c^0-\bar{y})(\wh{f} - f^0)({x_i}) + \lambda I_p(\wh{f}) \\
	&\le \frac{1}{2n}\sum_{i=1}^{n} \left( \e_i+c^0 - \bar{y}  \right)^2 + (f^*-f^0)^2({x_i})-2(\e_i+c^0-\bar{y})(f^*-f^0)({x_i}) + \lambda I_p(f^*).
	\end{align*}
	Re-arranging the terms and simplifying gives us
	\begin{align*}
	& \frac{1}{2}\|\wh{f} - f^0\|_n^2 - \langle \e + c^0 - \bar{y},\, \wh{f}-f^0\rangle_n + \lambda I_p(\wh{f}) \\
	&\le   \frac{1}{2}\|f^*-f^0\|_n^2 - \langle \e + c^0 - \bar{y},\, f^*-\wh{f}\rangle_n - \langle \e+c^0-\bar{y},\,\wh{f}-f^0 \rangle_n +\lambda I_p(f^*),
	\end{align*}
	which implies
	\begin{align*}
	&\frac{1}{2}\|\wh{f} - f^0\|_n^2  + \lambda I_p(\wh{f}) \le   \frac{1}{2}\|f^*-f^0\|_n^2 - \langle \e + c^0 - \bar{y},\, f^*-\wh{f}\rangle_n  +\lambda I_p(f^*),
	\end{align*}
	and finally gives us
	\begin{align*}
	&\frac{1}{2}\|\wh{f} - f^0\|_n^2  + \lambda I_p(\wh{f}) \le  |\langle \e,\, \wh{f} - f^* \rangle_n|+ | c^0-\bar{y}| \sum_{j=1}^{p}\|\wh{f}_j - f^*_j\|_n +\lambda I_p(f^*) +  \frac{1}{2}\|f^* - f^0\|_n^2.
	\end{align*}
	Now, for the second term
	\begin{align*}
	|c^0 - \bar{y}| &= |E\bar{y} - \bar{y}| = \left| \frac{1}{n}\sum_{i=1}^{n}(Ey_i - y_i)\right| = |\bar{\e}|,
	\end{align*}
	which leads us to 
	\begin{equation*}
	\frac{1}{2}\|\wh{f} - f^0\|_n^2 + \lambda I_p(\wh{f}) \le  |\langle \e,\, \wh{f} - f^* \rangle_n| + \lambda I_p(f^*) + |\bar{\e}|\sum_{j=1}^{p}\|\wh{f}_j - f^*_j\|_n  +  \frac{1}{2}\|f^* - f^0\|_n^2.
	\end{equation*}
\end{proof}

\begin{lemma}[Bounding the term $|\bar{\e}|$]
	For $\e = (\e_1,\ldots, \e_n)^{\top}$ such that $E(\e_i) = 0$ and
	\begin{equation*}
	L^2\left\{ {E}\Big(e^{\e_i^2/L^2}\Big) -1 \right\}\le \sigma_0^2\ ,
	\end{equation*}
	for all $\kappa>0$ and
	\begin{equation*}
	\rho = \kappa\max\left\{ n^{-{1}/{(2+\alpha)}} , \left( \frac{\log p}{n} \right)^{1/2} \right\},
	\end{equation*}
	we have that with probability at least $1- 2\exp\left( -{n\rho^2}/{c_1} \right)$,
	\begin{equation*}
	|\bar{\e}|\le \rho ,
	\end{equation*}
	for a constant $c_1$ that depends on $L$ and $\sigma_0$.
\end{lemma}
\begin{proof}
	By Lemma 8$\cdot$2 of \cite{vandegeer2000empirical}~(with $\gamma_n = {1}_n/n$) we have for all $t> 0$
	\begin{align*}
	\text{pr}\left( \left| \frac{1}{n}\sum_{i=1}^{n}\e_i \right| \ge t \right) \le 2\exp\left\{ -\frac{nt^2}{8(L^2+\sigma_0^2)} \right\}.
	\end{align*}
	The result follows by setting $t = \rho$.
\end{proof}


\begin{lemma}[Bounding the term $|\langle{\e}, \wh{f} - f^*\rangle_n| $]
	For $\lambda\ge 4\rho$ where $$\rho = \kappa\max\left\{ n^{-1/{(2+\alpha)}} , \left( \frac{\log p}{n} \right)^{1/2} \right\},$$ for some constant $\kappa$, if
	\begin{equation*}
	H[\delta, \{f\in \mathcal{F}:\up(f)\le 1 \}, Q_n] \le A_0\delta^{-\alpha},
	\end{equation*} 
	we then have with probability at least $1-c_2\exp\left(-c_3n\rho^2\right)$
	\begin{equation*}
	|\langle \e,\wh{f}_j - f^*_j\rangle_n|\le \rho \|\wh{f}_j - f^*_j\|_n + \rho\lambda \up(\wh{f}_j - f^*_j),
	\end{equation*}
	for all $j=1,\ldots,p$ and positive constants $c_2$ and $c_3$.
\end{lemma}

\begin{proof}
	Firstly, for $\mathcal{F}_0 = \{f\in \mathcal{F}: \up(f)\le 1\}$ we have by assumption a $\delta$ cover $f_1,\ldots,f_N$ such that for all $f\in \mathcal{F}_0$ we have $
	\min_{j\in \{1,\ldots,N\}} \| f_j -  f\|_n\le \delta$. Now we are interested in the set $\mathcal{F}_{0,\lambda} = \{f\in \mathcal{F}: \lambda\up(f)\le 1\}$.
	Firstly, for a function $f\in \mathcal{F}_{0,\lambda}$,
	\begin{align*}
	\min_{j\in \{1,\ldots,N\}} \| f - f_j/\lambda\|_n = \min_{j\in \{1,\ldots,N\}} \frac{1}{\lambda}\| \lambda f - f_j\|_n\le  \frac{\delta}{\lambda},
	\end{align*}
	because $\up(\lambda f) = \lambda \up(f) \le 1$ that implies $\lambda f \in \mathcal{F}_0$. This means that the set $\{f_1/\lambda,\ldots, f_N/\lambda\}$ is a $\de/\lambda$ cover of the set $\mathcal{F}_{0,\lambda}$.
	
	Thus $H(\delta, \mathcal{F}_0,Q_n) \le A_0\delta^{-\alpha}$ which implies that $H(\delta/\lambda, \mathcal{F}_{0,\lambda},Q_n) \le A_0\delta^{-\alpha}$ or equivalently $H(\delta, \mathcal{F}_{0,\lambda}, Q_n) \le A_0(\delta\lambda)^{-\alpha}$. Finally, since $\{f\in\mathcal{F}: I(f)\le 1\} \subset \{f\in\mathcal{F}: \up(f)\le \lambda^{-1}\}$ we have 
	\begin{align*}
	H[\delta, \{f\in \mathcal{F}: I(f)\le 1 \}, Q_n] \le A_0(\delta\lambda)^{-\alpha} .
	\end{align*}

	The same entropy bound holds for the class 
	\begin{equation}
	\label{eqn:class-scaled2}
	\wt{\mathcal{F}} = \left\{\frac{f_j - f_j^*}{\|f_j - f_j^*\|_n + \lambda \Om(f_j - f_j^*)}: f_j\in \mathcal{F}\right\},
	\end{equation}
	and we can now apply Corollary 8.3 of \cite{vandegeer2000empirical} by noting that 
	\begin{align*}
	\int_{0}^{1} H^{1/2}(u, \wt{\mathcal{F}}, Q_n)\, du \le \wt{A}_0\lambda^{-\alpha/2},
	\end{align*}
	for some constant $\wt{A}_0 = \wt{A}_0(A_0)$. For some $c_2 = c_2(L,\sigma_0)$ and all
	$\delta\ge 2c_2\wt{A}_0\lambda^{-\alpha/2}n^{-1/2}$ 
	we have 
	\begin{align}
	\label{eqn:lemmaUse}
	\text{pr}\left\{ \sup_{f_j\in \mathcal{F}}  \frac{\left| \langle \e, {f}_j - f^*_j \rangle_n\right|}{ \|f_j-f_j^*\|_n + \lambda \Om(f_j-f_j^*) } \ge \delta \right\} \le c_2\exp\left( -\frac{n\delta^2}{4c_2^2} \right).
	\end{align}
	Since $\lambda \ge \rho$ we note that $2c_2\wt{A}_0\lambda^{-\alpha/2}n^{-1/2} \le 2c_2\wt{A}_0\rho^{-\alpha/2}n^{-1/2}$ and that 
	\begin{align*}
	2c_2\wt{A}_0\rho^{-\alpha/2}n^{-1/2} \le \rho\  \text{ equivalently }\  \rho \ge \left( 2c_2\wt{A}_0 \right)^{{2}/{(2+\alpha)}} n^{-{1}/{(2+\alpha)}}.
	\end{align*}
	Which holds by definition since $\rho = \kappa\max\left\{ \left({\log p}/{n}\right)^{1/2}, n^{-{1}/{(2+ \alpha)}} \right\} \ge \kappa n^{-{1}/{(2+\alpha)}}$ and $\kappa$ is sufficiently large, any $ \kappa \ge \Big( 2c_2\wt{A}_0 \Big)^{{2}/{(2+\alpha)}}$ would suffice. Therefore, we can take $\delta = \rho$ in \eqref{eqn:lemmaUse} along with a union bound to obtain 
	\begin{align*}
	\text{pr}\left\{ \max_{j=1,\ldots,p}\sup_{f_j\in \mathcal{F}}  \frac{\left| \langle \e, {f}_j - f^*_j \rangle_n\right|}{ \|f_j-f_j^*\|_n + \lambda \Om(f_j-f_j^*) } \ge {\rho} \right\} &\le pc_2\exp\left( - \frac{n\rho^2}{4c_2^2} \right)\\
	&= c_2\exp\left\{ - n\rho^2\left( \frac{1}{4c_2^2} - \frac{\log p}{n\rho^2} \right)\right\}\\
	&\le c_2\exp\left( - n\rho^2c_3 \right),
	\end{align*}
	for some positive constant $c_3 = c_3(c_2,\wt{A}_0)$.
	

	Finally, we show that $c_3> 0$. This follows from the fact that $ {1}/{(4c_2^2)} - {\log p }/{(n\rho^2)}> 0$ which is equivalent to $n\rho^2 > 4c_2^2\log p$. This holds since $n\rho^2\ge \kappa^2\log p$ for $\kappa$ sufficiently large. Thus, we have with probability at least $1-c_2\exp\left(c_3n\rho^2\right)$ for all $j= 1,\ldots,p$
	\begin{align*}
	|\langle\e,\wh{f}_j - f^*_j\rangle_n| \equiv |\langle\e,\hd_j\rangle_n| \le \rho \|\hd_j\|_n + \rho\lambda\up(\hd_j)\ .
	\end{align*}
\end{proof}
\subsection{Using the active set}
We continue using the short-hand notation $\hd_j = \wh{f}_j - f_j^*\ (j=1,\ldots, p)$. So far we have shown that, for $\lambda\ge 4\rho$, with probability at least $1-2\exp\left(-{n\rho}/{c_1}\right) -c_2\exp\left(-c_3n\rho^2\right)$, the following inequality holds
\begin{align*}
\|\wh{f} - f^0\|_n^2 + 2\lambda\sum_{j=1}^{p}I(\wh{f}_j) &\le 2|\langle \e, \wh{f}-f^* \rangle_n| + 2|\bar{\e}|\sum_{j=1}^p \|\hd_j\|_n + 2\lambda \sum_{j=1}^{p}I(f^*_j)+\|f^*-f^0\|_n^2 \\
&\le \left\{ \sum_{j=1}^{p} 2\rho\|\hd_j\|_n + 2\rho\lambda\up(\hd_j)\right\}  + \left( 2\rho \sum_{j=1}^{p}\|\hd_j\|_n \right) \\
&+ \left\{ 2\lambda \sum_{j=1}^{p}I(f^*_j)\right\}+\|f^*-f^0\|_n^2.
\end{align*}
Thus we have
\begin{align*}\|\wh{f}-f^0\|_n^2 + 2\lambda\sum_{j=1}^{p}I(\wh{f}_j)  &\le \sum_{j=1}^{p} \left\{ \lambda \|\hd_j\|_n + \frac{\lambda^2}{2} \up(\hd_j) + 2\lambda\|f^*_j\|_n+2\lambda^2\up(f^*_j)\right\} + \|f^* - f^0\|_n^2.
\end{align*}
For notational convenience we will exclude the $\|f^*-f^0\|_n^2$ term in the following manipulations. If $S$ is the active set then we have on the right hand side, denoted by \RHS,
\begin{align*}
{\RHS} &= \lambda \sum_{j\in S} \left\{\|\hd_j\|_n +  \frac{\lambda}{2}\up(\hd_j) +2\|f^*_j\|_n + 2\lambda\up(f^*_j) \right\}+\lambda\sum_{j\in S^c} \left\{\|\wh{f}_j\|_n+\frac{\lambda}{2}\up(\wh{f}_j)\right\}\\
&\le \lambda \sum_{j\in S} \left\{\|\hd_j\|_n +  \frac{\lambda}{2}\up(\hd_j) +2\|\hd_j\|_n+2\|\wh{f}_j\|_n + 2\lambda\up(f^*_j) \right\}+\lambda\sum_{j\in S^c} \left\{\|\wh{f}_j\|_n+\frac{\lambda}{2}\up(\wh{f}_j)\right\}\\
&= 3\sum_{j\in S}\lambda\|\hd_j\|_n+2\sum_{j\in S}\lambda^2\up(f^*_j)+ \sum_{j\in S^c}\lambda\|\wh{f}_j\|+\frac{1}{2}\sum_{j\in S^c}\lambda^2\up(\wh{f}_j) + 2\sum_{j\in S}\lambda\|\wh{f}_j\|_n + \frac{1}{2}\sum_{j\in S}\lambda^2\up(\hd_j),
\end{align*}
where the inequality holds by the decomposition $\|f^*_j\|_n = \|f^*_j-\wh{f}_j+\wh{f}_j\|_n\le \|\hd_j\|_n + \|\wh{f}_j\|_n$.

\noindent On the left hand side, denoted by \LHS, we have
\begin{align*}
\LHS &= \|\wh{f} - f^0\|_n^2+2\lambda\sum_{j\in S} \left\{\|\wh{f}_j\|_n+\lambda\up(\wh{f}_j)\right\} + 2\lambda\sum_{j\in S^c} \left\{\|\wh{f}_j\|_n+\lambda\up(\wh{f}_j)\right\}\\
&\ge \|\wh{f} - f^0\|_n^2+2\lambda\sum_{j\in S} \left\{\|\wh{f}_j\|_n +\lambda\up(\hd_j)-\lambda\up(f^*_j)\right\} + 2\lambda\sum_{j\in S^c} \left\{\|\wh{f}_j\|_n+\lambda\up(\wh{f}_j)\right\},
\end{align*}
where the inequality follows from the triangle inequality $  \up(\wh{f}_j) + \up(f^*_j)\ge \up(\hd_j)$ since $\up(\cdot)$ is a semi-norm. By re-arranging the terms we obtain the inequality
\begin{equation*}
\|\wh{f} - f^0\|_n^2 + \lambda\sum_{j\in S^c}\left\{ \|\wh{f}_j\|_n + \frac{3\lambda}{2}\up(\wh{f}_j) \right\} + \frac{3\lambda^2}{2}\sum_{j\in S}\up(\hd_j) \le 3 \lambda\sum_{j\in S} \|\hd_j\|_n + {4\lambda^2} \sum_{j\in S}\up(f^*_j) + \|f^*-f^0\|_n^2
\end{equation*}
which implies that 
\begin{equation*}
\label{eqn:inequalitySlowRates2}
\|\wh{f} - f^0\|_n^2 + \lambda\sum_{j\in S^c} \|\hd_j\|_n + \frac{3\lambda^2}{2}\sum_{j=1}^p\up(\hd_j) \le 3 \lambda\sum_{j\in S} \|\hd_j\|_n + {4\lambda^2} \sum_{j\in S}\up(f^*_j) + \|f^*-f^0\|_n^2.
\end{equation*}
This implies the slow rates for convergence for $\lambda \ge 4\rho$ and $s = |S|$
\begin{equation*}
\label{eqn:inequalitySlowRates3}
\frac{1}{2}\|\wh{f} - f^0\|_n^2 + \le s\lambda\left\{ 3R + 2\lambda\sum_{j\in S}\up(f^*_j)/s\right\}  + \frac{1}{2}\|f^*-f^0\|_n^2.
\end{equation*}

This completes the proof of Theorem~\ref{thm:thm1}. In the next section we prove the oracle inequality with fast rates via the compatibility condition.

\subsection{Using the compatibility condition}
\label{sec:compatibility}

Recall the compatibility condition for $f = \sum_{j=1}^{p}f_j$, whenever
\begin{equation}
\label{eqn:compatibility1}
\sum_{j\in S^c}\|f_j\|_n 
\le 4\sum_{j\in S} \|f_j\|_n,
\end{equation}
then we have 
\begin{equation*}
\sum_{j\in S}\|f_j\|_n \le s^{1/2}\|f\|_n/\phi(S),
\end{equation*}
were $s = |S|$.
Once we assume the compatibility condition we can prove Proposition~\ref{thm:additiveFastRates} by considering the following two cases.

\noindent \textbf{Case 1:} $\lambda\sum_{j\in S}\|\hd_j\|_n\ge 4\lambda^2\sum_{j\in S}\up(f^*_j)$ in which case we have
\begin{align*}
\|\wh{f} - f^0\|_n^2 + \lambda\sum_{j\in S^c} \|\hd_j\|_n + \frac{3\lambda^2}{2}\sum_{j=1}^p\up(\hd_j) &\le 4 \lambda \sum_{j\in S} \|\hd_j\|_n + \|f^* - f^0\|_n^2 \ ,
\end{align*}
hence for the function $\wh{f} - f^* = \sum_{j=1}^{p}{\hd}_j$ (\ref{eqn:compatibility1}) holds and hence by the compatibility condition we have
\begin{align*}
\|\wh{f} - f^0\|_n^2 &+ \lambda\sum_{j\in S^c} \|\hd_j\|_n + \frac{3\lambda^2}{2}\sum_{j=1}^p\up(\hd_j) \le \frac{4\lambda s^{1/2}}{\phi(S)}\|\wh{f} - f^*\|_n + \|f^* - f^0\|_n^2 \\
&\le \frac{4\lambda s^{1/2}}{\phi(S)}\|\wh{f} - f^0\|_n + \frac{4\lambda s^{1/2}}{\phi(S)}\|{f}^* - f^0\|_n + \|f^* - f^0\|_n^2\\
&\le 2\left\{ \frac{2\lambda (2s)^{1/2}}{\phi(S)} \right\}\left( \frac{\|\wh{f} - f^0\|_n}{2^{1/2}} \right) + 2\left\{\frac{2\lambda s^{1/2}}{\phi(S)} \right\} \left(\|{f}^* - f^0\|_n\right) + \|f^* - f^0\|_n^2\\
&\le  \frac{4\lambda^2 (2s) }{\phi^2(S)} + \frac{\|\wh{f} - f^0\|^2_n}{2}  + \frac{4\lambda^2 {s}}{\phi^2(S)} + \|{f}^* - f^0\|^2_n + \|f^* - f^0\|_n^2\\
&\le \frac{12\lambda^2 s }{\phi^2(S)}+\frac{\|\wh{f} - f^0\|^2_n}{2}  +2\|{f}^* - f^0\|^2_n,
\end{align*}
where we use the inequality $2ab\le a^2+b^2$ and this implies that 
\begin{equation*}
\frac{1}{2} \|\wh{f} - f^0\|_n^2 + \lambda\sum_{j\in S^c} \|\hd_j\|_n + \frac{3\lambda^2}{2}\sum_{j=1}^p\up(\hd_j) \le \frac{12s\lambda^2}{\phi^2(S)} + 2\|f^* - f^0\|_n^2.
\end{equation*}

\noindent \textbf{Case 2:} $\lambda\sum_{j\in S}\|\hd_j\|_n\le 4\lambda^2\sum_{j\in S}\up(f^*_j)$ in which case we have 
\begin{align*}
\|\wh{f} - f^0\|_n^2 + \lambda\sum_{j\in S^c} \|\hd_j\|_n + \frac{3\lambda^2}{2}\sum_{j=1}^p\up(\hd_j) &\le 16 \lambda^2\sum_{j\in S} \up(f^*_j) + \|f^*-f^0\|_n^2 \\
&\le 16s\lambda^2\sum_{j\in S}\up(f^*_j)/s + \|f^*-f^0\|_n^2,
\end{align*}
which implies
\begin{align*}
\frac{1}{2}\|\wh{f} - f^0\|_n^2 + \lambda\sum_{j\in S^c} \|\hd_j\|_n + \frac{3\lambda^2}{2}\sum_{j\in S}\up(\hd_j) 
&\le 16s\lambda^2\sum_{j\in S} \up(f^*_j)/s + 2\|f^*-f^0\|_n^2.
\end{align*}

\section{Constraining the proposed penalty region}
\label{app:BoundRegion}
Recall the following definitions
\begin{equation}
\mathcal{H}^w_{K_n} = \left\{ {\beta}\in \mathbb{R}^{K_n}: \sum_{k=1}^{K_n} w_k\|\beta_{k:{K_n}}\|_2\le 1 \right\},
\label{eqn:hbRegionApp}
\end{equation}
\begin{equation}
E^w_{K_n} = \left\{ {\beta}\in \mathbb{R}^{K_n}: \sum_{k=1}^{K_n} \beta_k^2\left(w_1+\cdots+w_k\right)^2\le 1 \right\}\ .
\label{eqn:ellipsoidRegionApp}
\end{equation}

\begin{lemma}
	\label{lemma:upBound}
	For the regions $\mathcal{H}_{K_n}^w$ and $E_{K_n}^w$ as defined in (\ref{eqn:hbRegionApp}) and (\ref{eqn:ellipsoidRegionApp}), respectively, we have $\mathcal{H}_{K_n}^w \subseteq E_{K_n}^w$ for all $n\ge 1$ and non-negative weights.
\end{lemma}

\begin{proof}
	It is sufficient to show $\sum_{k=1}^{K_n} \beta_k^2\left(w_1+\cdots+w_k\right)^2 \le \left( \sum_{k=1}^{K_n} w_k\|\beta_{k:{K_n}}\|_2 \right)^2$. We now have
	\begin{align*}
	\left( \sum_{k=1}^{K_n} w_k\|{\beta}_{k:{K_n}}\|_2 \right)^2 &= \sum_{m=1}^{K_n} w_m^2 \|\beta_{m:{K_n}}\|_2^2 + 2\sum_{m< k}w_kw_m \|\beta_{m:{K_n}}\|_2\|\beta_{k:{K_n}}\|_2\\
	&= \sum_{m=1}^{K_n} w_m^2\sum_{l=m}^{K_n}\beta_l^2 + 2\sum_{m<k} w_kw_m \|\beta_{k:{K_n}}\|_2^2\underbrace{ \frac{\|\beta_{m:{K_n}}\|}{\|\beta_{k:{K_n}}\| } }_{\ge 1} \\
	&\ge \sum_{l=1}^{K_n}\sum_{m=1}^{K_n} w_m^2\beta_l^2\bs{1}(l\ge m)  + 2\sum_{k=2}^{K_n}\sum_{m=1}^{k-1}w_kw_m\sum_{l=1}^{K_n} \beta_l^2\bs{1}(l\ge k) \\
	&= \sum_{l=1}^{K_n} \beta_l^2 \sum_{m=1}^{l} w_m^2  + 2\sum_{l=1}^{K_n} \beta_l^2\sum_{k=2}^{K_n}\sum_{m=1}^{k-1}w_kw_m \bs{1}(l\ge k) \\
	&= \sum_{l=1}^{K_n} \beta_l^2 \left( \sum_{m=1}^{l} w_m^2  + 2\sum_{k=2}^{l}\sum_{m=1}^{k-1}w_kw_m \right)  = \sum_{l=1}^{K_n}\beta_l^2 \left( \sum_{m=1}^{l}w_m \right)^2.
	\end{align*}
\end{proof}

\begin{lemma}
	\label{lemma:lowBound}
	For the region $\mathcal{H}_{K_n}^w$ as defined in (\ref{eqn:hbRegionApp}), we have the inclusion $\mathbb{B}_{K_n}^w \subseteq \mathcal{H}_{K_n}^w$ where 
	\begin{equation}
	\mathbb{B}_{K_n}^w = \left\{ {\beta} \in \mathbb{R}^{K_n}: \sum_{k=1}^{K_n}\beta_k^2 \le (w_1+\cdots+a_{K_n})^{-2} \right\}.
	\label{eqn:circle}
	\end{equation}
\end{lemma}
\begin{proof}
	Let ${\beta}\in \mathbb{B}_{K_n}^w$ and for brevity we denote $\|\cdot\| = \|\cdot\|_2$. Then for ${\beta}\in \mathbb{B}_{K_n}^w$ 
	\begin{align*}
	1 &\ge \|{\beta}\|\left( w_1+\cdots+w_{K_n} \right)\\
	&\ge \|{\beta}\|\left( w_1\frac{\|\beta_{1:K_n}\|}{\|\beta_{1:K_n}\|} + 
	\cdots + w_{K_n}\frac{\|\beta_{K_n:K_n}\|}{\|\beta_{1:K_n}\|} \right)^2\\
	&= w_1\|\beta_{1:K_n}\| 
	+\cdots+w_{K_n}\|\beta_{K_n:K_n}\|,
	\end{align*}
	which implies that ${\beta}\in \mathcal{H}_{K_n}^w$.
\end{proof}

In Figure~\ref{fig:inclusion}, we demonstrate the above two lemma's for $K_n=2$ for the special case of $w_k =  k^{m} - (k-1)^m$.

\begin{figure}
	\centering
	\includegraphics[width = \textwidth]{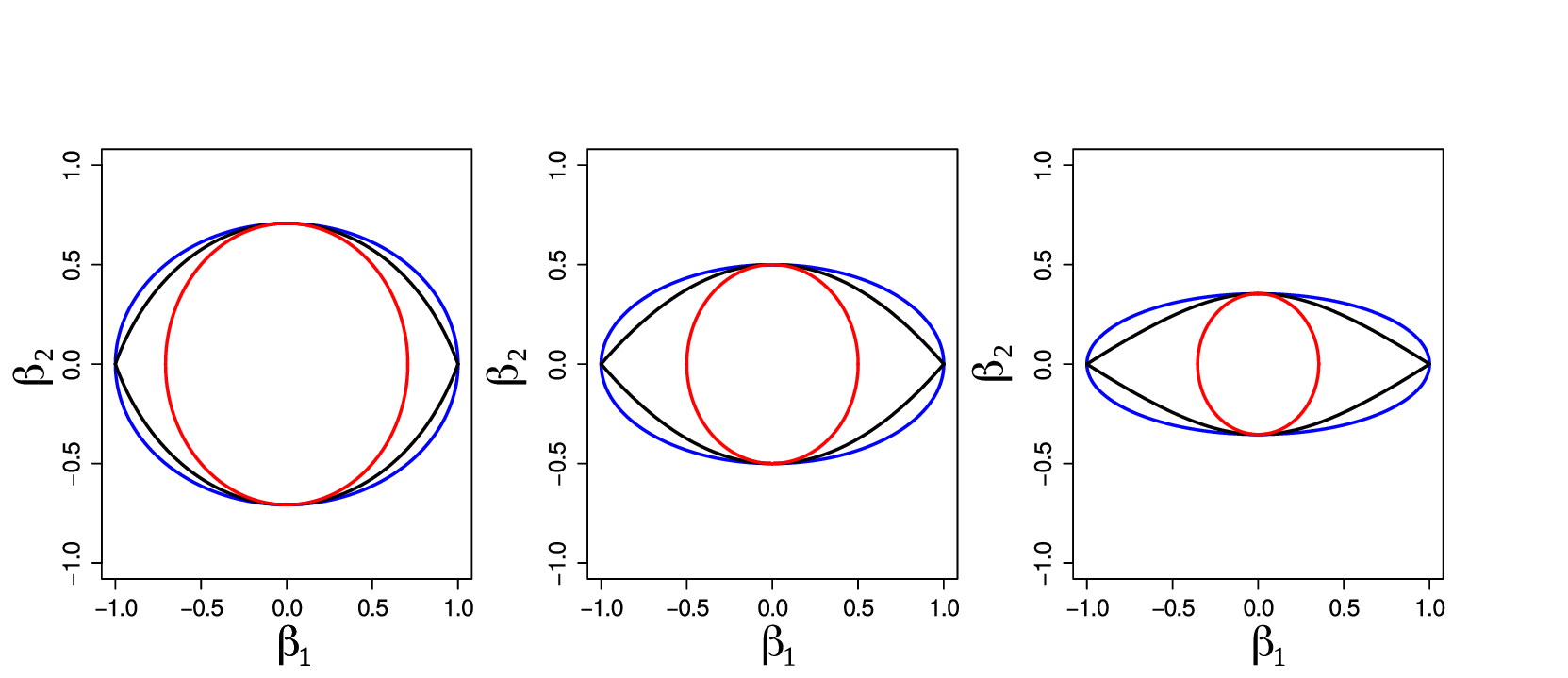}
	\caption{Demonstration of Lemma~\ref{app:BoundRegion}\ref{lemma:upBound} and Lemma~\ref{app:BoundRegion}\ref{lemma:lowBound} for the special case of $w_k = a_{j,\, m} = k^m - (k-1)^m$ and $K_n=2$. We show the region $E_{2}^w$~(\protect\includegraphics[height=0.5em]{NA_blue_1.eps}), $\mathbb{B}_2^w$~(\protect\includegraphics[height=0.5em]{NA_red_1.eps}) and $\mathcal{H}_2^w$~(\protect\includegraphics[height=0.5em]{NA_black_1.eps}). From left to right we have the plots for $m = 0.5,\, 1$ and $1.5$.}
	\label{fig:inclusion}
\end{figure}

We now present the proof of the claim made in \ref{sec:entropyResultsHB} of the manuscript regarding the relationship of our function class to weighted $L_p$ spaces. Recall the definition of our class $\mathcal{F}_{\infty}$:
\begin{equation*}
\mathcal{F}_{\infty} = \Big\{ f_{\beta}(x) = \sum_{k=1}^{\infty} \psi_k(x)\beta_k\, : \int \psi_k\psi_l\, dQ = 0 \text{ for } k\not= l, \int \psi_k^2\, dQ = 1 \Big\}\ .
\end{equation*}

\begin{lemma}
	For the hierarchical function class,
	\begin{equation*}
	\mathcal{F}_{\infty}^M = \{ f_{  {\beta}} \in \mathcal{F}_{\infty}: \sum_{k=1}^{\infty}\{k^m - (k-1)^m\}\|  {\beta}_{k:\infty}\|_2  \le M\},
	\end{equation*}
	and the weighted $L_q$ class 
	\begin{equation*}
	\mathcal{G}_q^M = \{ f_{  {\beta}} \in \mathcal{F}_\infty: \sum_{k=1}^{\infty}(k^m|\beta_k|)^q  \le M^q  \},
	\end{equation*}
	we have the following relationship:
	\begin{equation}
	\mathcal{G}_1^M \subseteq \mathcal{F}_{\infty}^M\subseteq \mathcal{G}_2^M.
	\end{equation}
\end{lemma}
\begin{proof}
	The inclusion $\mathcal{F}_{\infty}^M\subseteq \mathcal{G}_2^M$ follows from the proof of Lemma~F\ref{lemma:upBound} above. For the first inclusion it suffices to show that $\sum_{k=1}^{\infty} \{k^m - (k-1)^m\}\|\beta_{k:{\infty}}\|_2 \le \sum_{k=1}^{\infty} k^m|\beta_k|$. This follows from the fact that the $\ell_q$ norm is decreasing in $q$.
	\begin{align*}
	\sum_{k=1}^{\infty} \{k^m - (k-1)^m\}\|\beta_{k:{\infty}}\|_2 &\le \sum_{k=1}^{\infty} \{k^m - (k-1)^m\}\|\beta_{k:{\infty}}\|_1\\
	&= \sum_{k=1}^{\infty} \{k^m - (k-1)^m\}\sum_{j=k}^{\infty}|\beta_j|\\
	&= \sum_{k=1}^{\infty}\sum_{j=1}^{\infty} \{k^m - (k-1)^m\}|\beta_j|\bs{1}(j>k)\\
	&= \sum_{j=1}^{\infty}|\beta_j| \sum_{k=1}^{j}\{k^m - (k-1)^m\}\\
	&= \sum_{j=1}^{\infty}|\beta_j|j^m.
	\end{align*}
\end{proof}

\section{Some entropy results for ellipsoids}
\label{sec:entropyResults}

\subsection{An upper bound}
\label{sec:upperBoundEntropy}

In this section we establish some entropy results for the ellipsoid (\ref{eqn:ellipsoidRegionApp}) and the circle (\ref{eqn:circle}) which will allow us to establish entropy rates for the penalty region $\mathcal{H}_{K_n}^w$.

Since $K_n$ can potentially be $\infty$, or arbitrarily large, we need a way to handle this dimension. It turns out that this can be done using a simple argument which we demonstrate in the following theorem.

\begin{theorem}\citep{dumer2006covering}
	\label{thm:EllipsoidUp}
	For any $\theta\in (0,1/2)$, the $\delta$-entropy of the ellipsoid $E_{K_n}^w$ satisfies the following inequality
	\begin{equation*}
	H(\delta,\ E_{K_n}^w) \le \sum_{k=1}^{d-1}\log\left( \frac{1}{\delta\sum_{l=1}^{k}w_l} \right) + \mu_{\theta}\log(3/\theta)\ ,
	\end{equation*}
	where $\mu_{\theta}\le K_n$ is the largest integer such that $w_1+\cdots+w_{\mu_{\theta}}< \left\{(1-\theta)^{1/2}\delta\right\}^{-1}$ and $d\le K_n + 1$ is the largest integer such that $w_1+\cdots+w_{d-1}\le \delta^{-1}$.
	If $\delta^{-1}\le w_1$ then $H(\delta, E_{K_n}^w) = 0$ holds trivially.
\end{theorem}

\begin{corollary}[Sobolev Ellipsoids]
	\label{lemma:UpperBoundSobolev}
	For Theorem~\ref{thm:EllipsoidUp}, let $w_k = k^{m} - (k-1)^m$. Then we have the following upper bound:
	\begin{equation*}
	H(\delta,\, E_{K_n}^w) \le U_E\delta^{-1/m},
	\end{equation*}
	for some constant $U_E$ which only depends on $m$ and $\theta$.
\end{corollary}
\begin{proof}
	Firstly, we note that with this definition of $w_k$, we can let $K_n = \infty$. Thus if we can show that $H(\delta, E_{\infty}^w) \le U\delta^{-1/m}$ then the result follows since $E_{K_n}^w \subset E_{\infty}^w$ for all $K_n < \infty$.
	\newline 
	
	Now we have $w_1+\cdots+w_{\mu_{\theta}} = \mu_{\theta}^m$, hence
	\begin{align*}
	\mu_{\theta}^m < \frac{\delta^{-1}}{(1-\theta)^{1/2}}, \text{ and thus }   \mu_{\theta} \log(3/\theta) < \log(3/\theta) \left\{\frac{\delta^{-1}}{(1-\theta)^{1/2} }\right\}^{1/m} = U_1\delta^{-{1}/{m}}.
	\end{align*}
	Now for the second part we use the fact that $w_1+\cdots+w_{d-1} = (d-1)^m \le \delta^{-1} < d^{m}$ and we obtain
	\begin{align*}
	\sum_{j=1}^{d-1} \log\left( \frac{1}{\delta j^{m}} \right) &= (d-1)\log(\delta^{-1}) + \log\left[ \frac{1}{\{(d-1)!\}^m} \right]\le \delta^{-1/m}  \log(\delta^{-1}) -m\log\left\{ (d-1)! \right\} \\
	&\le \delta^{-1/m} \left[  \log(\delta^{-1}) -m \delta^{1/m}\log\left\{ (d-1)! \right\} \right] \le  \delta^{-1/m} \left[  \log(d^m) -m \delta^{1/m}\log\left\{ (d-1)! \right\} \right]\\
	&\le \delta^{-1/m}m\left[ \log(d) - d^{-1}\log\{(d-1)!\}  \right].
	\end{align*}
	Now by sterling's inequality we have for all $d\in \{1,2,\ldots\}$
	\begin{align*}
	\log(d+1) - \frac{\log(d!)}{d+1} &\le \log(d+1) - \frac{\log\left( 2^{1/2}\pi^{1/2}d^{d+1/2}e^{-d} \right)}{d+1}\\
	&= \log(d+1) +\frac{d}{d+1} - \frac{\log(2^{1/2}\pi^{1/2})}{d+1}- \frac{d+1/2}{d+1}\log d\\
	&\le \log(d+1) + 1 - \frac{d+1-1+1/2}{d+1}\log d \\
	&= 1+ \log\left( \frac{d+1}{d} \right) + (1/2)\frac{\log d}{d+1}\\
	&\le 1+ \log\left( 1+ \frac{1}{d} \right) + (1/2)\frac{\log d}{d+1}\le 1+\log 2 + 1.
	\end{align*}
	This implies that 
	\begin{align*}
	\sum_{j=1}^{d-1} \log\left( \frac{1}{\delta j^{m}} \right) \le \delta^{-1/m}m\left[ \log(d) - d^{-1}\log\{(d-1)!\}  \right] \le U_2\delta^{-1/m}.
	\end{align*}
\end{proof}

\begin{corollary}[Multivariate framework]
	\label{lemma:UpperBoundMultivariateHB}
	For Theorem~\ref{thm:EllipsoidUp}, let $w_{q_k} = k^m-(k-1)^m$ where for a fixed dimension $p$ we define $$q_k = \sum_{l=1}^{k}\binom{l+p-2}{l-1} = \binom{k+p-1}{p},$$ and all other $w_k= 0$. Then we have the following upper bound:
	\begin{equation*}
	H(\delta,\, E_{K_n}^w) \le U_E\delta^{-p/m},
	\end{equation*}
	for some constant $U_E$ which only depends on $m$ and $\theta$.
\end{corollary}

\begin{proof}
	Firstly, since $w_1=1$ the entropy is 0 for $\delta\ge 1$ and hence we will restrict ourselves to $\delta\in (0,1)$. We note that we must have $\mu_{\theta} = q_{k_1}-1 $ for some integer $k_1$. This is because all weights after $q_{k_1-1}$ are zero until $w_{q_k}$. Now we have by definition 
	\begin{align*}
	w_1+\cdots+w_{q_{k_1}-1}  = (k_1-1)^m \le \left\{\delta(1-\theta)^{1/2}\right\}^{-1},
	\end{align*}
	and we have 
	\begin{align*}
	\mu_{\theta} &= q_{k_1}-1 = \binom{k_1+p-1}{p}-1 < \frac{(k_1+p-1)^p}{p!}\\
	&\le \frac{\left[\left\{\delta (1-\theta)^{1/2} \right\}^{-1/m} +p \right]^p}{p!} = \frac{\delta^{-{p}/{m}} \left\{ \left( 1-\theta\right)^{-1/(2m)} + p\delta^{1/m} \right\}^p}{p!}\\
	&\le \delta^{-{p}/{m}}  \frac{\left\{ \left(1-\theta\right)^{-1/(2m)} + p \right\}^p}{p!},
	\end{align*}
	where the second line follows from the inequality 
	$$\binom{n}{k}\le n^k/k!.$$
	This implies that for $\delta\in (0,\, 1)$
	\begin{align*}
	\mu_{\theta}\log(3/\theta) \le U_1\delta^{-{p}/{m}}.
	\end{align*}
	
	Similarly, there is an integer $k_2$ such that $d - 1 = q_{k_2} - 1$. Which means that $(k_2-1)^m\le \delta^{-1}\le k_2^m$. For the other term we have 
	\begin{align*}
	&\sum_{k=1}^{d-1} \log\left( \frac{1}{\delta\sum_{l=1}^{k}w_l} \right) = (d-1)\left\{ \log(\delta^{-1}) - \frac{\sum_{k=1}^{d-1}\log(\sum_{l=1}^{k}w_l) }{d-1} \right\} \\
	&= (d-1)\left\{ \log(\delta^{-1}) - \frac{ (q_2-1)\log(1^m) +(q_3-q_2)\log(2^m)+\cdots+(q_{k_2}-1-q_{k_2-1})\log(q_{k_2-1}^m) }{d-1} \right\} \\
	&= (d-1)\left[ \log(\delta^{-1}) - \frac{ m \{f(k_2) - \log(q_{k_2-1})\} }{d-1} \right] ,
	\end{align*}
	where $f(k_2)= (q_2-1)\log(1) +(q_3-q_2)\log(2)+\cdots+(q_{k_2}-q_{k_2-1})\log(q_{k_2-1}) = \sum_{l=1}^{k_2}(q_{l} - q_{l-1})\log(q_{l-1})$. Hence we have
	\begin{align*}
	\sum_{j=1}^{d-1} \log\left( \frac{1}{\delta\sum_{l=1}^{k}a_l} \right) &\le (d-1)\left[ \log(k_2^m) - \frac{ m \{ f(k_2) - \log(q_{k_2-1})\} }{d-1} \right]\\
	&= m(d-1)\left\{ \log(k_2) - \frac{ f(k_2) - \log(q_{k_2-1}) }{q_{k_2}-1} \right\}.
	\end{align*}
	Now by induction we can show that $\frac{ f(k_2) - \log(q_{k_2-1}) }{q_{k_2} - 1} \ge \frac{\log\{(k_2-1)!\}}{k_2}$ which implies that 
	\begin{align*}
	m(d-1)\left\{ \log(k_2) - \frac{ f(k_2) - \log(q_{k_2-1}) }{q_{k_2}-1} \right\} &\le m(d-1)\left[ \log(k_2) - \frac{\log\{(k_2-1)!\}}{k_2} \right]\\
	&\le (d-1)m\left\{2+\log(2)\right\}.
	\end{align*}
	Finally, we note that 
	\begin{align*}
	d-1 &= q_{k_2}-1 = \binom{k_2+p-1}{p}-1 < \binom{k_2 + p-1}{p}\\
	&\le \frac{(k_2+p-1)^p}{p!}\le \frac{(\de^{-1/m}+p)^p}{p!} = \de^{-{p}/{m}} \frac{\left( 1 + p\de^{1/m} \right)^p}{p!}\le \de^{-{p}/m} \frac{\left( 1 + p\right)^p}{p!}.
	\end{align*}
	
\end{proof}

\section{Proof of Theorem~\ref{thm:upperRate}}
\label{app:proofDetails}

\begin{proof}
	
	By definition
	\begin{align*}
	\frac{1}{2}\|\wh{f}- {y}\|_n^2 + \lambda_n^2\Om(\wh{f}\mid Q_n) &\le  \frac{1}{2}\| f_{n}^* - {y}\|_n^2 + \lambda_n^2\Om(f_{n}^*\mid Q_n), 
	\end{align*}
	which leads to the following inequality
	\begin{align*}
	\frac{1}{2}\|\wh{f}-f^0\|_n^2 + \lambda_n^2\Om(\wh{f}\mid Q_n) &\le |\langle\e, \wh{f}-f_{n}^*\rangle_n|+ \frac{1}{2}\|f_n^*-f^0\|_n^2 + \lambda_n^2\Om(f_n^*\mid Q_n)\ ,
	\end{align*}
	where $\langle \e, f\rangle_n = {n}^{-1}\sum_{i=1}^{n}\e_if(x_i)$. Via the simple decomposition $\|\wh{f} - f_n^*\|_n^2 \le 2\|\wh{f} - f^0\|_n^2 + 2\|f^0-f_n^*\|_n^2$ we obtain 
	\begin{align*}
	\frac{1}{2}\|\wh{f}-f_n^*\|_n^2 + \lambda_n^2\Om(\wh{f}\mid Q_n) &\le \|\wh{f}-f^0\|_n^2 + \|f^0-f_n^*\|_n^2 + 2\lambda_n^2\Om(\wh{f}\mid Q_n)\\
	&= \|f^0-f_n^*\|_n^2 +2\left\{ \frac{1}{2}\|\wh{f}-f^0\|_n^2 +  \lambda_n^2\Om(\wh{f}\mid Q_n)\right\}\\
	&\le \|f^0-f_n^*\|_n^2 +2\left\{|\langle\e, \wh{f}-f_{n}^*\rangle_n| + \frac{1}{2}\|f_n^*-f^0\|_n^2 + \lambda_n^2\Om(f_n^*\mid Q_n) \right\}\\
	&= 2|\langle\e, \wh{f} - f_n^*\rangle_n| + 2\lambda_n^2\Om(f_n^*\mid Q_n) + \|f^0-f_n^*\|_n^2 \\
	&\le \max\left\{ 4|\langle\e, \wh{f} - f_n^*\rangle_n| + 4\lambda_n^2\Om(f_n^*\mid Q_n),\  2\|f^0-f_n^*\|_n^2\right\}.
	\end{align*} 
	Thus our basic inequality is given by 
	\begin{align*}
	\frac{1}{2}\|\wh{f}-f_n^*\|_n^2 + \lambda_n^2\Om(\wh{f}\mid Q_n) \le 2\max\left\{ 2|\langle\e, \wh{f} - f_n^*\rangle_n| + 2\lambda_n^2\Om(f_n^*\mid Q_n),\  \|f^0-f_n^*\|_n^2\right\}.
	\end{align*} 
	
	\noindent Hence from the basic inequality either $\frac{1}{2}\|\wh{f}-f_{n}^*\|_n^2 + \lambda_n^2\Om(\wh{f}\mid Q_n) \le 2\|f^0 - f_n^*\|_n^2\ $ which implies the result or 
	\begin{align*}
	\frac{1}{2}\|\wh{f}-f_{n}^*\|_n^2 + \lambda_n^2\Om(\wh{f}\mid Q_n) &\le 4\langle\e, \wh{f}-f_{n}^*\rangle_n+  4\lambda_n^2\Om(f_{n}^*\mid Q_n)\ .
	\label{eqn:mainCase1}
	\end{align*}
	
	\noindent Now note that $H(\de, \{ f\in \mathcal{F}_{n}:\Om(f\mid Q_n)\le 1\}, Q_n)\le A_1\de^{-\alpha}$ implies
	\begin{equation*}
	H\left[\de,\,\left\{  \frac{f-f_{n}^* }{\Om(f\mid Q_n)+\Om(f_{n}^*\mid Q_n)}: f\in \mathcal{F}_{n}\right\},\, Q_n \right]\le \wt{A}_1\de^{-\alpha}.
	\end{equation*}
	Thus we invoke Lemma 8$\cdot$4 of \cite{vandegeer2000empirical} and conclude that with probability at least $1-c\exp\{-(T/c)^2\}$ for a constant $c>0$ and  all $T\ge c$, we have
	\begin{equation*}
	|\langle\e, \wh{f} - f_{n}^*\rangle_n| \le Tn^{-1/2}\|\wh{f}-f_{n}^*\|_n^{1-{\alpha}/{2}}\left\{ \Om(\wh{f}\mid Q_n)+\Om(f_{n}^*\mid Q_n) \right\}^{{\alpha}/{2}}.
	\end{equation*}
	
	Define the set $\mathcal{T}$ as
	\begin{equation*}
	\mathcal{T} = \left\{\sup_{f\in \mathcal{F}_n}  \Big|\langle\e,\, {f}-f_{n}^*\rangle_n \Big| \le Tn^{-1/2}\|f - f_{n}^*\|_n^{1-{\alpha}/{2}}\left\{ \Om(f\mid Q_n)+\Om(f_{n}^*\mid Q_n) \right\}^{{\alpha}/{2}} \right\}\ ,
	\end{equation*}
	then on the set $\mathcal{T}$ we have
	\begin{align*}
	\frac{1}{2}\|\wh{f}-f_{n}^*\|_n^2 + \lambda_n^2\Om(\wh{f}\mid Q_n) \le 4Tn^{-1/2}\|\wh{f}-f_{n}^*\|_n^{1-{\alpha}/{2}} \left\{ \Om(\wh{f}\mid Q_n)+\Om(f_{n}^*\mid Q_n) \right\}^{{\alpha}/{2}}+ 4\lambda_n^2\Om(f_{n}^*\mid Q_n)\ .
	\end{align*}
	Which means we have either 
	\begin{align*}
	\frac{1}{2}\|\wh{f}-f_{n}^*\|_n^2 + \lambda_n^2\Om(\wh{f}\mid Q_n) \le  8\lambda_n^2\Om(f_{n}^*\mid Q_n) ,
	\end{align*}
	which is of the desired form or
	\begin{align}
	\frac{1}{2}\|\wh{f}-f_{n}^*\|_n^2 + \lambda_n^2\Om(\wh{f}\mid Q_n) \le 8Tn^{-1/2}\|\wh{f}-f_{n}^*\|_n^{1-\frac{\alpha}{2}} \left\{ \Om(\wh{f}\mid Q_n)+\Om(f_{n}^*\mid Q_n) \right\}^{\frac{\alpha}{2}}\ .
	\label{eqn:Case2}
	\end{align}
	We now consider $(\ref{eqn:Case2})$ only.
	
	\subsection{Case 1: $\Om(\wh{f}\mid Q_n)\ge \Om(f_{n}^*\mid Q_n)$}
	In this case we have
	\begin{equation*}
	\frac{1}{2}\|\wh{f}-f_{n}^*\|_n^2 + \lambda_n^2\Om(\wh{f}\mid Q_n) \le 8Tn^{-1/2}\|\wh{f}-f_{n}^*\|_n^{1-{\alpha}/{2}} \left\{ 2\Om(\wh{f}\mid Q_n) \right\}^{{\alpha}/{2}}.
	\label{eqn:Case1b}
	\end{equation*}
	which gives us the following equivalent inequalities
	\begin{align*}
	\lambda_n^2\Om(\wh{f}\mid Q_n) &\le 8Tn^{-1/2}\|\wh{f}-f_{n}^*\|_n^{1-{\alpha}/{2}} \left\{ 2\Om(\wh{f}\mid Q_n) \right\}^{{\alpha}/{2}}\\
	\left\{\Om(\wh{f}\mid Q_n)\right\}^{1-{\alpha}/{2}} &\le 2^{3+{\alpha}/{2}}Tn^{-{1}/{2}}\lambda_n^{-2} \|\wh{f} - f_n^*\|_n^{1-{\alpha}/{2}}\\
	\Om(\wh{f}\mid Q_n) &\le \left( 2^{3+{\alpha}/{2}}Tn^{-{1}/{2}}\lambda_n^{-2} \right)^{{2}/{(2-\alpha)}}\|\wh{f} - f_n^*\|_n.
	\end{align*}
	Plugging this into the right hand side of (\ref{eqn:Case1b}) and solving for $\|\wh{f} - f_{n}^*\|_n$ we obtain the following series of inequalities
	\begin{align*}
	\frac{1}{2}\|\wh{f}-f_{n}^*\|_n^2 &\le T 2^{3+{\alpha}/{2}} n^{-{1}/{2}} \|\wh{f} - f_n^*\|_n^{1-{\alpha}/{2}} \left( 2^{3+{\alpha}/{2}}Tn^{-{1}/{2}}\lambda_n^{-2} \right)^{\{{2}/{(2-\alpha)}\} \{{\alpha}/{2} \} }\|\wh{f} - f_n^*\|_n^{\alpha/2} \\ 
	\frac{1}{2}\|\wh{f}-f_{n}^*\|_n&\le T^{{2}/{(2-\alpha)} } 2^{ {(6+\alpha)}/{(2-\alpha)} }n^{-{1}/{(2-\alpha)}} \lambda_n^{-{2\alpha}/{(2-\alpha)}} \\ 
	\frac{1}{2}\|\wh{f} - f_n^*\|_n^2 &\le C_1n^{-{2}/{(2-\alpha)}} \lambda_n^{-{4\alpha}/{(2-\alpha)}} = C_1\lambda_n^2\Om_{}(f_n^*\mid Q_n) \ ,
	\end{align*}
	where $C_1 = T^{{4}/{(2-\alpha)}}2^{{(14+\alpha)}/{(2-\alpha)}}$ and recall the definition $\lambda_n^{-1} = n^{{1}/{(2+\alpha)}} \left\{ \Om_{}(f_n^*\mid Q_n) \right\}^{{(2-\alpha)}/{\{2(2+\alpha)\}}}.$

	\subsection{Case 2:  $\Om_{}(\wh{f}\mid Q_n)\le \Om_{}(f_{n}^*\mid Q_n)$}
	
	In this case we have
	\begin{equation*}
	\frac{1}{2}\|\wh{f}-f_{n}^*\|_n^2 + \lambda_n^2\Om_{}(\wh{f}\mid Q_n) \le 8Tn^{-1/2}\|\wh{f}-f_{n}^*\|_n^{1-{\alpha}/{2}} \left\{ 2\Om_{}(f_{n}^{*}\mid Q_n) \right\}^{{\alpha}/{2}}.
	\label{eqn:Case2b}
	\end{equation*}
	From which we directly get the following inequalities
	\begin{align*}
	\frac{1}{2}\|\wh{f}-f_{n}^*\|_n^2 &\le 8{T}n^{-{1}/{2}}\|\wh{f}-f_{n}^*\|_n^{1-{\alpha}/{2} } \left\{ 2\Om_{}(f_{n}^*\mid Q_n) \right\}^{{\alpha}/{2}} \\
	\frac{1}{2}\|\wh{f}-f_{n}^*\|_n^{1+{\alpha}/{2}} &\le 2^{3+{\alpha}/{2}}{T}n^{-{1}/{2}} \left\{ \Om_{}(f_{n}^*\mid Q_n) \right\}^{{\alpha}/{2}}\\ 
	\|\wh{f}-f_{n}^*\|_n &\le 2^{{(8+\alpha)}/{(2+\alpha)}} T^{{2}/{(2+\alpha)}} n^{-{1}/{(2+\alpha)}} \left\{ \Om_{}(f_{n}^*\mid Q_n) \right\}^{{\alpha}/{(2+\alpha)}}\\
	\frac{1}{2} \|\wh{f}-f_{n}^*\|_n^2 &\le C_2 n^{-{2}/{(2+\alpha)}} \left\{ \Om_{}(f_{n}^*\mid Q_n) \right\}^{{2\alpha}/{(2+\alpha)}} = C_2 \lambda_n^2\Om_{}(f_n^*\mid Q_n),
	\end{align*}
	where $C_2 =  T^{{4}/(2+\alpha)} 2^{{(14+\alpha)}/{(2+\alpha)}}$. 
	Thus we have shown that on the set $\mathcal{T}$ we have
	\begin{equation*}
	\frac{1}{2}\|\wh{f} - f_{n}^*\|_n^2 \le \max(8,C_1,C_2)\lambda_n^2\Om_{}(f_n^*\mid Q_n) = C_0\lambda_n^2\Om_{}(f_n^*\mid Q_n).
	\end{equation*}
	
	We have shown that with probability at least $1-c\exp\left\{-({T}/{c})^2\right\}$ we have the inequality 
	\begin{equation*}
	\frac{1}{2}\|\wh{f} - f_n^*\|_n^2 \le \max \left\{2\|f^0-f_n^*\|_n^2,\ C_0\lambda_n^2\Om_{}(f_n^*\mid Q_n)  \right\}
	\end{equation*}
	
	To complete the proof we note that 
	\begin{align*}
	\frac{1}{2}\|\wh{f} - f^0\|_n^2 &\le \|\wh{f} - f_ n^*\|_n^2 + \|f_n^*-f^0\|_n^2\\
	&\le   2\max \left\{2\|f^0-f_n^*\|_n^2,\ C_0\lambda_n^2\Om_{}(f_n^*\mid Q_n)  \right\} + \|f_n^*-f^0\|_n^2\\
	&\le \frac{5}{2}\max \left\{2\|f^0-f_n^*\|_n^2,\ C_0\lambda_n^2\Om_{}(f_n^*\mid Q_n)  \right\} .
	\end{align*}
\end{proof}

\end{document}